%% file: Quotients_Free.tex
\documentclass[11pt]{article}
\usepackage[dvips]{graphicx}
\usepackage{float}
\usepackage{amsfonts}

\usepackage{bm}
\usepackage{latexsym}
\usepackage{amssymb}
\usepackage{amsmath}
\usepackage{amsthm}
\usepackage{hyperref}
\usepackage{verbatim}
\usepackage[dvips]{graphics}
\usepackage{times}

\usepackage{amsmath}
\usepackage{amsfonts}
\usepackage{amssymb}

\usepackage{graphicx}
\usepackage{epic}
\usepackage{eepic}
\usepackage{epsfig,float}
\usepackage{color}
\usepackage{tabularx}

\usepackage{verbatim}

\newtheorem{theorem}{Theorem}
\newtheorem{lemma}{Lemma}

\newtheorem{proposition}{Proposition}

\renewcommand {\emptyset}{\varnothing}
\newcommand {\eps}{\varepsilon}

\let\oldmarginpar\marginpar
\renewcommand\marginpar[1]{\-\oldmarginpar[\raggedright\footnotesize #1]%
{\raggedleft\footnotesize #1}}

\newcommand{\rhoR}{R}
\newcommand{\ol}{\overline}
\newcommand{\emp}{\emptyset}
\newcommand{\Sig}{\Sigma}

\newcommand{\bi}{\begin{itemize}}
\newcommand{\ei}{\end{itemize}}
\newcommand{\be}{\begin{enumerate}}
\newcommand{\ee}{\end{enumerate}}
\newcommand{\bd}{\begin{description}}
\newcommand{\ed}{\end{description}}
\newcommand{\bq}{\begin{quote}}
\newcommand{\eq}{\end{quote}}
\newcommand{\txt}[1]{\mbox{ #1 }}

\newcommand{\cC}{{\mathcal C}}
\newcommand{\cD}{{\mathcal D}}
\newcommand{\cL}{{\mathcal L}}
\newcommand{\cN}{{\mathcal N}}

\begin{document}

\title{Quotient Complexity of  Bifix-, Factor-, and Subword-Free Regular Languages
\thanks{This work was supported by the Natural Sciences and Engineering Research Council of Canada under grant no.\ OGP0000871 and by 
the Slovak Research and Development Agency under contract
APVV-0035-10 ``Algorithms, Automata, and Discrete Data Structures''.
}}

 \author{Janusz Brzozowski$^1$
   \quad Galina Jir\'askov\'a$^2$
   \quad Baiyu Li$^1$
   \quad Joshua Smith$^1$}

\date{}

\maketitle

\begin{center}
\vskip-20pt
 {\normalsize $^1$ David R. Cheriton School of Computer Science, University of Waterloo,}\\
 {\normalsize Waterloo, ON, Canada N2L 3G1 }\\
 {\normalsize {\tt \{brzozo@,b5li@,j45smith@student.math.\}uwaterloo.ca }}

 {\normalsize$^2$  Mathematical Institute, Slovak Academy of Sciences,}\\
 {\normalsize Gre\v{s}\'akova 6, 040 01 Ko\v{s}ice, Slovakia}\\
 {\normalsize {\tt jiraskov@saske.sk}}

\end{center}

\begin{abstract}
\label{***abst}
A language $L$ is prefix-free if,
whenever  words $u$ and $v$ are in $L$ and $u$ is a prefix of $v$, then $u=v$.
Suffix-, 
factor-, and subword-free languages 
are defined similarly, where ``subword"  means ``subsequence". 
A language is bifix-free if it is both prefix- and suffix-free.
We study the quotient complexity, more commonly known as state complexity, 
of  operations in the classes of  bifix-, factor-, and subword-free regular languages. 
We find tight upper bounds on the quotient 
complexity
of  intersection, union, difference, symmetric difference, 
concatenation, star,  and reversal 
in these three classes of languages.

\end{abstract}

\pagestyle{plain}
\pagenumbering{arabic}

\section{Introduction}
\label{***intro}
The state complexity of a regular language $L$ 
is the number of states  in the minimal deterministic finite automaton (dfa) 
accepting $L$~\cite{Yu01}.
This complexity is the same as the quotient complexity~\cite{Brz10} of $L$, 
which is the number of distinct left quotients of $L$. 
We prefer quotient complexity since it is more closely related to properties of languages.
The quotient complexity of an operation in a class $\cC$ of regular languages 
is the worst-case quotient complexity of the language resulting  from the operation, 
taken as a function of the quotient complexities of the operands in class  $\cC$. 
For surveys on state and quotient complexity see~\cite{Brz10,Yu01}.

One of the first results concerning the state complexity of an operation 
is the 1966 theorem by Mirkin~\cite{Mir66}, 
who showed that the bound  $2^n$ for the reversal of an $n$-state dfa can be attained.
In 1970 Maslov~\cite{Mas70} stated without proof the bounds 
on the complexities of union, concatenation,  star, 
and several other operations in the class of regular languages, 
and gave languages meeting these bounds. 
In 1994 these operations, along with intersection, reversal, and left and right quotients, 
were studied in detail by Yu, Zhuang and Salomaa~\cite{YZS94}.

Results exist also for proper subclasses of the class of regular languages: unary~\cite{PiSh02,YZS94}, 
finite~\cite{CCSY01,HaSa08,Yu01}, cofinite~\cite{BGN10}, 
prefix-free~\cite{HSW09,JiKr10}, suffix-free~\cite{Cmo11,HaSa09,JiOl09}, 
ideal~\cite{BJL10}, and closed~\cite{BJZ10}.
The bounds can vary considerably.

Free languages  (with the exception of $\{\eps\}$, where $\eps$ is the empty word) 
are  codes, which constitute an important class 
of languages and have applications in such areas 
as cryptography, data compression, and information transmission. 
They have been studied extensively; see, for example,~\cite{BPR10,JuKo97}.
In particular, \emph{prefix} and \emph{suffix codes}~\cite{BPR10} 
are prefix-free and suffix-free languages, respectively, 
\emph{infix codes}~\cite{Shy01,ShTh74} are factor-free, 
and \emph{hypercodes}~\cite{Shy01,ShTh74} are subword-free, 
where by subword we mean subsequence.
Moreover, free languages are special cases of convex languages~\cite{AnBr09,Thi73}.
We are interested only in regular free languages.

The state complexities of intersection, union, concatenation,  star, and reversal  
were first studied by Han, K.~Salomaa, and Wood~\cite{HSW09} for prefix-free languages, 
and  by Han and K.~Salomaa~\cite{HaSa09} for suffix-free languages.
In the present paper, these results are extended to bifix-, factor- and subword-free languages. 
In particular, we obtain tight upper bounds on the complexities 
of intersection, union, difference, symmetric difference, star, concatenation, and reversal 
in these three classes of free languages. 

\section{Preliminaries}
\label{***prelim}

It is assumed that the reader is familiar with finite automata 
and regular languages as treated in~\cite{Per90,Yu97}, for example.
If $\Sig$ is a finite non-empty alphabet, 
then $\Sig^*$ is the set of all  words over this alphabet,
with $\eps$ as the empty word. 
For $w\in\Sig^*$,  let $|w|$ be the length of $w$.
A~language is any subset of $\Sig^*$.

The following set operations are defined on languages:  
\emph{complement} ($\ol{L}=\Sig^*\setminus L$),  
\emph{union}  ($K\cup L$),  
\emph{intersection} ($K\cap L$),  
\emph{difference} ($K\setminus L$), and 
\emph{symmetric difference} ($K\oplus L$). 
A general \emph{boolean operation} with two arguments is denoted by $K\circ L$. 

We also  define the 
\emph{product,} usually called  \emph{concatenation} or \emph{catenation,}
($KL=\{w\in \Sig^*\mid w=uv, u\in K, v\in L\}$),
(Kleene) \emph{star} ($L^*=\bigcup_{i\ge 0}L^i$ with $L^0=\{\eps\}$), and 
\emph{positive closure} ($L^+=\bigcup_{i\ge 1}L^i$).

The reverse $w^R$ of a word $w\in\Sig^*$ is defined inductively as follows: 
$\eps^R=\eps$, and $(wa)^R=aw^R$ for every symbol $a$ in $\Sigma$
and every word $w$ in $\Sigma^*$.
The \emph{reverse} of a language $L$ is denoted by $L^\rhoR$ and 
is defined as $L^R=\{w^R\mid w\in L\}$.

\emph{Regular languages} over $\Sig$ 
are languages that can be obtained from the \emph{set of basic languages} 
$\{\emp,\{\eps\}\} \cup \{\{a\}\mid  a\in \Sig\}$,  
using a finite number of operations of union, product,  and  star. 
We use regular expressions to represent languages.
If $E$ is a regular expression, then $\cL(E)$ is the language denoted by that expression. 
For example, the regular expression  $E=(\eps\cup a)^* b$ denotes language $L=\cL(E)=(\{\eps\}\cup\{a\})^*\{b\}$. 
We usually do not distinguish notationally between regular languages and regular expressions.

Whenever convenient, we derive upper bounds on the state complexity 
of operations on free languages following the approach of~\cite{Brz10}.
A \emph{quotient} of a language $L$ by a word $w$ 
is defined as $L_w=\{x\in \Sig^* \mid wx \in L \}$. 
The number of distinct quotients of a language  
is called its \emph{quotient complexity}  and is denoted by $\kappa(L)$. 

Quotients of regular languages~\cite{Brz64,Brz10} 
can be computed as follows:
First, the \emph{$\eps$-function}  $L^\eps$ of a regular language $L$ is
 $L^\eps=\emptyset$ if $\eps\not\in L$, and
 $L^\eps=\eps$ if $\eps\in L$.
The \emph{quotient by a letter} $a$ in $\Sigma$ 
is computed by  induction:
 $b_a=\emptyset$ if $b\in \{\emptyset,\eps\}$ or $b\in\Sigma$ and $b\not= a$,
 and
 $b_a=\eps$ if $b=a$;
 $(\overline{L})_a =\overline{L_a};\,
 (K\circ L)_a = K_a\circ L_a; \,
 (KL)_a = K_aL \cup K^\eps L_a;\,
 (L^*)_a = L_aL^*$.
The quotient by a word $w$ in $\Sigma^*$  is computed by
induction on the length of $w$:
 $L_\eps =  L$
 and $L_{wa} = (L_w)_a$.
A~quotient $L_w$ is \emph{accepting} if $\eps\in L_w$; 
otherwise it is \emph{rejecting}.

A \emph{deterministic finite automaton} (dfa) 
is a quintuple $\cD=(Q,\Sigma,\delta,q_0,F)$, where 
$Q$ is a finite set of \emph{states,} 
$\Sigma$ is a finite \emph{alphabet,} 
$\delta:Q\times\Sigma\rightarrow Q$ is the \emph{transition function,} 
$q_0$ is the \emph{initial state}, and 
$F\subseteq Q$ is the set of \emph{final} or \emph{accepting states\/}. 
As usual, the transition function is extended to  $Q\times\Sigma^*$. 
The dfa $\cD$ accepts a word $w$ in $\Sigma^*$ if ${\delta}(q_0,w)\in F$. 
The set of all words accepted by $\cD$ is $L(\cD)$. 
By the \emph{language of a state} $q$ of $\cD$ 
we mean the language $L_q$ accepted 
by the automaton $(Q,\Sigma,\delta,q,F)$. 
A state is \emph{empty} if its language is empty.

The \emph{quotient automaton} of a regular language $L$ is the
dfa $\cD=(Q, \Sig, \delta, q_0,F)$, where 
$Q=\{L_w\mid w\in\Sig^*\}$, 
$\delta(L_w,a)=L_{wa}$, 
$q_0=L_\eps$,  
$F=\{L_w\mid \eps\in L_w\}$.
This  is the minimal dfa accepting $L$. 
Hence the quotient complexity of $L$ is equal to the state complexity of $L$, 
and we call it simply  \emph{complexity}.

\section{Free Languages}
\label{***freeprops}

If $u,v,w,x\in \Sig^*$ and $w=uxv$, then 
$u$ is a \emph{prefix} of $w$, 
$x$ is a \emph{factor} of $w$, and 
$v$ is a \emph{suffix} of $w$. 
Both $u$ and $v$ are also factors of $w$. 
If $w=u_0v_1u_1\cdots v_nu_n$, where $u_i,v_i \in \Sig^*$, then 
$v=v_1v_2\cdots v_n$ is a \emph{subword} of $w$. 
Every factor of $w$ is also a subword of $w$.

A language $L$ is \emph{prefix-free} 
(respectively, \emph{suffix-, factor-,} or \emph{subword-free}) 
if, whenever  words $u$ and $v$ are in $L$ and $u$ is a prefix 
(respectively, suffix, factor, or subword) of $v$, then $u=v$.  
Additionally, $L$ is \emph{bifix-free} if it is both prefix and suffix-free. 
All subword-free languages are factor-free, 
and all factor-free languages are bifix-free. 
For convenience, we refer to prefix-, suffix-, bifix-, factor-, and subword-free languages together as 
\emph{free} languages.

If $\eps$ is a quotient of $L$,
then $L$ also has the empty quotient,
since $\eps_a=\emp$, for all $a$ in $\Sig$.
We say that a quotient $L_w$ is \emph{uniquely reachable} 
if $L_w=L_x$ implies that $w=x$.
We now restate two propositions from~\cite{HaSa09,HSW09}
in our terminology. 

\begin{proposition}\label{prop:pf}
A non-empty  language is prefix-free if and only if it
 has exactly one accepting quotient and that quotient is $\eps$.
\end{proposition}

\begin{proposition}\label{prop:sf}
 The quotient by $\eps$  of a non-empty suffix-free language
 is uniquely reachable, and the language has the empty quotient.
\end{proposition}

Let $L$ be any language. 
If $(L_u)_x =L_v$ for some words $u,v$ and a non-empty word $x$,
then $L_v$ is \emph{positively reachable} from $L_u$, 
and we denote this by $L_u \to~L_v$.
The relation $\to$ is transitive.
The next proposition uses this relation to characterize finite languages. 
 
\begin{proposition}\label{prop:finite}
 If $L$ is any  language
 with the set of quotients 
 $\{L_1,L_2,\ldots,L_n\}$, and $u,v\in\Sig^*$,
 then the following are equivalent:\\
 1. $L$ is finite.\\
 2. $L_u\to L_v$ and $L_v\to L_u$ if and only if $L_u=L_v=\emp$.\\
 3. There exists a total order $\preceq$ on the set of quotients:

       $L=L_1 \preceq L_2\preceq \cdots \preceq L_{n-1} \preceq L_n=\emp$\\
       which satisfies the condition that $(L_i)_a=L_j$
       implies \mbox{$L_i\prec L_j$ or $L_i=L_j=L_n$.} 
\end{proposition}

\begin{proof}
 Suppose $L$ is a finite language.
 If $L_u\to L_v$ and $L_v\to L_u$,
 then $(L_u)_x=L_v$ and $(L_v)_y=L_u$,
 for some words $x$ and $y$.
 If also $L_u\neq \emp$,
 then $u(xy)^kw\in L$ for every nonnegative $k$ and any word $w$ in $L_u$,
 which contradicts that $L$ is finite. 
 Note also that  $L_u\neq \emp$ if $L_v\neq \emp$.
 If $L_u=L_v=\emp$, then $(L_u)_a=L_u$ for every $a$ in $\Sig$,
 and we have $L_u\to L_u$.
 Thus (1) implies (2).

 Now suppose  that $L$ is infinite and $\kappa(L)=n$.
 Then there is a word $uxv$ \mbox{in $L$}  of length at least $n$
 such that $L_u=L_{ux}$
 and $x\in\Sig^+$.
 Thus $L_u\to L_u$ and $L_u\neq \emp$,
 showing that (2) cannot hold.
 Hence (2) implies (1).

 If (1) holds, we can take the reflexive  closure $\to'$ of the relation $\to$.
 Then the relation $\to'$ is a partial order,
 and we can use any total order $\preceq$
 consistent with relation $\to'$ to get~(3). 
 Conversely, if (3) holds, then $L$ cannot be infinite,
 by the same argument as was used to prove that (2) implies (1).
\end{proof}

Since every subword-free language is finite,
we get the next lemma, 
which we use later to prove that upper bounds on the quotient complexity 
of some operations on subword-free languages 
cannot be reached if the alphabet of the language 
does not have sufficiently many letters.

\begin{lemma}\label{lem:letter}
 Let $L$ be a subword-free language with $\kappa(L)=n$, where $n\ge4$.
 Let the distinct quotients 
 $L=L_\eps=L_1,L_2,\dots,L_{n-2},L_{n-1}=\eps, L_n=\emptyset$ of $L$ 
 be ordered as in Proposition \ref{prop:finite}. 
 If $L_w=L_2$ for some word $w$, then $|w|=1$.
\end{lemma}

\begin{proof}
 Since $n\ge4$, the quotients $L$ and  $L_2$ are not empty.
 Let $v$ be a word in~$L_2$.
 If $L_w=L_2$, then $w$ cannot be $\eps$ because $L_2\neq L_1$.
 If $|w|>1$, 
 then $w=ua$ for a letter $a $ and a non-empty word $u$.
 Then $L_u\neq L$ since $L$ is uniquely reachable.
 If~$L_u=L_2$, then $uv\in L$ and $uav\in L$,  and language $L$ is not subword-free.
 Thus, if $L_u=L_i$, for some~$i$, 
 then $i> 2$, $L_{w}=(L_u)_a=(L_i)_a=L_j$ where $j\ge i > 2$, 
 contradicting that $L_{w}= L_2$. 
 Thus $w$ must be a one-letter word.
\end{proof}

Finally, we describe a simple method of constructing free languages.

\begin{proposition}\label{prop:construct}
 Let $L \subseteq \Sig^*$  be any language,
 and let $a \notin \Sig$.  Then
 (1)  $aL$ is suffix-free,
 (2) $La$ is prefix-free,
 (3) $aLa$ is factor-free.
\end{proposition}

\begin{proof}

 (1) Every proper suffix of a word in $aL$
 is a word over the alphabet $\Sigma$,
 and so is not in $aL$.
 Therefore $aL$ is suffix-free.\\
 (2) The proof is dual to that of (1).\\
 (3) Every proper factor of a word in $aLa$
 contains at most one $a$
  and therefore is not in $aLa$.
\end{proof}

\section{Boolean Operations}
\label{***bool}

The complexity of boolean operations,
in the class of prefix- and suffix-free regular languages,
except for the difference and symmetric difference of suffix-free languages,
was studied in \cite{HaSa09,HSW09,JiKr10,JiOl09}.
It was shown that for prefix-free languages,
the tight bounds for union, intersection, difference, and symmetric difference
are $mn-2$, $mn-2(m+n-3)$, $mn-(m+2n-4)$, and $mn-2$, respectively.
For  union and  intersection of suffix-free languages, the tight bounds are
$mn-(m+n-2)$ and 
$mn-2(m+n-3)$, respectively.
 The bounds for difference and symmetric difference 
are  $mn-(m+2n-4)$ and
 $mn-(m+n-2)$, respectively, and the bounds for all four boolean operations
 are met by binary suffix-free languages~\cite{Cmo11}.
 The next two theorems provide results
for boolean operations on bifix-, \mbox{factor-,} and subword-free languages.

\begin{theorem}[\bf Boolean Operations: Bifix- and Factor-Free Languages]\label{thm:bifix,factor}
 Let $K$ and $L$  be bifix- or factor-free languages
 over an alphabet~$\Sigma$
 with $\kappa(K)=m$ and $\kappa(L)=n$, where $m,n\ge4$. Then\\
   \hglue10pt  1. $\kappa(K\cap L)\le mn-3(m+n-4)$;\\ 
   \hglue10pt  2. $\kappa(K\setminus L)\le mn-(2m+3n-9)$;\\ 
   \hglue10pt  3. $\kappa(K\cup L), \kappa(K\oplus L)\le mn-(m+n)$.\\ 
 All the bounds are tight if $|\Sigma|\ge3$.
\end{theorem}

\begin{proof} 
 Since $K$ and $L$ are bifix-free,
 by unique reachability we get a reduction of $m+n-2$ from the general bound $mn$.
 Moreover, both languages $K$ and $L$ have $\eps$ and $\emptyset$  as quotients.
 For intersection,
 we have $\emptyset \cap L_w=K_w\cap \emptyset = \emptyset$,
 and the quotients $\eps \cap L_w$ and $K_w\cap \eps$
 are either empty or equal to $\eps$.
 This gives the upper bound. 
 For difference,
 we eliminate $m+n-2$ quotients by unique reachability,
 $n-2$ quotients by the fact that $\emp\setminus L_w=\emp$
 (keeping only one representative $\emp\setminus\emp$),
 $m-2$ quotients by the fact that $K_w\setminus \emp= K_w\setminus\eps$
 (keeping $K_w\setminus\emp$ as a representative),
 and $n-3$ more quotients by the rule $\eps\setminus L_w=\eps$,
 for a total reduction of $(2m+3n-9)$.
  For union,
 we have the unique reachability reduction of $m+n-2$,
 and a further reduction of 2
 by the rule
 $\eps\cup\eps=\eps\cup\emptyset=\emptyset \cup\eps=\eps$.
  For symmetric difference,
 we note that 
 $\eps \oplus \eps=\emptyset \oplus \emptyset=\emptyset$
 and $\eps \oplus \emptyset=\emptyset \oplus \eps =\eps$.

 For tightness, 
 consider  $K=a(c^*(a\cup b))^{m-3}$, 
 $L=a(b^*(a\cup c))^{n-3}$; see
 Figure~\ref{fig:factor}.
 If $w\in K$, then $w=av$ for some word $v$ containing $m-3$ 
 occurences of symbols from $\{a,b\}$ and ending in $a$ or $b$.
 This means that no proper factor of $w$ \nopagebreak is in $K$,
 and so $K$ is factor-free.
 A similar proof applies to $L$.
 
 \begin{figure}[b]
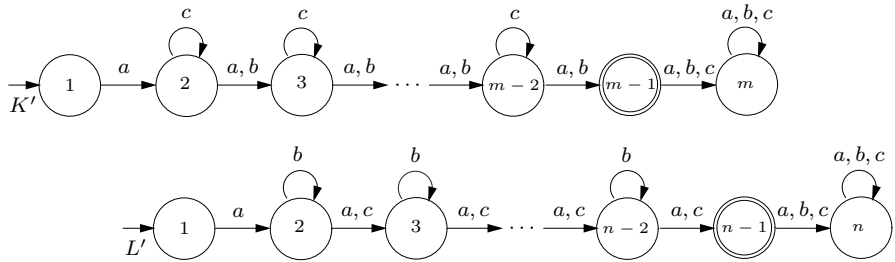

 \begin{center}
 \input factor_new.eepic
 \end{center}
 \caption{Factor-free languages meeting the upper bounds for boolean operations.}
 \label{fig:factor}
 \end{figure}

\label{factor}
\begin{figure}[h]
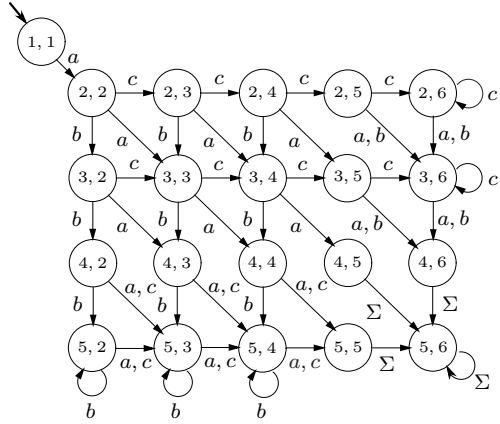

\begin{center}
\input crossprod.eepic
\end{center}
\caption{Cross product automaton for boolean operations
         on factor-free languages from Figure~\ref{fig:factor};
         $m=5, n=6$.}
\label{fig:cross}
\end{figure}

 In the cross-product automaton of Figure~\ref{fig:cross}
for the boolean operations on languages $K$ and $L$,
 all the states are reached from the initial state $(1,1)$
 by a word in $ab^*c^*\cup ac^*b^*$, except for  state $(m-1,n-1)$
 which is reached from state $(m-2,n-2)$ by $a$.

 For intersection,
 the only accepting state  
 is $(m-1,n-1)$.
 All the rejecting states in rows $m-1$ and $m$
 and columns $n-1$ and $n$
are empty.
 The word $a$ is accepted only from 
 $(m-2,n-2)$,
  word $b^{m-2-i}c^{n-2-j}a$  ($2\le i\le m-2$, $2\le j\le n-2$)
 only from state
 $(i,j)$,
 and the word $ab^{n-4}c^{n-4}a$ only from state (1,1).
 This  gives $mn-3(m+n-4)$ reachable and pairwise distinguishable states.
  
 For difference,
 all the states of the cross-product automaton in row $m-1$, except for $(m-1,n-1)$,
 are accepting and accept $\eps$.
 All the states in row $m$,
as well as state $(m-1,n-1)$ are empty.
 Moreover, states $(i,n-1)$ and $(i,n)$ with $2\le i \le m-2$ are equivalent.
 The word $ab^{m-3}$ is accepted only from $(1,1)$.
 Now let
 $(i,j)$ and $(k,\ell)$, 
 where $2\le i \le n-1, 2\le j \le m-2$,
 be two distinct states. 
 If $i<k$, then $c^nb^{m-1-i}$
 is accepted from $(i,j)$ but not from $(k,\ell)$.
 If $i=k$ and $j<\ell$, then $b^{m-2-i}c^{n-2-j}a$
 is not accepted from $(i,j)$
 but is accepted from $(k,\ell)$.
 This means that $mn-(2m+3n-9)$ states  are pairwise distinguishable.

 For union,
 all the states in row $m-1$ and in column $n-1$ are accepting,
 and moreover, the three states $(m,n-1)$, $(m-1,n-1)$, and $(m-1,n)$
 are equivalent.
 The word $ab^{m-3}$ is accepted only from $(1,1)$.
 Consider two distinct rejecting states $(i,j)$ and $(k,\ell)$.
 If $i<k$, then $c^nb^{m-1-i}$ is accepted from $(i,j)$ but not from $(k,\ell)$.
 If $j<\ell$, then $b^mc^{n-1-j}$ is accepted from $(i,j)$ but not from $(k,\ell)$.
 Now consider two distinct accepting states different from $(m,n-1)$ and $(m-1,n)$.
 By $c$, the two states either  go to two states 
 one of which is accepting and the other rejecting, 
 or to two distinct rejecting, and hence distinguishable, states.
 This proves distinguishability of $mn-(m+n)$ states.
 
 The proof for symmetric difference is the same as for union,
 except that  state $(m-1,n-1)$ is empty and states $(m,n-1)$ and $(m-1,n)$ are equivalent.
\end{proof}
\newpage

The next result shows that the upper bounds 
for intersection and difference of factor-free languages 
are also tight in the binary case.

\begin{proposition} [\bf Intersection and Difference: Binary Factor-Free Languages]
\label{prop5}
 There exist binary factor-free languages $K$ and $L$
 with $\kappa(K) = m$ and $\kappa(L) = n$, where $m,n \ge 6$,
 such that\\
   \hglue10pt  1.  $\kappa(K \cap L) \ge mn-3(m+n-4)$ and \\
   \hglue10pt  2.  $\kappa(K \setminus L)\ge mn-(2m+3n-9)$.
\end{proposition}

\begin{proof}
 Let $K$ and $L$
 be the binary factor-free languages accepted
 by the quotient automata of Figure~\ref{fig:binary_inter_diff}.
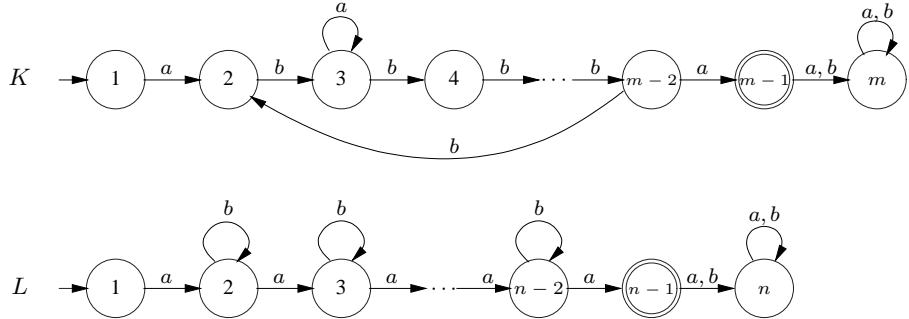
\begin{figure}[t]
 \centering
 \input{factor_binary.eepic}
 \caption{Binary factor-free witnesses  for intersection and difference.
         Missing transitions in the automaton accepting $K$ ($L$)  all go to the empty state $m$ ($n$).}
 \label{fig:binary_inter_diff}
\end{figure}

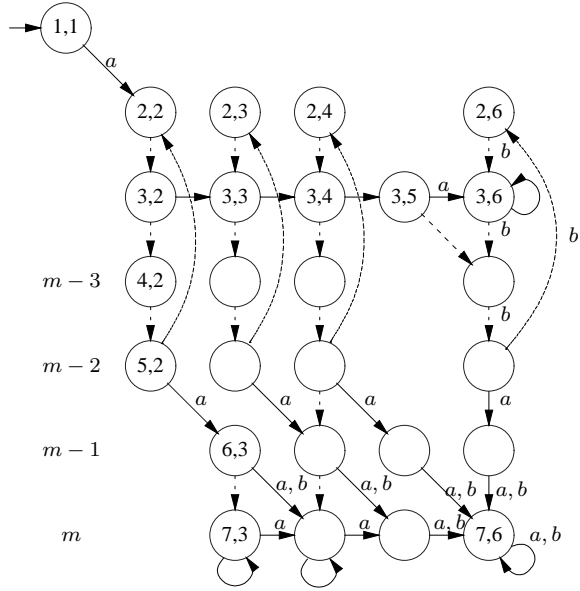
\begin{figure}[t]
 \centering
 \input{difference_factor_bin.eepic}
 \caption{Cross-product automaton for $m=6,n=7$. Missing transitions all go to state $(7,6)$.}
 \label{fig:binFacProd}
\end{figure}

 In the corresponding cross-product automaton of Figure~\ref{fig:binFacProd},
 except for $(1,1)$, no states in row 1 or column 1 are reachable.
 Also, states $(m-1,2)$ and $(m,2)$ are unreachable, as are the states in column $n-1$, except $(3,n-1)$, $(m-1,n-1)$, and $(m,n-1)$. The remaining states are all reachable.
 
 For intersection, the only accepting state is $(m-1,n-1)$,
 and all the other states in the last two rows and columns
 are empty.
 We will prove that states $(1,1)$, $(i,j)$ with $2\le i\le m-2$ and $2\le j\le n-2$, $(m-1,n-1)$,
 and $(m,n)$, which represents all the empty states, are all distinguishable.
 Then it follows that
 $\kappa(K\cap L) \ge (m-3)(n-3)+3=mn-3(m+n-4)$.
 
 State $(m,n)$ is the only empty state in our set.
 We show that for each other rejecting state $(i,j)$,
 there exists a word $w_{ij}$ 
 that is accepted only from state $(i,j)$.
 We have $w_{m-2,n-2}=a$ because word $a$ is accepted only from state $(m-2,n-2)$.
 Since only one transition on letter $b$ goes to state  $(m-2,n-2)$,
 and it goes from state $(m-3,n-2)$, the word $ba$
 is accepted only from state $(m-3,n-2)$. 
 Therefore $w_{m-3,n-2}=ba=bw_{m-2,n-2}$.
 For similar reasons we have

\begin{tabular}{ll}
 $w_{i,n-2}=bw_{i+1,n-2}$  & for $i=2,3,\ldots,m-3$,\\
 $w_{3j}=aw_{3,j+1}$       & for $j=2,3,\ldots,n-3$,\\
 $w_{2j}=bw_{3j}$          & for $j=2,3,\ldots,n-3$,\\
 $w_{m-2,j}=bw_{2j}$       & for $j=2,3,\ldots,n-3$,\\
 $w_{ij}=bw_{i+1,j}$       & for $i=4,5,\ldots,m-3$ and $j=2,3,\ldots,n-3$,\\
 $w_{11}=aw_{22}$,         &
\end{tabular}\\

\bigskip
\noindent
 which proves that $mn-3(m+n-4)$ states  are pairwise distinguishable.

 In the case of difference,
 all the states in row $m$, as well as state $(m-1,n-1)$ are empty.
 All the other states in row $m-1$ accept $\eps$, and so are equivalent.
 For each $i$ with $2\le i \le m-2$, 
 states $(i,n-1)$ and $(i,n)$ are equivalent.
 Among the other reachable states consider two distinct states $p$ and $q$.
 If they are in different rows,
 then by a word in $b^*$ we can send  $p$ to a state $p'$ in row 3,
 and $q$  to a state $q'$ that is not in row 3.
 Now by $a^n$, state $q'$ goes to the empty state,
 while $p'$ goes to state $(3,n)$ that is not empty.
 Two distinct states in the same row go by a word in $b^*$ to row 3. 
 Then, by a word in $a^*$, the first goes to 
 state $(3,n-2)$ while  the second to $(3,n)$,
 and now $b^{m-2-3}a$ distinguishes them.
 In summary,  $\kappa(K\setminus L) \ge (m-3)(n-3) +m-3 +3=mn-(2m+3n-9)$.
\end{proof}

The next proposition gives  lower bounds for union and symmetric difference
of binary bifix-free languages.

\begin{proposition} [\bf Union, Symmetric Difference: Binary Bifix-Free Languages; Lower Bound]
\label{prop6}
 Let $m,n \ge 6$. 
 There exist binary bifix-free languages $K$ and $L$
 with $\kappa(K) = m$ and $\kappa(L) = n$ 
 such that
 $\kappa(K \cup L), \kappa(K \oplus L) \ge mn-(m+n)-2$.
\end{proposition}

\begin{proof}
  Consider the binary languages 
 \begin{eqnarray*}
   K &=& a( (ba^*)^{m-5}b \cup a) (b ((ba^*)^{m-5}b \cup a) )^* a,\\
   L &=& a(a\cup b)^{n-4} (b(a\cup b)^{n-4})^* a.
 \end{eqnarray*}
 Quotient automata for $m=7$ and $n=6$ are shown in Figure~\ref{fig:xxx}.
Since both languages have  $\eps$ as the only accepting 
 quotient, they are prefix-free. Since the reverse automata are deterministic, the 
 reversed languages also have  $\eps$ as the only accepting quotient, and so are 
 prefix-free. Thus both languages are bifix-free.
 
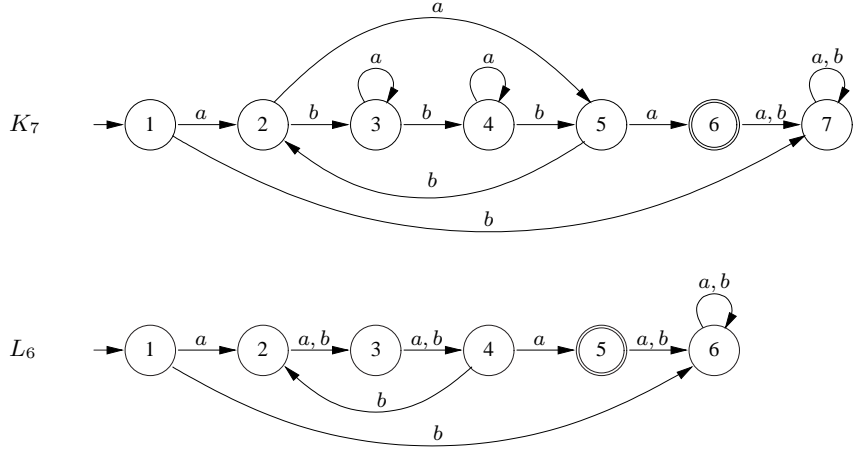
\begin{figure}[t]
\begin{center}
\input{binary_bifix_KL.eepic}
\end{center}
\caption{Binary bifix-free languages meeting the bound $mn-(m+n)-2$
            for union and symmetric difference.}
\label{fig:xxx}
\end{figure}
\begin{figure}[t]
\begin{center}
\input{bifix_bin_product.eepic}
\end{center}
\caption{Cross-product automaton for automata from Figure~\ref{fig:xxx}, 
        where dashed-transitions are on input $b$, and unspecified transitions go to
         state (7,6).}         
\label{fig:zzz}
\end{figure}
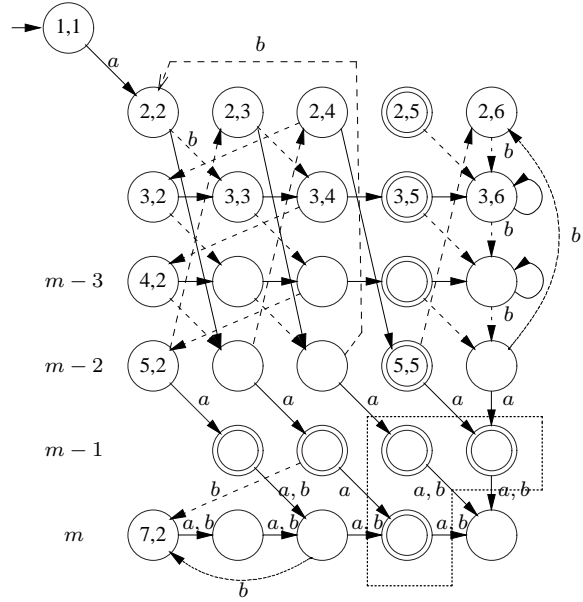
 The cross-product automaton is shown in Figure~\ref{fig:zzz}. 
 States in row~1 and column~1 
 are unreachable, with the exception of the initial state (1,1).
 Also, states $(2,n-1)$ and  $(m-1,2)$ are unreachable. 
 The initial state $(1,1)$ goes to state $(2,2)$ by $a$
 and then to state $(3,3)$ by $b$.
 From $(3,3)$, all the other states in row 3, except for $(3,2)$ are reached by $a$-transitions.
 Next, state $(3,n-2)$ goes to state $(4,2)$ by $b$,
 and then to $(4,j)$ by $a^{j-2}$ ($3 \le j \le n$).
 In this way, all the states in rows $4,5,\ldots,m-3$ can be reached.
 State $(m-3,n-2)$ goes to state $(m-2,2)$ by $b$, and states  $(m-2,j)$ with $j\ge3$,
 except for state $(m-2,n-1)$ that is reached from $(2,n-2)$ by $a$,
 are reached from  states  $(m-3,j-1)$ by $b$.
 States $(2,j)$ with $j\ge3$, except for $(2,n-1)$,
 are reached from    $(m-2,j-1)$ by $b$.
 State $(2,n-2)$ goes to  $(3,2)$ by $b$.
 From states in row $m-2$ all reachable states in row $m-1$ 
 are reached by $a$. State $(m,2)$ is reached by $b$ from $(m-1,n-2)$; 
 from here, all the other states is row $m$ are reached by words in $a^*$.
 
 For union, the three accepting states $(m-1,n-1),(m-1,n)$ and $(m,n-1)$
are equivalent.
 Consider the other reachable states.
 First, let $p=(i,j)$ and $q=(k,\ell)$ be two rejecting states with $i<k$.
 We can use $b$-transi\-tions to get  $p$
 into a state $p'$ in row 3, and $q$ into a state $q'$  in a row $i$ with $i\neq 3$.
 By $a^n$, state $p'$ goes to  $(3,n)$, 
 while $q'$ goes to $(i,n)$.
 Now  $b^{m-2-3}a$ is accepted from $(3,n)$ but not from $(i,n)$.
 Next, let $p$ and $q$ be two distinct rejecting states in the same row.
 If they are in the last row,
 then a word in $a^*$ distinguishes them.
 Otherwise, we can get them
 into states $(3,j)$ and $(3,\ell)$ with $j<\ell$,
 using $b$-transitions. Now  $(3,j)$ accepts $a^{n-1-j}$
 while  $(3,\ell)$ goes to the rejecting state $(3,n)$.
 Finally, consider two distinct accepting states
 different from $(m-1,n)$, $(m,n-1)$.
 By $b$, they go to two distinct rejecting, and so distinguishable,  states.
 The proof for symmetric difference is similar,
 except that now state  $(m-~1,n-~1)$ is empty.
\end{proof}

We now show that the upper bound for union of binary bifix-free languages
is the same as the lower bound
in the proposition above.

\begin{proposition}[\bf Union: Binary Bifix-Free Languages; Upper Bound]
\label{prop7}
 Let \linebreak $m,n\ge 4$
 and let $K$ and $L$ be binary bifix-free languages
 with $\kappa(K)=m$ and $\kappa(L)=n$. Then $\kappa(K\cup L)\le mn-(m+n)-2$.
\end{proposition}

\begin{proof}
  Let $K$  be a bifix-free  language
 accepted by the quotient automaton 
 ${\mathcal A}$   over $\{a,b\}$
 with states $1,2,\ldots,m$, where 1 is the initial state,
 $m-1$  is the only accepting state and it accepts only $\eps$,
 and $m$  is the empty state.
 Let $L$ be a similar language
 accepted by ${\mathcal B}$ with states  $1,2,\ldots,n$,
 initial state 1, state $n-1$ accepting $\varepsilon$, and empty state $n$.
 
 Construct the corresponding cross-product automaton
 with states $(i,j)$,
 where $i$ is a state of ${\mathcal A}$ and $j$ is a state of ${\mathcal B}$.
 In this cross-product automaton,
 we cannot go from columns $n-1$ and $n$, 
 as well as from rows $m-1$ and $m$,
 back to any state $(i,j)$ with $i<m-1$ or $j<n-1$.

 If state 1 of ${\mathcal A}$ goes by both inputs $a$ and $b$
 to a state in $\{m-1,m\}$, then no row $i$ with $i<m-1$ can be reached.
 Therefore, if the bound is to be met,
 at least one input, say $a$,
 takes state 1 to a state $i$ with $i<m-1$.
 Suppose also that $b$ takes 1 to a state in $\{m-1,m\}$. 
 A~similar condition applies to $L$.
  Suppose that input $b$ takes state 1 of ${\mathcal B}$ to a state $j$ with $j<n-1$,
 and $a$, to a state in $\{n-1,n\}$.
 Then no state $(i,j)$ with $i<m-1$ or $j<n-1$ can be reached.
 It follows that, without loss of generality,
 each automaton must take its initial state by $a$
 to a state that is neither accepting nor empty;
 for convenience, let this state be 2 in both automata.
 Then no other transition by $a$ may go to state 2 in the two automata,
 otherwise they would not be suffix-free.

 It follows that  in the cross-product automaton,
 all the states in row 2 and column~2, except for $(2,2)$,
 must be reached from some states by input $b$.
 Thus, if all the states are reachable,
 there must be an incoming transition by $b$ to each state $i$ with $i\ge 2$ in ${\mathcal A}$
 and $j$ with $j\ge 2$ in ${\mathcal B}$.
 In particular, if  state $(m-1,2)$ or $(2,n-1)$  is reachable,
 then some state, say $p_1$ (respectively $q_1$)
 different from $m-1$ (respectively  $n-1$) 
 must go to state $m-1$ (respectively $n-1$) in ${\mathcal A}$ (respectively ${\mathcal B}$).
 Now since $p_1$ goes to $m-1$ by $b$, it cannot go anywhere else by $b$.
 Thus there must be some other state $p_2$ not in $\{p_1,m-1,m\}$
 that goes to $p_1$ by $b$.
 Then there must be a state $p_3$ not in $\{p_2,p_1,m-1,m\}$ that goes to $p_2$ by $b$,
  and so on.
  Eventually, we have
  $$
  p_{m-3} \stackrel{b}{\rightarrow} p_{m-4} \stackrel{b}{\rightarrow} \cdots 
          \stackrel{b}{\rightarrow} p_3 
          \stackrel{b}{\rightarrow} p_2 
          \stackrel{b}{\rightarrow} p_1 
          \stackrel{b}{\rightarrow} m-1
          \stackrel{b}{\rightarrow} m,
 $$
 where all the states are pairwise distinct,
 and no state, except possibly state 1, goes by $b$ to state $p_{m-3}$.
 
 First assume state 1 goes to state $p_{m-3}$ by $b$.
 If $p_{m-3}=2$, then state 1 goes to state 2 by $a$ and by $b$.
 This means that there is no other transition to state 2,
 and so row 2 is not reachable in the cross-product automaton.
 If $p_{m-3}>2$ and 1 goes to  $p_{m-3}$ by $b$,
 then no other state goes to $p_{m-3}$ by $b$ because of suffix-freeness,
 and so  row $p_{m-3}$ may only be reached by $a$'s.
 However, in such a case state $(p_{m-3},2)$ is unreachable,
 since it is in row $p_{m-3}$ that can be reached only by $a$'s
 and at the same time in column 2 that can be reached only by $b$'s.

 Now assume that there is no transition by $b$ going to state $p_{m-3}$.
 If $p_{m-3}\ge3$, then $(p_{m-3},2)$  is unreachable.
 If $p_{m-3}=2$, then the whole row  2, except for $(2,2)$ is unreachable.
 The same considerations hold for automaton ${\mathcal B}$.
 This gives the desired upper bound $mn-(m+n)-2$. 
\end{proof}

We finally consider union and symmetric difference
of binary factor-free languages,
and give  upper bounds. 
We conjecture that the bounds are tight.

\begin{proposition} [\bf Union, Symmetric Difference: 
Binary Factor-Free Languages]
\label{prop8}
 Let $m,n \ge 6$. There exist binary factor-free languages $K$ and $L$
 with $\kappa(K) = m$ and $\kappa(L) = n$ such that
 $\kappa(K \cup L), \kappa(K \oplus L) 
 \ge mn-(m+n)-\min\{m-3,n-3\}$.
 We conjecture that this is largest bound for binary factor-free languages.
\end{proposition}

\begin{proof}
 Consider binary languages
   $K = a(b^*a)^{m-3}$, and
   $L = (a\cup b)(ba^*)^{n-4}b$.
 Quotient automata for $K$ and $L$ are shown in Figure~\ref{fig:yyy}. 

 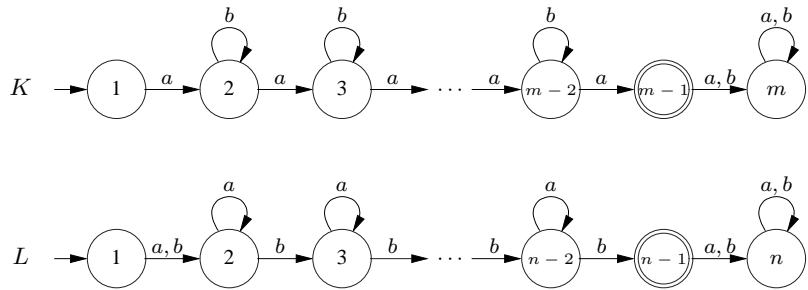
\begin{figure}[b]
 \centering
  \input{factor_bin_KL.eepic}
 \caption{Binary factor-free languages $K$ and $L$ meeting
   quotient complexity $mn - (m+n) - (m-3)$ 
   for union and symmetric difference.}
 \label{fig:yyy}
 \end{figure}

 To show that the languages are factor free,
 observe that every word $w$ in $K$ has exactly $m-2$  $a$'s, 
 while every proper factor of $w$ has less than $m-2$  $a$'s.
 Thus $K$ is factor-free. 
 For $L$, every word $w$ in $L$ either has $a$ as a prefix and has $n-3$  $b$'s,
 or has  $n-2$  $b$'s.
 However, every proper factor of $w$ either has $a$ as a prefix and has $n-4$
  $b$'s, or has $n-3$  $b$'s. Thus $L$ is also factor-free.

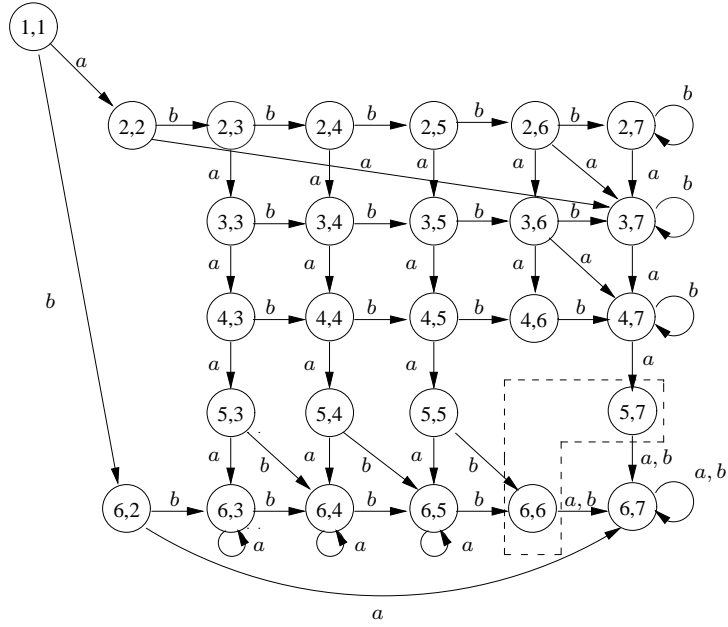
\begin{figure}[t]
\begin{center}
\input{union_factorfree.eepic}
\end{center}
\caption{Cross-product automaton for automata from Figure~\ref{fig:yyy}; 
 $m = 6$, $n = 7$.}
\label{fig:binfactorfreeunion}
\end{figure}

Construct the cross-product automaton for language $K\cup L$;
 see Figure~\ref{fig:binfactorfreeunion}.
 
 Consider the following family $\mathcal{R}$
 of $mn-(m+n)-(m-3)$ states: 
 \begin{eqnarray*}
   \mathcal{R} = \{(1,1),(2,2)\}  &\cup& \{(i,j)\mid 2\le i\le m-2, 3\le j\le n\} \cup\\
                                  & & \{(m-1,j)\mid 3\le j\le n-2\} \cup\\
                                  & & \{(m,j)\mid 2\le j\le n\},
 \end{eqnarray*}
 and let  us show that all states in $\mathcal{R}$
 are reachable and pairwise distinguishable.  
 The initial state $(1,1)$ goes to state $(2,2)$ by $a$,
 then to state $(2,3)$ by $b$,
 and then to state $(i,j)$ with $2\le i\le m-2$ and $3\le j\le n$ by $a^{i-2}b^{j-3}$.
 Each state $(m-2,j)$ with $3\le j\le n-2$ goes to state $(m-1,j)$ by $a$.
 State $(m,j)$ with $2\le j\le n$ is reached from the initial state $(1,1)$ by $b^{j-1}$.
 Thus all the states in $\mathcal{R}$ are reachable.

 For distinguishability, notice that 
 $a^{m-1}$ is accepted only from state $(1,1)$.
 Among the other states,
 two rejecting states in two distinct rows
 go to two distinc states in column $n$ by $b^n$,
 and the two states in column $n$ are  distinguished by a word in $a^*$.
 Two rejecting states in the same row $i$ go  by a word in $b^*$
 to states $(i,n-1)$ and $(i,n)$ that are distinguished by $\eps$.
 Two distinct accepting states in family $\mathcal{R}$ go by $b$
 either to two states, one of which is accepting and the other rejecting,
 or to two distinct rejecting, and so distinguishable, states.

 The proof for symmetric difference is exactly the same;
 notice that the languages are disjoint,
 and so their symmetric difference is the same as their union.

  Since union is a commutative operation,
 we may assume $m\le n$,
 and then the lower bound for  binary factor-free languages
 is $mn-(m+n)-(m-3)$.
 We did some computations
 by  enumerating all the binary factor-free
 automata in the case of $m,n \le 6$.
 The following table contains all the enumerated results:

\bigskip
\centerline{
 \begin{tabular}{l|lll}
  $m/n$    &$4$    & 5  &  6\\
 \hline
 $4$  & 7   &   &  \\
 $5$  & 10  &   13 & \\
 $6$  & 13  &  17  & 21
 \end{tabular}}
\bigskip 

\noindent
 All the entries, except for 21 ($m = n = 6$),
 are the same as for binary bifix-free languages.
 In case $m = n = 6$, 
 the complexity of union of binary factor-free languages is 21, that is $mn - (m+n) - (m-3)$.
 Thus it is the same as our lower bound.
 This is confirmed by the partial
 enumeration for $m = 6$ and $n = 7$,
 where we used
 a partial list of binary factor-free automata for $n = 7$.
 
 After quite a few unsuccessful attempts
 to get a larger value by the union of binary factor-free languages,
 we conjecture that $mn - (m+n) - (m-3)$
 is an upper bound if $m\le n$.
\end{proof}

We now turn our attention to subword-free languages.
The next theorem gives tight bounds for all four boolean operations
and shows that the bounds cannot be met using a fixed alphabet.

\begin{theorem}[\bf Boolean Operations: Subword-Free Languages]\label{thm:subword}
 Let $K$ and $L$  be subword-free languages
 over an alphabet $\Sigma$
 with $\kappa(K)=m$ and $\kappa(L)=n$,  
 where $m,n\ge4$.
Then\\
1. $\kappa(K\cup L), \kappa(K\oplus L)\le mn-(m+n)$,
                 and the bound is tight if $|\Sigma|\ge m+n-3$; \\
2. $\kappa(K\cap L)\le mn-3(m+n-4)$,
                 and the bound is tight if $|\Sigma|\ge m+n-7$;\\
3. $\kappa(K\setminus L)\le mn-(2m+3n-9)$,
                 and the bound is tight if $|\Sigma|\ge m+n-6$.\\
Moreover, the bounds cannot be met for smaller alphabets.
\end{theorem}

\begin{proof}
 Since subword-free languages are bifix-free,
 all the upper bounds apply.
 To prove tightness, let 
 $\Sigma=\{a,b,c\}\cup\{d_i\mid 3\le i \le m-1\}\cup\{e_j\mid 3\le j \le n-1\} $.
 Consider the languages $K$ and $L$ defined by the following quotient equations:

 \begin{tabular}{lcl}
 $K_1$ &=& $(a\cup b \cup e_3\cup\cdots\cup e_{n-1})K_2
 \cup \bigcup_{i=3}^{m-1}d_iK_i$,\\
 $K_i$ &=& $ aK_{i+1}\cup d_{i+1}K_{m-1}$ 
                               \hfill            $i=2, 3,\ldots, m-3$,\\
 $K_{m-2}$&=&$(a\cup b\cup d_{m-1}\cup e_3\cup e_4\cup\cdots\cup e_{n-1})K_{m-1}$,\\
 $K_{m-1}$&=&$\eps$,\\
 $K_{m}$&=&$\emptyset$, \\
 &&\\
 $L_1 $&=& $(a\cup c\cup  d_3\cup \cdots\cup d_{m-1})L_2
 \cup \bigcup_{j=3}^{n-1}e_jL_j$,\\
 $L_j$ &=&  $aL_{j+1}\cup e_{j+1}L_{n-1}$ 
                   \qquad \qquad \qquad \qquad \quad $j=2, 3,\ldots, n-3$,\\
 $L_{n-2}$&=&$(a\cup c\cup e_{n-1}\cup d_3\cup d_4\cup\cdots\cup d_{m-1})L_{n-1}$,\\
 $L_{n-1}$&=&$\eps$,\\
 $L_{n}$&=&$\emptyset$. 
 \end{tabular}

 \bigskip
 \noindent
 Figure~\ref{fig:booleansub}
 shows the quotient automata for languages $K$ and $L$ if $m=5$ and $n=6$.
 All the omited transitions go to the empty states $m$ and $n$.
\begin{figure}[t]
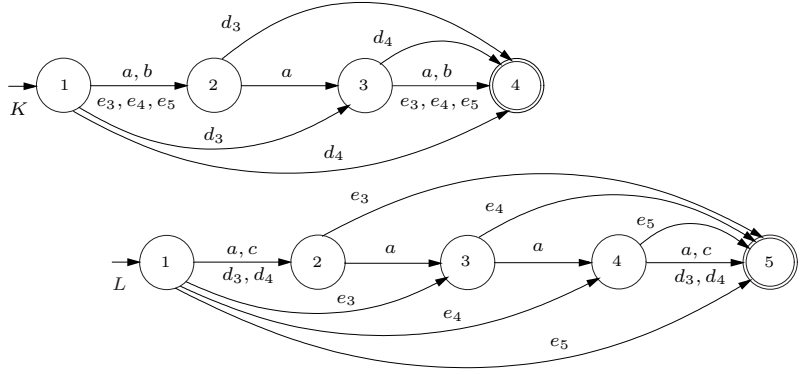

 \begin{center}
 \input booleansub.eepic
 \end{center}
 \caption{Subword-free witness languages for boolean operations; $m=5$, $n=6$.}
 \label{fig:booleansub}
 \end{figure}
 
 Let us show that languages 
 $K$ and $L$ are subword-free. 
 For this purpose, let
 $$\Gamma=\{a,b,e_3,e_4,\ldots, e_{n-1}\}, \text{ and }
 \Delta=\{d_3,d_4,\ldots, d_{m-1}\}.$$
 Notice that no word in $\Gamma^*$ of length less than $m-2$ is in $K$.
 Now let $w$ be a word in language $K$. 
 Then word $w$ either contains no letter from $\Delta$, 
 or  contains at most two such letters.
 If $w$ contains no letter from $\Delta$,
 then $w$ is a word in $\Gamma^*$
 of length $m-2$, and so no its proper subword is in $K$.
 If $w$ contains exactly one letter from $\Delta$, then
   either $w=ud_i$ for some word $u$  in $\Gamma^*$ of length $i-2$,
   or $w=d_iv$ for some word $v$  in $\Gamma^*$ of length $m-1-i$.
 In both cases, no proper subword of $w$ is in language $K$.
 Finally, if $w$ contains two letters from $\Delta$,
 then $w=d_ia^{k}d_{i+k+1}$ where $k\ge 0$ and $3\le i <i+k+1\le m-2$.
 No proper subword of such a word is in language $K$.
 This means that language $K$ is subword-free.
 The proof for language $L$ is similar.

 Figure~\ref{fig:boolean_sbw_prod} depicts
 the cross-product automaton of the dfa's
 for languages $K$ and $L$ defined in Figure~\ref{fig:booleansub},
 where we show only the transitions necessary
 to prove reachability and those caused by $b$ and $c$.
 In the cross-product automaton,
 states in the first row and the first column, except for the initial state   $(1,1)$,
 are unreachable. Now consider the remaining states.
 All the states in the second row and the second column
 are reached from  $(1,1)$ by symbols in $\Sigma$.
 Each other state  is reached from a state
 in the second row or second column by a word in $a^*$.

 \begin{figure}[t]
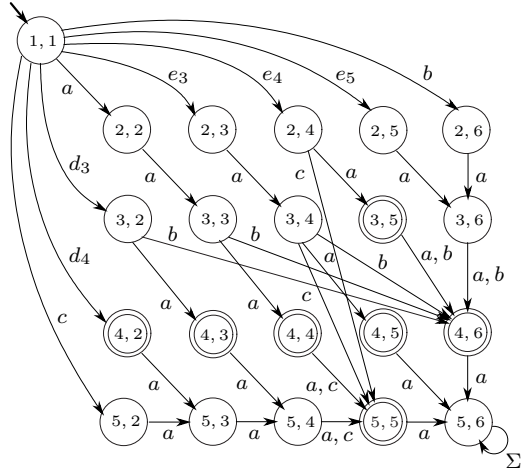

 \begin{center}
 \input KunionL.eepic
 \end{center}
 \caption{Reachability in the cross-product automaton
          for the union of languages from Figure~\ref{fig:booleansub}
          and transitions by $b$  and $c$.}
 \label{fig:boolean_sbw_prod}
 \end{figure}

 For union,
 all the states in row $m-1$ and in column $n-1$
 are accepting, and the three states 
 $(m,n-1)$, $(m-1,n-1)$, and $(m-1,n)$ 
 accept only $\eps$, and so
are equivalent.
 These three states are distinguishable from all other accepting states,
 since each of the other accepting states accepts at least one non-empty word. 
 Now let $(i,j)$ and $(k,\ell)$ be two distinct states 
 other than the three states accepting only word $\eps$.
 First assume that $i<k$.
 If $i=m-1$, then 
 state $(i,j)$ is accepting while  state $(k,\ell)$ is rejecting.
 If $i\le m-2$, then $a^{m-2-i}b$ is accepted from  state  $(i,j)$,
 but not from  state  $(k,\ell)$.
 Symmetrically, if $j<\ell$,
 then either $\eps$ or $a^{n-2-j}c$ distinguishes the two states.
 Therefore all the  $mn-(m+n)$ 
 states are pairwise distinguishable.

 For symmetric difference,  $(m-1,n-1)$ is empty;
 the rest of the proof is the same as for union.

 For intersection,
 the only accepting state is $(m-1,n-1)$,
 and all the rejecting states in the last two rows and last two columns are empty.
 Next, the word $a$ is accepted only from state $(m-2,n-2)$,
 the word $d_i$ ($3\le i \le n-2$)
 is accepted only from state $(i-1,n-2)$,
 while the word $e_j$ ($3\le i \le m-2$),
 only from state $(m-2,j-1)$.
 This means that for each state $(i,j)$,
 there exists a word in 
 $a^*(a\cup d_3 \cup \cdots \cup d_{m-2}\cup e_3 \cup \cdots \cup e_{n-2})$
 that is accepted only from $(i,j)$.
 So we get $mn-3(m+n-4)$ pairwise distinguishable states.
 Notice, that here we do not use transitions by symbols $b,c,d_{m-1},e_{n-1}$,
 and so we can simply omit these symbols
 to get witness languages over an alphabet of size $m+n-7$.

 For difference,
 all the states in row $m-1$, except for state $(m-1,n-1)$,
 are accepting and accept $\eps$.
 All the states in the last row, as well as state $(m-1,n-1)$,
 are empty, and states $(i,n-1)$ and $(i,n)$ with $2\le i \le m-2$
 are equivalent.
 States in different rows (up to row $m-1$)
 are distinguished by a word in $a^*b$.
 States in row $m-2$ are distinguished 
 by a word in $a\cup e_3\cup e_4 \cup \cdots \cup e_{n-2}$ 
 because $a$ distinguishes states $(m-2,n-2)$ and $(m-2,n-1)$,
 and if $2\le j<\ell\le n-1$ and $j\neq n-2$, 
 then word $e_{j+1}$ is not accepted from $(m-2,j)$ 
 but is accepted from $(m-2,\ell)$.
 Next, states $(i,n-2)$  and $(i,n-1)$ with $2\le i \le m-3$
 are distinguished by $d_{i+1}$.
 Finally, if two distinct states are in the same row,
 then there is a word in $a^*$,
 by which the two states either go to two distinct states in row $m-2$, or
 to two states $(i,n-2)$ and $(i,n-1)$ with $2\le i \le m-3$. 
 In both cases the resulting states are distinguishable,
 which proves the distinguishability of $mn-(2m+3n-9)$ states.
 Notice that now we do not use transitions by $c, d_{m-1}, e_{n-1}$,
 and so the bound is met for an alphabet of size $m+n-6$.

 We now show that the upper bounds cannot be met using smaller alphabets.
 Let the quotients of $K$ and $L$  be 
 $K=K_1,K_2,\dots,K_{m-2},K_{m-1}=\eps, K_{m}=\emp,$  and
 $L=L_\eps=L_{1},L_{2},\dots,L_{n-2},L_{n-1}=\eps, L_{n}=\emp,$
 ordered  as in Proposition~\ref{prop:finite}. 
 By Lemma~\ref{lem:letter},
 all the quotients of the form $K_2\cup L_i$ or $K_j\cup L_2$
 must  be reached by letters if the bound is to hold,
 and this is impossible if the size of the alphabet
 is  smaller than the number of such quotients. 
\end{proof}

\section{Product and Star}
\label{***prod}

The complexity of product of prefix-free languages is $m+n-2$~\cite{HSW09}.
For suffix-free languages,
the complexity  is $(m-1)2^{n-1}+1$~\cite{HaSa09}. 
Since bifix-free languages are prefix-free, 
and the witness prefix-free languages $a^{m-2}$ and $a^{n-2}$
are also subword-free, and we have the following result.

\begin{theorem}[\bf Product]\label{thm: prod1}
 If $K$ and $L$ are bifix-free with $\kappa(K)=m$ and $\kappa(L)=n$,
 where $m,n\ge2$,
 then $\kappa(KL)\leq m+n-2$.
 Furthermore, there are unary subword-free languages that meet this bound.
\end{theorem}

The complexity of star is $n$ for prefix-free languages~\cite{HSW09}, 
and $2^{n-2}+1$ for suffix-free languages  \cite{HaSa09}. 
We now extend these results to bifix-, factor-, and subword-free languages.
The quotient of $L^*$ by $\eps$ is $L^*=\eps\cup LL^*$, 
and the  following formula holds for a quotient of $L^*$ 
by a non-empty word $w$~\cite{Brz10}:
$$
    (L^*)_w=(L_w\cup \bigcup_{\substack{w=uv\\ \;\;u,v\in\Sig^+}}
            (L^*)_u ^\eps L_v)L^*.
$$

\begin{theorem}[\bf Star]\label{thm: star1}
 If $L$ is bifix-free with $\kappa(L)=n$,
 where $n\ge3$,
 then $\kappa(L^*)\leq n-1$. 
 Furthermore, there are binary subword-free languages that meet this bound.
\end{theorem}

\begin{proof}
 Assume that $L$ is bifix-free.
 Then it is prefix-free, has only one accepting quotient, namely $\eps$,
 and has the empty quotient, by Proposition~\ref{prop:pf}.
 Moreover, since $L$ is suffix-free,
 the quotient $L$ is uniquely reachable by $\eps$, by Proposition~\ref{prop:sf}.

 Let $L_w$ be a  non-empty quotient of $L$ by a non-empty word $w$.
 Let us show that $(L^*)^\eps_u=\emptyset$
 for every proper non-empty prefix $u$ of $w$.
 Assume for contradiction 
 that $\eps \in (L^*)_u$, where
 $w=uv$ for some non-empty words \mbox{$u$ and $v$.}
 Then $u \in L^*$,
 and  so there exist words $x$ in $L$ and  $y$ in $L^*$ such that $u=xy$.
 This gives $L_w=L_{xyv}=\eps_{yv}=\emptyset$
 because $x\in L$ implies  $L_x=\eps$. 
 This is a contradiction,
 and so we must have $(L^*)_u^\eps=\emptyset$.
 Hence, if $L_w$ is  non-empty,
 then $(L^*)_w=L_wL^*$, by the equation above. 
 Now if $L_w$ is accepting, then $L_w=\eps$,
 and so $(L^*)_w=L^*=(L^*)_\eps$.
 There are $n-2$ choices for 
 rejecting and non-empty quotients $L_w$.
 But, for a non-empty word $w$, we have $L_w\neq L$
 since $L$ is uniquely reachable by $\eps$.
 This reduces the number of choices to $n-3$ (since we have $n\ge3$).
 If $L_w=\emp$, then by the  observation above,
 $(L^*)_w = (L^*)^\eps_u L_v L^*$,
 where $w=uv$ and $v$ is the shortest word such that $L_v\neq \emp$.
 Such a quotient is either empty or has already been counted.
 In total, there are at most $n-1$ quotients of $L^*$.

 The  subword-free language $a^{n-2}$ over the alphabet $\{a,b\}$
 meets the bound since the language $(a^{n-2})^*$ has $n-2$ quotients
 of the form $a^{n-2-i}(a^{n-2})^*$ for $i=1,2,\dots,n-2$,
 and it has the empty quotient, for a total of $n-1$.
\end{proof}

\section{Reversal}
\label{***rev}

The last operation we consider is reversal. 
In \cite{HaSa09,HSW09} it was shown 
that the complexity of reversal is $2^{n-2}+1$ for suffix-free or prefix-free languages.
We show that this bound can be reduced  for bifix-free languages.
We use the standard method of reversing 
the quotient dfa $\cD$ of $L$ to obtain an nfa $\cN$ for $L^R$, 
and then we use  subset construction to find the dfa $\cD^R$ for $L^R$.

\begin{theorem}[\bf Reversal: Bifix- and Factor-Free Languages]\label{thm:rev1}
 If $L$ is a bifix-free language  with $\kappa(L)=n$,
 where $n\ge3$,
 then  $\kappa(L^R)\leq 2^{n-3}+2$.
 Moreover, there exist ternary factor-free languages that meet this bound.
\end{theorem}

\begin{proof}
 If $L$ is bifix-free, then so is $L^R$.
 Since $L$ is prefix-free,
 it has exactly one accepting quotient, $\eps$, and also has the empty quotient.
 
 Consider the quotient automaton $\cD$ for $L$,
 and remove the empty quotient and all the transitions to the empty quotient.
 Reverse this incomplete dfa
 to get an $(n-1)$-state nfa $\cN$ for $L^R$.
 Apply the subset construction to $\cN$ to get
 a dfa $\cD^R$ for $L^R$.
 The initial state of dfa $\cD^R$
 is the singleton set $\{f\}$,
 where $f$ is the $\eps$ quotient in quotient automaton $\cD$.
 No other subset containing state $f$ is reachable in $\cD^R$
 since no transition goes to state $f$ in nfa $\cN$.
 This gives at most $2^{n-2}+1$ reachable states.
 However, language $L^R$ is prefix-free,
 and so all the accepting states of $\cD^R$
 accept only the empty word,
 and can be merged into one state.
 Hence $\kappa(L^R)\leq 2^{n-3}+2$.
 
 If $n=3$ or $n=4$, then factor-free languages $a$ and $aa$, respectively,
 meet the bounds.
 
If $n\ge5$, then  consider the language $L=cKc$,
 where $K$  is  a regular language over the alphabet  $\{a,b\}$
 with $\kappa(K)=n-3$
 meeting the upper bound $2^{n-3}$ for  reversal \cite{Seb10}.
 The quotient automaton of  $L$ without the empty state 
 is shown in  Figure~\ref{fig:seb3}.
 \begin{figure}[t]
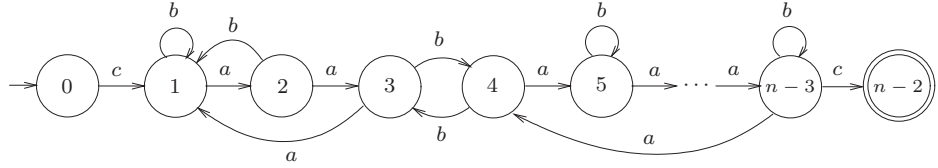

 \begin{center}
 \input sebej3.eepic
 \end{center}
 \caption{The ternary factor-free language meeting the $2^{n-3}+2$ bound for reversal.}  
 \label{fig:seb3}
 \end{figure}
 
 By Proposition~\ref{prop:construct},
 language $L$ is factor-free, and $\kappa(L)=n$.
 Since $\kappa(K^R)=2^{n-3}$,
 there exists a set $S$ of $2^{n-3}$ words over $\{a,b\}$
 that define distinct quotients of language $K^R$.
 Then the quotients of  $cK^Rc$
 by $2^{n-3}+2$ words $\eps$, $cw$ with $w\in S$,
 and $cuc$ for some word $u$ in $K^R$,
 are distinct as well.
 This gives $\kappa(L^R)=2^{n-3}+2$.
\end{proof}

\begin{theorem}[\bf Reversal: Subword-Free Languages]\label{thm:rev_sbw}
 If $L$ is a subword-free language  over an alphabet $\Sigma$
 with $\kappa(L)=n$,
 where $n\ge4$,
 then $\kappa(L^R) \le 2^{n-3}+2$.
 The bound is tight if $|\Sigma|\ge2^{n-3}-1$,
 but cannot be met for smaller alphabets.
 The bound cannot be met if  $L$
 contains a word of length at least 3.
\end{theorem} 

\begin{proof}
 Suppose $L $ is a subword-free language such that $\kappa(L^R) = 2^{n-3}+2$.
 Let $\cD=(Q,\Sig,\delta,s,f)$ be the quotient dfa
 of  $L$ with  $Q=\{s,q_1,\dots,q_{n-3},f,e\}$ as the state set,
 where $e$ and $f$ correspond to the quotients $\emp$ and $\eps$.
 Construct a dfa $\cD^R$ for $L^R$
 as in the proof of Theorem \ref{thm:rev1}.
 If $\kappa(L^R) = 2^{n-3}+2$,
 then the state $\{q_1,q_2,\dots,q_{n-3}\}$ must be reachable.
 Therefore there must exist a non-empty word $v$ such that,
 for all $q_i$, we have ${\delta}(q_i, v)=f$.
 Now suppose there exists a word $w$ in $L$ such that $|w|>2$.
 Let $w=abx$ where $a,b \in \Sig$ and $x \in \Sig^+$.
 Also suppose $\delta(s, a)=q_i$ and $\delta(q_i, b)=q_j$.
 Then we have $av, abv\in L$, showing that $L$ is not subword-free,
 which is a contradiction.
 Hence, if any word in $L$ has length at least 3,
 then $\kappa(L^R) < 2^{n-3}+2$.
 Now note that, if all the words in $L$ have length at most 2,
 the only possible quotients of $L^R$
 are $L^R$, $(L^R)_a$ for all $a \in \Sig$, $\eps$, and $\emp$.
 Therefore $\kappa(L^R) \leq |\Sig| + 3$,
 and the second claim follows.

 Now consider tightness.
 If $n=3$, then the bound is met
 by the unary subword-free language $a$.
 Let $n\ge 4$  and $\ell=2^{n-3}-1$.
 Also let $\Sig=\{a_1,a_2,\dots,a_\ell\}$, and let
 $S_1, S_2, \ldots, S_\ell$ 
 be all the non-empty subsets of $\{1,2,\ldots,n-3\}$.
 Now let
 $$
    L^R= a_1 (\bigcup_{j\in S_1} a_j) \cup 
         a_2 (\bigcup_{j\in S_2} a_j) \cup \cdots \cup
         a_\ell (\bigcup_{j\in S_\ell} a_j).
 $$
 Since $L^R$ only contains two-letter words,
 languages $L^R$ and $L$ are subword-free.
 The quotients of $L^R$ are
 $L^R$,
 $(L^R)_{a_i}=\bigcup_{j\in S_i} a_j$ for  $i=1,2,\ldots,\ell$,
 $ \eps$, and 
 $\emptyset$.
 Therefore $\kappa(L^R)=l+3=2^{n-3}+2$. 
 But for $L$, the only possible and distinct quotients are
 $L$,
 $L_{a_i}$ for $i=1,2,\ldots,n-3$,
 $\eps$, and $\emptyset$.
 Thus $\kappa(L)=n$.
\end{proof}

\section{Conclusions}
\label{***conc}

Our  results are summarized in Tables~\ref{tab:bool} and~\ref{tab:gen}, 
where ``B-, F-free'' stands for  bifix-free and factor-free, 
and ``S-free'' for subword-free.
The bounds for operations on prefix-free languages are from \cite{HaSa09,JiKr10},
for operations  on suffix-free languages 
from \cite{Cmo11,HSW09,JiOl09},
and those for regular languages, from~\cite{Lei81,Mas70,YZS94}.
For languages over a unary alphabet $\Sig=\{a\}$, 
the concepts prefix-, suffix-, factor-, and subword-free coincide, 
and $L$ is free with  $\kappa(L)=n$ if and only if $L=\{a^{n-2}\}$.

In the case of subword-free languages 
the size of the alphabet cannot be decreased.
In the other cases,  whenever the size of the alphabet is greater than~2,
we do not know whether or not the bounds are tight for smaller alphabets.

The fact that our bounds usually  apply only when $m,n\ge3$ 
is not a limitation, 
since bifix-free languages with smaller quotient complexities are simple. 
For $n=1$, we have only $\emp$,
for $n=2$, only $\eps$, and for $n=3$, a subset of $\Sig$. 
The~complexities of operations on such languages can be computed directly.

\setlength{\extrarowheight}{2pt}
\begin{table}[ht]
\begin{center}
\begin{scriptsize}
$
\begin{array}{| l | cc | cc | cc | }  
\hline
    &  K\cup L,  K \oplus L  &|\Sig|   & K\cap L      &|\Sig|   & K\setminus L &|\Sig|\\
\hline\hline
\txt{free unary} &\max(m,n)  &         & m=n          &         & m            &    \\
\hline\hline
\txt{prefix}     &mn-2       & 2       & mn-2(m+n-3)  &  2      & mn-(m+2n-4)  & 2  \\
\hline\hline 
\txt{suffix}     &mn-(m+n-2) & 2       & mn-2(m+n-3)  & 2       & mn-(m+2n-4)  & 2  \\
\hline\hline
\txt{B-, F-free} &mn-(m+n)   & 3       & mn-3(m+n-4)  & 2       & mn-(2m+3n-9) & 2  \\
\hline\hline
\txt{S-free}     &mn-(m+n)   & s_1     & mn-3(m+n-4)  & s_2     & mn-(2m+3n-9) & s_3\\
\hline\hline
\txt{regular }   &mn         & 2       & mn           & 2       & mn           &  2 \\
\hline
\end{array}
$
\end{scriptsize}
\end{center}
\caption{Complexities of boolean operations on free languages; 
           $s_1=m+n-3$, $s_2=m+n-7,s_3=m+n-6$.}
\label{tab:bool}
\end{table}

\begin{table}[ht]
\begin{center}
\begin{scriptsize}
$
\begin{array}{| l |cc|cc|cc|}
\hline
                  &  KL         &|\Sig|& L^*          &|\Sig|&  L^R       &|\Sig|       \\
\hline\hline
\txt{free unary}  & m+n-2          &   & n               &   & 2^{n-2}+1  &             \\
\hline\hline
\txt{prefix-free} & m+n-2          & 1 & n               & 2 & 2^{n-2}+1  & 3           \\
\hline\hline 
\txt{suffix-free} & (m-1)2^{n-1}+1 & 3 & 2^{n-2}+1       & 2 & 2^{n-2}+1  & 3           \\
\hline\hline 
\txt{B-, F-free}  & m+n-2          & 1 & n-1             & 2 &\ 2^{n-3}+2 & 3           \\
\hline\hline
\txt{S-free}      & m+n-2          & 1 & n-1             & 2 & 2^{n-3}+2  &\ 2^{n-3}-1 \ \\
\hline\hline
\txt{regular}     & (2m-1)2^{n-1}  & 2 &\ 2^{n-1}+2^{n-2}& 2 & 2^n        &  2           \\
\hline
\end{array}
$
\end{scriptsize}
\end{center}
\caption{Complexities of product, star, and reversal of free languages.}
\label{tab:gen}
\end{table}

\vskip-60pt

\providecommand{\noopsort}[1]{}

\end{document}

%% file: factor_new.eepic
\setlength{\unitlength}{0.00039370in}
\begingroup\makeatletter\ifx\SetFigFont\undefined%
\gdef\SetFigFont#1#2#3#4#5{%
  \reset@font\fontsize{#1}{#2pt}%
  \fontfamily{#3}\fontseries{#4}\fontshape{#5}%
  \selectfont}%
\fi\endgroup%
{\renewcommand{\dashlinestretch}{30}
\begin{picture}(11672,3449)(0,-10)
\put(5379.415,973.439){\arc{447.624}{2.4860}{6.8999}}
\blacken\thicklines
\path(5566.133,984.051)(5562.000,844.000)(5637.773,961.855)(5566.133,984.051)
\thinlines
\put(3864.415,981.439){\arc{447.624}{2.4860}{6.8999}}
\blacken\thicklines
\path(4051.133,992.051)(4047.000,852.000)(4122.773,969.855)(4051.133,992.051)
\thinlines
\put(8177.415,981.439){\arc{447.624}{2.4860}{6.8999}}
\blacken\thicklines
\path(8364.133,992.051)(8360.000,852.000)(8435.773,969.855)(8364.133,992.051)
\thinlines
\put(11237.415,988.439){\arc{447.624}{2.4860}{6.8999}}
\blacken\thicklines
\path(11424.133,999.051)(11420.000,859.000)(11495.773,976.855)(11424.133,999.051)
\thinlines
\put(3871.415,2863.439){\arc{447.624}{2.4860}{6.8999}}
\blacken\thicklines
\path(4058.133,2874.051)(4054.000,2734.000)(4129.773,2851.855)(4058.133,2874.051)
\thinlines
\put(2356.415,2871.439){\arc{447.624}{2.4860}{6.8999}}
\blacken\thicklines
\path(2543.133,2882.051)(2539.000,2742.000)(2614.773,2859.855)(2543.133,2882.051)
\thinlines
\put(6669.415,2871.439){\arc{447.624}{2.4860}{6.8999}}
\blacken\thicklines
\path(6856.133,2882.051)(6852.000,2742.000)(6927.773,2859.855)(6856.133,2882.051)
\thinlines
\put(9729.415,2878.439){\arc{447.624}{2.4860}{6.8999}}
\blacken\thicklines
\path(9916.133,2889.051)(9912.000,2749.000)(9987.773,2866.855)(9916.133,2889.051)
\thinlines
\put(9718,452){\ellipse{810}{810}}
\put(8178,442){\ellipse{810}{810}}
\put(11259,455){\ellipse{810}{810}}
\put(9721,455){\ellipse{720}{720}}
\put(5393,445){\ellipse{810}{810}}
\put(3871,439){\ellipse{810}{810}}
\put(8210,2342){\ellipse{810}{810}}
\put(6670,2332){\ellipse{810}{810}}
\put(2363,2329){\ellipse{810}{810}}
\put(823,2333){\ellipse{810}{810}}
\put(9751,2345){\ellipse{810}{810}}
\put(8213,2345){\ellipse{720}{720}}
\put(3885,2335){\ellipse{810}{810}}
\put(2331,443){\ellipse{810}{810}}
\path(8608,447)(9283,447)
\blacken\thicklines
\path(9148.000,409.500)(9283.000,447.000)(9148.000,484.500)(9148.000,409.500)
\thinlines
\path(10138,447)(10813,447)
\blacken\thicklines
\path(10678.000,409.500)(10813.000,447.000)(10678.000,484.500)(10678.000,409.500)
\thinlines
\path(7071,447)(7746,447)
\blacken\thicklines
\path(7611.000,409.500)(7746.000,447.000)(7611.000,484.500)(7611.000,409.500)
\thinlines
\path(5804,447)(6479,447)
\blacken\thicklines
\path(6344.000,409.500)(6479.000,447.000)(6344.000,484.500)(6344.000,409.500)
\thinlines
\path(4274,447)(4949,447)
\blacken\thicklines
\path(4814.000,409.500)(4949.000,447.000)(4814.000,484.500)(4814.000,409.500)
\thinlines
\path(2744,447)(3419,447)
\blacken\thicklines
\path(3284.000,409.500)(3419.000,447.000)(3284.000,484.500)(3284.000,409.500)
\thinlines
\path(1520,447)(1895,447)
\blacken\thicklines
\path(1760.000,409.500)(1895.000,447.000)(1760.000,484.500)(1760.000,409.500)
\thinlines
\path(7100,2337)(7775,2337)
\blacken\thicklines
\path(7640.000,2299.500)(7775.000,2337.000)(7640.000,2374.500)(7640.000,2299.500)
\thinlines
\path(8630,2337)(9305,2337)
\blacken\thicklines
\path(9170.000,2299.500)(9305.000,2337.000)(9170.000,2374.500)(9170.000,2299.500)
\thinlines
\path(5563,2337)(6238,2337)
\blacken\thicklines
\path(6103.000,2299.500)(6238.000,2337.000)(6103.000,2374.500)(6103.000,2299.500)
\thinlines
\path(4296,2337)(4971,2337)
\blacken\thicklines
\path(4836.000,2299.500)(4971.000,2337.000)(4836.000,2374.500)(4836.000,2299.500)
\thinlines
\path(2766,2337)(3441,2337)
\blacken\thicklines
\path(3306.000,2299.500)(3441.000,2337.000)(3306.000,2374.500)(3306.000,2299.500)
\thinlines
\path(1236,2337)(1911,2337)
\blacken\thicklines
\path(1776.000,2299.500)(1911.000,2337.000)(1776.000,2374.500)(1776.000,2299.500)
\thinlines
\path(12,2337)(387,2337)
\blacken\thicklines
\path(252.000,2299.500)(387.000,2337.000)(252.000,2374.500)(252.000,2299.500)
\put(6613,379){\makebox(0,0)[lb]{\smash{{\SetFigFont{8}{9.6}{\familydefault}{\mddefault}{\updefault}$\cdots$}}}}
\put(5322,403){\makebox(0,0)[lb]{\smash{{\SetFigFont{7}{8.4}{\familydefault}{\mddefault}{\updefault}$3$}}}}
\put(3778,395){\makebox(0,0)[lb]{\smash{{\SetFigFont{7}{8.4}{\familydefault}{\mddefault}{\updefault}$2$}}}}
\put(7851,387){\makebox(0,0)[lb]{\smash{{\SetFigFont{6}{7.2}{\familydefault}{\mddefault}{\updefault}$n-2$}}}}
\put(9420,395){\makebox(0,0)[lb]{\smash{{\SetFigFont{6}{7.2}{\familydefault}{\mddefault}{\updefault}$n-1$}}}}
\put(11144,409){\makebox(0,0)[lb]{\smash{{\SetFigFont{6}{7.2}{\familydefault}{\mddefault}{\updefault}$n$}}}}
\put(10152,612){\makebox(0,0)[lb]{\smash{{\SetFigFont{8}{9.6}{\familydefault}{\mddefault}{\updefault}$a,b,c$}}}}
\put(8750,605){\makebox(0,0)[lb]{\smash{{\SetFigFont{8}{9.6}{\familydefault}{\mddefault}{\updefault}$a,c$}}}}
\put(7197,605){\makebox(0,0)[lb]{\smash{{\SetFigFont{8}{9.6}{\familydefault}{\mddefault}{\updefault}$a,c$}}}}
\put(5937,605){\makebox(0,0)[lb]{\smash{{\SetFigFont{8}{9.6}{\familydefault}{\mddefault}{\updefault}$a,c$}}}}
\put(4399,612){\makebox(0,0)[lb]{\smash{{\SetFigFont{8}{9.6}{\familydefault}{\mddefault}{\updefault}$a,c$}}}}
\put(2967,598){\makebox(0,0)[lb]{\smash{{\SetFigFont{8}{9.6}{\familydefault}{\mddefault}{\updefault}$a$}}}}
\put(5317,1325){\makebox(0,0)[lb]{\smash{{\SetFigFont{8}{9.6}{\familydefault}{\mddefault}{\updefault}$b$}}}}
\put(8098,1317){\makebox(0,0)[lb]{\smash{{\SetFigFont{8}{9.6}{\familydefault}{\mddefault}{\updefault}$b$}}}}
\put(3763,1332){\makebox(0,0)[lb]{\smash{{\SetFigFont{8}{9.6}{\familydefault}{\mddefault}{\updefault}$b$}}}}
\put(10894,1340){\makebox(0,0)[lb]{\smash{{\SetFigFont{8}{9.6}{\familydefault}{\mddefault}{\updefault}$a,b,c$}}}}
\put(5105,2269){\makebox(0,0)[lb]{\smash{{\SetFigFont{8}{9.6}{\familydefault}{\mddefault}{\updefault}$\cdots$}}}}
\put(2293,3222){\makebox(0,0)[lb]{\smash{{\SetFigFont{8}{9.6}{\familydefault}{\mddefault}{\updefault}$c$}}}}
\put(3831,3215){\makebox(0,0)[lb]{\smash{{\SetFigFont{8}{9.6}{\familydefault}{\mddefault}{\updefault}$c$}}}}
\put(6628,3207){\makebox(0,0)[lb]{\smash{{\SetFigFont{8}{9.6}{\familydefault}{\mddefault}{\updefault}$c$}}}}
\put(755,2263){\makebox(0,0)[lb]{\smash{{\SetFigFont{7}{8.4}{\familydefault}{\mddefault}{\updefault}$1$}}}}
\put(3814,2293){\makebox(0,0)[lb]{\smash{{\SetFigFont{7}{8.4}{\familydefault}{\mddefault}{\updefault}$3$}}}}
\put(2270,2285){\makebox(0,0)[lb]{\smash{{\SetFigFont{7}{8.4}{\familydefault}{\mddefault}{\updefault}$2$}}}}
\put(6343,2277){\makebox(0,0)[lb]{\smash{{\SetFigFont{6}{7.2}{\familydefault}{\mddefault}{\updefault}$m-2$}}}}
\put(9621,2299){\makebox(0,0)[lb]{\smash{{\SetFigFont{6}{7.2}{\familydefault}{\mddefault}{\updefault}$m$}}}}
\put(7883,2285){\makebox(0,0)[lb]{\smash{{\SetFigFont{6}{7.2}{\familydefault}{\mddefault}{\updefault}$m-1$}}}}
\put(9416,3215){\makebox(0,0)[lb]{\smash{{\SetFigFont{8}{9.6}{\familydefault}{\mddefault}{\updefault}$a,b,c$}}}}
\put(8652,2487){\makebox(0,0)[lb]{\smash{{\SetFigFont{8}{9.6}{\familydefault}{\mddefault}{\updefault}$a,b,c$}}}}
\put(7242,2488){\makebox(0,0)[lb]{\smash{{\SetFigFont{8}{9.6}{\familydefault}{\mddefault}{\updefault}$a,b$}}}}
\put(5697,2487){\makebox(0,0)[lb]{\smash{{\SetFigFont{8}{9.6}{\familydefault}{\mddefault}{\updefault}$a,b$}}}}
\put(4429,2495){\makebox(0,0)[lb]{\smash{{\SetFigFont{8}{9.6}{\familydefault}{\mddefault}{\updefault}$a,b$}}}}
\put(2891,2495){\makebox(0,0)[lb]{\smash{{\SetFigFont{8}{9.6}{\familydefault}{\mddefault}{\updefault}$a,b$}}}}
\put(1459,2495){\makebox(0,0)[lb]{\smash{{\SetFigFont{8}{9.6}{\familydefault}{\mddefault}{\updefault}$a$}}}}
\put(2263,373){\makebox(0,0)[lb]{\smash{{\SetFigFont{7}{8.4}{\familydefault}{\mddefault}{\updefault}$1$}}}}
\put(1550,87){\makebox(0,0)[lb]{\smash{{\SetFigFont{8}{9.6}{\familydefault}{\mddefault}{\updefault}$L'$}}}}
\put(20,1977){\makebox(0,0)[lb]{\smash{{\SetFigFont{8}{9.6}{\familydefault}{\mddefault}{\updefault}$K'$}}}}
\end{picture}
}

%% file: crossprod.eepic
\setlength{\unitlength}{0.00027997in}
\begingroup\makeatletter\ifx\SetFigFont\undefined%
\gdef\SetFigFont#1#2#3#4#5{%
  \reset@font\fontsize{#1}{#2pt}%
  \fontfamily{#3}\fontseries{#4}\fontshape{#5}%
  \selectfont}%
\fi\endgroup%
{\renewcommand{\dashlinestretch}{30}
\begin{picture}(9152,8081)(0,-10)
\put(5334,1085){\makebox(0,0)[lb]{\smash{{\SetFigFont{8}{9.6}{\familydefault}{\mddefault}{\updefault}$a,c$}}}}
\put(8760.744,6326.269){\arc{545.260}{3.8230}{8.6925}}
\blacken\thicklines
\path(8691.439,6101.278)(8558.000,6144.000)(8650.287,6038.576)(8691.439,6101.278)
\thinlines
\put(8750.744,4715.269){\arc{545.260}{3.8230}{8.6925}}
\blacken\thicklines
\path(8681.439,4490.278)(8548.000,4533.000)(8640.287,4427.576)(8681.439,4490.278)
\thinlines
\put(1584.269,828.000){\arc{544.154}{5.3955}{10.2615}}
\blacken\thicklines
\path(1359.383,896.527)(1402.000,1030.000)(1296.648,937.630)(1359.383,896.527)
\thinlines
\put(3224.269,818.000){\arc{544.154}{5.3955}{10.2615}}
\blacken\thicklines
\path(2999.383,886.527)(3042.000,1020.000)(2936.648,927.630)(2999.383,886.527)
\thinlines
\put(4863.269,800.000){\arc{544.154}{5.3955}{10.2615}}
\blacken\thicklines
\path(4638.383,868.527)(4681.000,1002.000)(4575.648,909.630)(4638.383,868.527)
\thinlines
\put(1579,6337){\ellipse{902}{902}}
\put(3205,6336){\ellipse{902}{902}}
\put(8084,6334){\ellipse{902}{902}}
\put(6448,4712){\ellipse{902}{902}}
\put(8082,4714){\ellipse{902}{902}}
\put(4841,3082){\ellipse{902}{902}}
\put(6466,3092){\ellipse{902}{902}}
\put(8084,1481){\ellipse{902}{902}}
\put(4846,6340){\ellipse{902}{902}}
\put(6462,6327){\ellipse{902}{902}}
\put(8075,3093){\ellipse{902}{902}}
\put(4830,1463){\ellipse{902}{902}}
\put(1595,4714){\ellipse{902}{902}}
\put(1602,3088){\ellipse{902}{902}}
\put(3212,3094){\ellipse{902}{902}}
\put(3218,1478){\ellipse{902}{902}}
\put(1586,1472){\ellipse{902}{902}}
\put(3226,4707){\ellipse{902}{902}}
\put(4841,4719){\ellipse{902}{902}}
\put(618,7316){\ellipse{902}{902}}
\put(6467,1481){\ellipse{902}{902}}
\path(2036,6331)(2716,6331)
\blacken\thicklines
\path(2581.000,6293.500)(2716.000,6331.000)(2581.000,6368.500)(2581.000,6293.500)
\thinlines
\path(3657,6331)(4337,6331)
\blacken\thicklines
\path(4202.000,6293.500)(4337.000,6331.000)(4202.000,6368.500)(4202.000,6293.500)
\thinlines
\path(5287,6331)(5967,6331)
\blacken\thicklines
\path(5832.000,6293.500)(5967.000,6331.000)(5832.000,6368.500)(5832.000,6293.500)
\thinlines
\path(6918,6331)(7598,6331)
\blacken\thicklines
\path(7463.000,6293.500)(7598.000,6331.000)(7463.000,6368.500)(7463.000,6293.500)
\thinlines
\path(5296,4710)(5976,4710)
\blacken\thicklines
\path(5841.000,4672.500)(5976.000,4710.000)(5841.000,4747.500)(5841.000,4672.500)
\thinlines
\path(6917,4710)(7597,4710)
\blacken\thicklines
\path(7462.000,4672.500)(7597.000,4710.000)(7462.000,4747.500)(7462.000,4672.500)
\thinlines
\path(3665,1477)(4345,1477)
\blacken\thicklines
\path(4210.000,1439.500)(4345.000,1477.000)(4210.000,1514.500)(4210.000,1439.500)
\thinlines
\path(5286,1468)(5966,1468)
\blacken\thicklines
\path(5831.000,1430.500)(5966.000,1468.000)(5831.000,1505.500)(5831.000,1430.500)
\thinlines
\path(6935,1468)(7615,1468)
\blacken\thicklines
\path(7480.000,1430.500)(7615.000,1468.000)(7480.000,1505.500)(7480.000,1430.500)
\thinlines
\path(1589,5885)(1589,5205)
\blacken\thicklines
\path(1551.500,5340.000)(1589.000,5205.000)(1626.500,5340.000)(1551.500,5340.000)
\thinlines
\path(3693,4710)(4373,4710)
\blacken\thicklines
\path(4238.000,4672.500)(4373.000,4710.000)(4238.000,4747.500)(4238.000,4672.500)
\thinlines
\path(2064,4710)(2744,4710)
\blacken\thicklines
\path(2609.000,4672.500)(2744.000,4710.000)(2609.000,4747.500)(2609.000,4672.500)
\thinlines
\path(2045,1458)(2725,1458)
\blacken\thicklines
\path(2590.000,1420.500)(2725.000,1458.000)(2590.000,1495.500)(2590.000,1420.500)
\thinlines
\path(1598,4264)(1598,3584)
\blacken\thicklines
\path(1560.500,3719.000)(1598.000,3584.000)(1635.500,3719.000)(1560.500,3719.000)
\thinlines
\path(1607,2624)(1607,1944)
\blacken\thicklines
\path(1569.500,2079.000)(1607.000,1944.000)(1644.500,2079.000)(1569.500,2079.000)
\thinlines
\path(3228,5876)(3228,5196)
\blacken\thicklines
\path(3190.500,5331.000)(3228.000,5196.000)(3265.500,5331.000)(3190.500,5331.000)
\thinlines
\path(3228,4246)(3228,3566)
\blacken\thicklines
\path(3190.500,3701.000)(3228.000,3566.000)(3265.500,3701.000)(3190.500,3701.000)
\thinlines
\path(3210,2634)(3210,1954)
\blacken\thicklines
\path(3172.500,2089.000)(3210.000,1954.000)(3247.500,2089.000)(3172.500,2089.000)
\thinlines
\path(4849,5886)(4849,5206)
\blacken\thicklines
\path(4811.500,5341.000)(4849.000,5206.000)(4886.500,5341.000)(4811.500,5341.000)
\thinlines
\path(4849,4265)(4849,3585)
\blacken\thicklines
\path(4811.500,3720.000)(4849.000,3585.000)(4886.500,3720.000)(4811.500,3720.000)
\thinlines
\path(4849,2634)(4849,1954)
\blacken\thicklines
\path(4811.500,2089.000)(4849.000,1954.000)(4886.500,2089.000)(4811.500,2089.000)
\thinlines
\path(8092,5895)(8092,5215)
\blacken\thicklines
\path(8054.500,5350.000)(8092.000,5215.000)(8129.500,5350.000)(8054.500,5350.000)
\thinlines
\path(8092,4256)(8092,3576)
\blacken\thicklines
\path(8054.500,3711.000)(8092.000,3576.000)(8129.500,3711.000)(8054.500,3711.000)
\thinlines
\path(8082,2634)(8082,1954)
\blacken\thicklines
\path(8044.500,2089.000)(8082.000,1954.000)(8119.500,2089.000)(8044.500,2089.000)
\thinlines
\path(1902,2731)(2899,1771)
\blacken\thicklines
\path(2775.743,1837.625)(2899.000,1771.000)(2827.764,1891.651)(2775.743,1837.625)
\thinlines
\path(5169,4405)(6166,3445)
\blacken\thicklines
\path(6042.743,3511.625)(6166.000,3445.000)(6094.764,3565.651)(6042.743,3511.625)
\thinlines
\path(6758,4364)(7755,3404)
\blacken\thicklines
\path(7631.743,3470.625)(7755.000,3404.000)(7683.764,3524.651)(7631.743,3470.625)
\thinlines
\path(3519,6006)(4516,5046)
\blacken\thicklines
\path(4392.743,5112.625)(4516.000,5046.000)(4444.764,5166.651)(4392.743,5112.625)
\thinlines
\path(5158,6025)(6155,5065)
\blacken\thicklines
\path(6031.743,5131.625)(6155.000,5065.000)(6083.764,5185.651)(6031.743,5131.625)
\thinlines
\path(6811,6010)(7808,5050)
\blacken\thicklines
\path(7684.743,5116.625)(7808.000,5050.000)(7736.764,5170.651)(7684.743,5116.625)
\thinlines
\path(899,6974)(1216,6657)
\blacken\thicklines
\path(1094.024,6725.943)(1216.000,6657.000)(1147.057,6778.976)(1094.024,6725.943)
\path(22,8044)(307,7699)
\blacken\path(192.110,7779.197)(307.000,7699.000)(249.932,7826.963)(246.815,7771.856)(192.110,7779.197)
\thinlines
\path(1893,4381)(2890,3421)
\blacken\thicklines
\path(2766.743,3487.625)(2890.000,3421.000)(2818.764,3541.651)(2766.743,3487.625)
\thinlines
\path(3509,2773)(4506,1813)
\blacken\thicklines
\path(4382.743,1879.625)(4506.000,1813.000)(4434.764,1933.651)(4382.743,1879.625)
\thinlines
\path(3540,4388)(4537,3428)
\blacken\thicklines
\path(4413.743,3494.625)(4537.000,3428.000)(4465.764,3548.651)(4413.743,3494.625)
\thinlines
\path(1894,6003)(2891,5043)
\blacken\thicklines
\path(2767.743,5109.625)(2891.000,5043.000)(2819.764,5163.651)(2767.743,5109.625)
\thinlines
\path(5195,2792)(6192,1832)
\blacken\thicklines
\path(6068.743,1898.625)(6192.000,1832.000)(6120.764,1952.651)(6068.743,1898.625)
\thinlines
\path(6786,2779)(7783,1819)
\blacken\thicklines
\path(7659.743,1885.625)(7783.000,1819.000)(7711.764,1939.651)(7659.743,1885.625)
\put(1326,6257){\makebox(0,0)[lb]{\smash{{\SetFigFont{6}{7.2}{\familydefault}{\mddefault}{\updefault}$2,2$}}}}
\put(2945,6266){\makebox(0,0)[lb]{\smash{{\SetFigFont{6}{7.2}{\familydefault}{\mddefault}{\updefault}$2,3$}}}}
\put(4557,6257){\makebox(0,0)[lb]{\smash{{\SetFigFont{6}{7.2}{\familydefault}{\mddefault}{\updefault}$2,4$}}}}
\put(6193,6256){\makebox(0,0)[lb]{\smash{{\SetFigFont{6}{7.2}{\familydefault}{\mddefault}{\updefault}$2,5$}}}}
\put(7802,6246){\makebox(0,0)[lb]{\smash{{\SetFigFont{6}{7.2}{\familydefault}{\mddefault}{\updefault}$2,6$}}}}
\put(7812,4631){\makebox(0,0)[lb]{\smash{{\SetFigFont{6}{7.2}{\familydefault}{\mddefault}{\updefault}$3,6$}}}}
\put(4556,3016){\makebox(0,0)[lb]{\smash{{\SetFigFont{6}{7.2}{\familydefault}{\mddefault}{\updefault}$4,4$}}}}
\put(6184,3016){\makebox(0,0)[lb]{\smash{{\SetFigFont{6}{7.2}{\familydefault}{\mddefault}{\updefault}$4,5$}}}}
\put(7801,3018){\makebox(0,0)[lb]{\smash{{\SetFigFont{6}{7.2}{\familydefault}{\mddefault}{\updefault}$4,6$}}}}
\put(7822,1392){\makebox(0,0)[lb]{\smash{{\SetFigFont{6}{7.2}{\familydefault}{\mddefault}{\updefault}$5,6$}}}}
\put(6192,1402){\makebox(0,0)[lb]{\smash{{\SetFigFont{6}{7.2}{\familydefault}{\mddefault}{\updefault}$5,5$}}}}
\put(4556,1382){\makebox(0,0)[lb]{\smash{{\SetFigFont{6}{7.2}{\familydefault}{\mddefault}{\updefault}$5,4$}}}}
\put(2946,1393){\makebox(0,0)[lb]{\smash{{\SetFigFont{6}{7.2}{\familydefault}{\mddefault}{\updefault}$5,3$}}}}
\put(6193,4649){\makebox(0,0)[lb]{\smash{{\SetFigFont{6}{7.2}{\familydefault}{\mddefault}{\updefault}$3,5$}}}}
\put(1317,4632){\makebox(0,0)[lb]{\smash{{\SetFigFont{6}{7.2}{\familydefault}{\mddefault}{\updefault}$3,2$}}}}
\put(1322,3005){\makebox(0,0)[lb]{\smash{{\SetFigFont{6}{7.2}{\familydefault}{\mddefault}{\updefault}$4,2$}}}}
\put(1311,1393){\makebox(0,0)[lb]{\smash{{\SetFigFont{6}{7.2}{\familydefault}{\mddefault}{\updefault}$5,2$}}}}
\put(2926,3005){\makebox(0,0)[lb]{\smash{{\SetFigFont{6}{7.2}{\familydefault}{\mddefault}{\updefault}$4,3$}}}}
\put(2928,4641){\makebox(0,0)[lb]{\smash{{\SetFigFont{6}{7.2}{\familydefault}{\mddefault}{\updefault}$3,3$}}}}
\put(4574,4633){\makebox(0,0)[lb]{\smash{{\SetFigFont{6}{7.2}{\familydefault}{\mddefault}{\updefault}$3,4$}}}}
\put(2260,6485){\makebox(0,0)[lb]{\smash{{\SetFigFont{8}{9.6}{\familydefault}{\mddefault}{\updefault}$c$}}}}
\put(1203,5414){\makebox(0,0)[lb]{\smash{{\SetFigFont{8}{9.6}{\familydefault}{\mddefault}{\updefault}$b$}}}}
\put(2844,5425){\makebox(0,0)[lb]{\smash{{\SetFigFont{8}{9.6}{\familydefault}{\mddefault}{\updefault}$b$}}}}
\put(4464,5424){\makebox(0,0)[lb]{\smash{{\SetFigFont{8}{9.6}{\familydefault}{\mddefault}{\updefault}$b$}}}}
\put(1204,3803){\makebox(0,0)[lb]{\smash{{\SetFigFont{8}{9.6}{\familydefault}{\mddefault}{\updefault}$b$}}}}
\put(2844,3804){\makebox(0,0)[lb]{\smash{{\SetFigFont{8}{9.6}{\familydefault}{\mddefault}{\updefault}$b$}}}}
\put(4464,3803){\makebox(0,0)[lb]{\smash{{\SetFigFont{8}{9.6}{\familydefault}{\mddefault}{\updefault}$b$}}}}
\put(4474,2183){\makebox(0,0)[lb]{\smash{{\SetFigFont{8}{9.6}{\familydefault}{\mddefault}{\updefault}$b$}}}}
\put(2825,2163){\makebox(0,0)[lb]{\smash{{\SetFigFont{8}{9.6}{\familydefault}{\mddefault}{\updefault}$b$}}}}
\put(1222,2172){\makebox(0,0)[lb]{\smash{{\SetFigFont{8}{9.6}{\familydefault}{\mddefault}{\updefault}$b$}}}}
\put(3919,6513){\makebox(0,0)[lb]{\smash{{\SetFigFont{8}{9.6}{\familydefault}{\mddefault}{\updefault}$c$}}}}
\put(7152,6494){\makebox(0,0)[lb]{\smash{{\SetFigFont{8}{9.6}{\familydefault}{\mddefault}{\updefault}$c$}}}}
\put(3899,4873){\makebox(0,0)[lb]{\smash{{\SetFigFont{8}{9.6}{\familydefault}{\mddefault}{\updefault}$c$}}}}
\put(2297,4883){\makebox(0,0)[lb]{\smash{{\SetFigFont{8}{9.6}{\familydefault}{\mddefault}{\updefault}$c$}}}}
\put(7132,4874){\makebox(0,0)[lb]{\smash{{\SetFigFont{8}{9.6}{\familydefault}{\mddefault}{\updefault}$c$}}}}
\put(5530,6504){\makebox(0,0)[lb]{\smash{{\SetFigFont{8}{9.6}{\familydefault}{\mddefault}{\updefault}$c$}}}}
\put(5502,4874){\makebox(0,0)[lb]{\smash{{\SetFigFont{8}{9.6}{\familydefault}{\mddefault}{\updefault}$c$}}}}
\put(348,7219){\makebox(0,0)[lb]{\smash{{\SetFigFont{6}{7.2}{\familydefault}{\mddefault}{\updefault}$1,1$}}}}
\put(1115,6905){\makebox(0,0)[lb]{\smash{{\SetFigFont{8}{9.6}{\familydefault}{\mddefault}{\updefault}$a$}}}}
\put(2052,5328){\makebox(0,0)[lb]{\smash{{\SetFigFont{8}{9.6}{\familydefault}{\mddefault}{\updefault}$a$}}}}
\put(3702,5309){\makebox(0,0)[lb]{\smash{{\SetFigFont{8}{9.6}{\familydefault}{\mddefault}{\updefault}$a$}}}}
\put(5332,5309){\makebox(0,0)[lb]{\smash{{\SetFigFont{8}{9.6}{\familydefault}{\mddefault}{\updefault}$a$}}}}
\put(5360,3688){\makebox(0,0)[lb]{\smash{{\SetFigFont{8}{9.6}{\familydefault}{\mddefault}{\updefault}$a$}}}}
\put(3701,3670){\makebox(0,0)[lb]{\smash{{\SetFigFont{8}{9.6}{\familydefault}{\mddefault}{\updefault}$a$}}}}
\put(2054,3679){\makebox(0,0)[lb]{\smash{{\SetFigFont{8}{9.6}{\familydefault}{\mddefault}{\updefault}$a$}}}}
\put(8872,682){\makebox(0,0)[lb]{\smash{{\SetFigFont{8}{9.6}{\familydefault}{\mddefault}{\updefault}$\Sigma$}}}}
\put(6526,3712){\makebox(0,0)[lb]{\smash{{\SetFigFont{8}{9.6}{\familydefault}{\mddefault}{\updefault}$a,b$}}}}
\put(6804,2033){\makebox(0,0)[lb]{\smash{{\SetFigFont{8}{9.6}{\familydefault}{\mddefault}{\updefault}$\Sigma$}}}}
\put(7055,1025){\makebox(0,0)[lb]{\smash{{\SetFigFont{8}{9.6}{\familydefault}{\mddefault}{\updefault}$\Sigma$}}}}
\put(8266,2172){\makebox(0,0)[lb]{\smash{{\SetFigFont{8}{9.6}{\familydefault}{\mddefault}{\updefault}$\Sigma$}}}}
\put(9127,6196){\makebox(0,0)[lb]{\smash{{\SetFigFont{8}{9.6}{\familydefault}{\mddefault}{\updefault}$c$}}}}
\put(9137,4575){\makebox(0,0)[lb]{\smash{{\SetFigFont{8}{9.6}{\familydefault}{\mddefault}{\updefault}$c$}}}}
\put(4716,132){\makebox(0,0)[lb]{\smash{{\SetFigFont{8}{9.6}{\familydefault}{\mddefault}{\updefault}$b$}}}}
\put(3086,123){\makebox(0,0)[lb]{\smash{{\SetFigFont{8}{9.6}{\familydefault}{\mddefault}{\updefault}$b$}}}}
\put(1465,123){\makebox(0,0)[lb]{\smash{{\SetFigFont{8}{9.6}{\familydefault}{\mddefault}{\updefault}$b$}}}}
\put(6572,5399){\makebox(0,0)[lb]{\smash{{\SetFigFont{8}{9.6}{\familydefault}{\mddefault}{\updefault}$a,b$}}}}
\put(8184,5445){\makebox(0,0)[lb]{\smash{{\SetFigFont{8}{9.6}{\familydefault}{\mddefault}{\updefault}$a,b$}}}}
\put(8175,3814){\makebox(0,0)[lb]{\smash{{\SetFigFont{8}{9.6}{\familydefault}{\mddefault}{\updefault}$a,b$}}}}
\put(2204,2520){\makebox(0,0)[lb]{\smash{{\SetFigFont{8}{9.6}{\familydefault}{\mddefault}{\updefault}$a,c$}}}}
\put(5465,2577){\makebox(0,0)[lb]{\smash{{\SetFigFont{8}{9.6}{\familydefault}{\mddefault}{\updefault}$a,c$}}}}
\put(3798,2537){\makebox(0,0)[lb]{\smash{{\SetFigFont{8}{9.6}{\familydefault}{\mddefault}{\updefault}$a,c$}}}}
\put(2111,1094){\makebox(0,0)[lb]{\smash{{\SetFigFont{8}{9.6}{\familydefault}{\mddefault}{\updefault}$a,c$}}}}
\put(3722,1103){\makebox(0,0)[lb]{\smash{{\SetFigFont{8}{9.6}{\familydefault}{\mddefault}{\updefault}$a,c$}}}}
\thinlines
\put(8624.634,1130.976){\arc{521.428}{4.5985}{9.2084}}
\blacken\thicklines
\path(8466.372,973.296)(8370.000,1075.000)(8400.110,938.162)(8414.269,991.510)(8466.372,973.296)
\end{picture}
}

%% file: factor_binary.eepic
\setlength{\unitlength}{0.00065617in}
\begingroup\makeatletter\ifx\SetFigFont\undefined%
\gdef\SetFigFont#1#2#3#4#5{%
  \reset@font\fontsize{#1}{#2pt}%
  \fontfamily{#3}\fontseries{#4}\fontshape{#5}%
  \selectfont}%
\fi\endgroup%
{\renewcommand{\dashlinestretch}{30}
\begin{picture}(7097,2584)(0,-10)
\put(5955,1906){\ellipse{468}{468}}
\put(5955,1906){\ellipse{406}{406}}
\put(6855,1906){\ellipse{468}{468}}
\put(5055,1906){\ellipse{468}{468}}
\put(3480,1906){\ellipse{468}{468}}
\put(2580,1906){\ellipse{468}{468}}
\put(1680,1906){\ellipse{468}{468}}
\put(780,1906){\ellipse{468}{468}}
\put(780,241){\ellipse{468}{468}}
\put(1680,241){\ellipse{468}{468}}
\put(2580,241){\ellipse{468}{468}}
\put(4155,241){\ellipse{468}{468}}
\put(5055,241){\ellipse{468}{468}}
\put(5055,241){\ellipse{406}{406}}
\put(5955,241){\ellipse{468}{468}}
\path(330,241)(555,241)
\blacken\path(435.000,211.000)(555.000,241.000)(435.000,271.000)(435.000,211.000)
\path(1005,241)(1455,241)
\blacken\path(1335.000,211.000)(1455.000,241.000)(1335.000,271.000)(1335.000,211.000)
\path(1905,241)(2355,241)
\blacken\path(2235.000,211.000)(2355.000,241.000)(2235.000,271.000)(2235.000,211.000)
\path(2805,241)(3255,241)
\blacken\path(3135.000,211.000)(3255.000,241.000)(3135.000,271.000)(3135.000,211.000)
\path(330,1906)(555,1906)
\blacken\path(435.000,1876.000)(555.000,1906.000)(435.000,1936.000)(435.000,1876.000)
\path(1005,1906)(1455,1906)
\blacken\path(1335.000,1876.000)(1455.000,1906.000)(1335.000,1936.000)(1335.000,1876.000)
\path(1905,1906)(2355,1906)
\blacken\path(2235.000,1876.000)(2355.000,1906.000)(2235.000,1936.000)(2235.000,1876.000)
\path(3705,1906)(4155,1906)
\blacken\path(4035.000,1876.000)(4155.000,1906.000)(4035.000,1936.000)(4035.000,1876.000)
\path(5280,1906)(5730,1906)
\blacken\path(5610.000,1876.000)(5730.000,1906.000)(5610.000,1936.000)(5610.000,1876.000)
\path(6180,1906)(6630,1906)
\blacken\path(6510.000,1876.000)(6630.000,1906.000)(6510.000,1936.000)(6510.000,1876.000)
\path(4425,1906)(4830,1906)
\blacken\path(4710.000,1876.000)(4830.000,1906.000)(4710.000,1936.000)(4710.000,1876.000)
\path(2805,1906)(3233,1906)
\blacken\path(3113.000,1876.000)(3233.000,1906.000)(3113.000,1936.000)(3113.000,1876.000)
\path(5280,241)(5708,241)
\blacken\path(5588.000,211.000)(5708.000,241.000)(5588.000,271.000)(5588.000,211.000)
\path(4380,241)(4808,241)
\blacken\path(4688.000,211.000)(4808.000,241.000)(4688.000,271.000)(4688.000,211.000)
\path(3502,241)(3930,241)
\blacken\path(3810.000,211.000)(3930.000,241.000)(3810.000,271.000)(3810.000,211.000)
\path(6751,2107)(6749,2110)(6745,2117)
	(6739,2128)(6731,2143)(6723,2160)
	(6715,2179)(6708,2199)(6703,2219)
	(6700,2241)(6701,2264)(6706,2287)
	(6714,2305)(6723,2320)(6732,2332)
	(6740,2340)(6747,2346)(6753,2351)
	(6759,2355)(6764,2358)(6771,2361)
	(6780,2364)(6791,2368)(6805,2372)
	(6822,2375)(6841,2377)(6860,2375)
	(6877,2372)(6891,2368)(6902,2364)
	(6911,2361)(6918,2358)(6924,2354)
	(6929,2351)(6935,2346)(6942,2340)
	(6950,2332)(6959,2320)(6968,2305)
	(6976,2287)(6981,2264)(6982,2241)
	(6979,2219)(6974,2199)(6967,2179)
	(6959,2160)(6951,2143)(6931,2107)
\blacken\path(6963.052,2226.468)(6931.000,2107.000)(7015.502,2197.330)(6963.052,2226.468)
\path(4830,1816)(4829,1815)(4828,1814)
	(4825,1812)(4820,1808)(4813,1802)
	(4803,1795)(4792,1786)(4777,1775)
	(4760,1762)(4740,1747)(4718,1731)
	(4693,1713)(4665,1694)(4635,1673)
	(4604,1652)(4570,1629)(4535,1607)
	(4498,1584)(4460,1561)(4420,1538)
	(4379,1515)(4336,1493)(4292,1471)
	(4246,1450)(4199,1429)(4149,1410)
	(4098,1391)(4044,1373)(3988,1356)
	(3929,1340)(3867,1325)(3802,1312)
	(3734,1301)(3663,1291)(3589,1284)
	(3513,1279)(3435,1276)(3361,1276)
	(3287,1279)(3214,1283)(3144,1290)
	(3075,1299)(3009,1309)(2945,1320)
	(2883,1333)(2823,1347)(2766,1362)
	(2710,1378)(2656,1395)(2603,1412)
	(2552,1430)(2502,1449)(2453,1469)
	(2406,1489)(2360,1509)(2314,1530)
	(2271,1550)(2228,1571)(2187,1591)
	(2148,1611)(2111,1631)(2075,1649)
	(2042,1667)(2012,1684)(1984,1699)
	(1959,1713)(1937,1726)(1918,1737)
	(1902,1746)(1889,1754)(1879,1760)
	(1871,1764)(1860,1771)
\blacken\path(1977.346,1731.885)(1860.000,1771.000)(1945.133,1681.265)(1977.346,1731.885)
\path(2483,2141)(2481,2144)(2477,2151)
	(2471,2162)(2463,2177)(2455,2194)
	(2447,2213)(2440,2233)(2435,2253)
	(2432,2275)(2433,2298)(2438,2321)
	(2446,2339)(2455,2354)(2464,2366)
	(2472,2374)(2479,2380)(2485,2385)
	(2491,2389)(2496,2392)(2503,2395)
	(2512,2398)(2523,2402)(2537,2406)
	(2554,2409)(2573,2411)(2592,2409)
	(2609,2406)(2623,2402)(2634,2398)
	(2643,2395)(2650,2392)(2656,2388)
	(2661,2385)(2667,2380)(2674,2374)
	(2682,2366)(2691,2354)(2700,2339)
	(2708,2321)(2713,2298)(2714,2275)
	(2711,2253)(2706,2233)(2699,2213)
	(2691,2194)(2683,2177)(2663,2141)
\blacken\path(2695.052,2260.468)(2663.000,2141.000)(2747.502,2231.330)(2695.052,2260.468)
\path(1572,457)(1571,458)(1567,461)
	(1562,466)(1555,474)(1545,483)
	(1535,494)(1524,506)(1514,519)
	(1504,533)(1495,548)(1487,564)
	(1482,580)(1479,598)(1478,617)
	(1482,637)(1489,656)(1499,673)
	(1509,688)(1518,700)(1527,710)
	(1535,718)(1542,725)(1550,731)
	(1557,736)(1565,742)(1575,748)
	(1587,754)(1602,760)(1620,766)
	(1640,770)(1662,772)(1684,770)
	(1704,766)(1722,760)(1737,754)
	(1749,748)(1759,742)(1767,736)
	(1775,731)(1782,725)(1789,718)
	(1797,710)(1806,700)(1815,688)
	(1825,673)(1835,656)(1842,637)
	(1846,617)(1845,598)(1842,580)
	(1837,564)(1829,548)(1820,533)
	(1810,519)(1800,506)(1789,494)
	(1779,483)(1769,474)(1752,457)
\blacken\path(1815.640,563.066)(1752.000,457.000)(1858.066,520.640)(1815.640,563.066)
\path(2490,466)(2489,467)(2485,470)
	(2480,475)(2473,483)(2463,492)
	(2453,503)(2442,515)(2432,528)
	(2422,542)(2413,557)(2405,573)
	(2400,589)(2397,607)(2396,626)
	(2400,646)(2407,665)(2417,682)
	(2427,697)(2436,709)(2445,719)
	(2453,727)(2460,734)(2468,740)
	(2475,745)(2483,751)(2493,757)
	(2505,763)(2520,769)(2538,775)
	(2558,779)(2580,781)(2602,779)
	(2622,775)(2640,769)(2655,763)
	(2667,757)(2677,751)(2685,745)
	(2693,740)(2700,734)(2707,727)
	(2715,719)(2724,709)(2733,697)
	(2743,682)(2753,665)(2760,646)
	(2764,626)(2763,607)(2760,589)
	(2755,573)(2747,557)(2738,542)
	(2728,528)(2718,515)(2707,503)
	(2697,492)(2687,483)(2670,466)
\blacken\path(2733.640,572.066)(2670.000,466.000)(2776.066,529.640)(2733.640,572.066)
\path(4076,452)(4074,454)(4069,459)
	(4061,467)(4050,477)(4038,491)
	(4024,506)(4011,523)(3999,540)
	(3988,559)(3980,578)(3974,600)
	(3972,622)(3975,646)(3982,666)
	(3991,684)(4001,698)(4010,711)
	(4019,721)(4027,728)(4034,735)
	(4042,741)(4049,746)(4058,752)
	(4068,757)(4080,763)(4095,769)
	(4112,775)(4133,779)(4155,781)
	(4177,779)(4197,775)(4215,769)
	(4230,763)(4242,757)(4252,751)
	(4260,745)(4268,740)(4275,734)
	(4282,727)(4290,719)(4299,709)
	(4308,697)(4318,682)(4328,665)
	(4335,646)(4339,626)(4338,607)
	(4335,589)(4330,573)(4322,557)
	(4313,542)(4303,528)(4293,515)
	(4282,503)(4272,492)(4262,483)(4245,466)
\blacken\path(4308.640,572.066)(4245.000,466.000)(4351.066,529.640)(4308.640,572.066)
\path(5865,466)(5863,469)(5859,476)
	(5853,487)(5845,502)(5837,519)
	(5829,538)(5822,558)(5817,578)
	(5814,600)(5815,623)(5820,646)
	(5828,664)(5837,679)(5846,691)
	(5854,699)(5861,705)(5867,710)
	(5873,714)(5878,717)(5885,720)
	(5894,723)(5905,727)(5919,731)
	(5936,734)(5955,736)(5974,734)
	(5991,731)(6005,727)(6016,723)
	(6025,720)(6032,717)(6038,713)
	(6043,710)(6049,705)(6056,699)
	(6064,691)(6073,679)(6082,664)
	(6090,646)(6095,623)(6096,600)
	(6093,578)(6088,558)(6081,538)
	(6073,519)(6065,502)(6045,466)
\blacken\path(6077.052,585.468)(6045.000,466.000)(6129.502,556.330)(6077.052,585.468)
\put(15,1861){\makebox(0,0)[b]{\smash{{\SetFigFont{9}{10.8}{\familydefault}{\mddefault}{\updefault}$K$}}}}
\put(780,196){\makebox(0,0)[b]{\smash{{\SetFigFont{8}{9.6}{\familydefault}{\mddefault}{\updefault}1}}}}
\put(1680,196){\makebox(0,0)[b]{\smash{{\SetFigFont{8}{9.6}{\familydefault}{\mddefault}{\updefault}2}}}}
\put(2580,196){\makebox(0,0)[b]{\smash{{\SetFigFont{8}{9.6}{\familydefault}{\mddefault}{\updefault}3}}}}
\put(1185,286){\makebox(0,0)[b]{\smash{{\SetFigFont{8}{9.6}{\familydefault}{\mddefault}{\updefault}$a$}}}}
\put(2085,286){\makebox(0,0)[b]{\smash{{\SetFigFont{8}{9.6}{\familydefault}{\mddefault}{\updefault}$a$}}}}
\put(2985,286){\makebox(0,0)[b]{\smash{{\SetFigFont{8}{9.6}{\familydefault}{\mddefault}{\updefault}$a$}}}}
\put(3750,286){\makebox(0,0)[b]{\smash{{\SetFigFont{8}{9.6}{\familydefault}{\mddefault}{\updefault}$a$}}}}
\put(15,196){\makebox(0,0)[b]{\smash{{\SetFigFont{9}{10.8}{\familydefault}{\mddefault}{\updefault}$L$}}}}
\put(780,1861){\makebox(0,0)[b]{\smash{{\SetFigFont{8}{9.6}{\familydefault}{\mddefault}{\updefault}1}}}}
\put(1680,1861){\makebox(0,0)[b]{\smash{{\SetFigFont{8}{9.6}{\familydefault}{\mddefault}{\updefault}2}}}}
\put(2580,1861){\makebox(0,0)[b]{\smash{{\SetFigFont{8}{9.6}{\familydefault}{\mddefault}{\updefault}3}}}}
\put(3480,1861){\makebox(0,0)[b]{\smash{{\SetFigFont{8}{9.6}{\familydefault}{\mddefault}{\updefault}4}}}}
\put(1185,1951){\makebox(0,0)[b]{\smash{{\SetFigFont{8}{9.6}{\familydefault}{\mddefault}{\updefault}$a$}}}}
\put(2085,1951){\makebox(0,0)[b]{\smash{{\SetFigFont{8}{9.6}{\familydefault}{\mddefault}{\updefault}$b$}}}}
\put(2985,1951){\makebox(0,0)[b]{\smash{{\SetFigFont{8}{9.6}{\familydefault}{\mddefault}{\updefault}$b$}}}}
\put(3885,1951){\makebox(0,0)[b]{\smash{{\SetFigFont{8}{9.6}{\familydefault}{\mddefault}{\updefault}$b$}}}}
\put(4605,1951){\makebox(0,0)[b]{\smash{{\SetFigFont{8}{9.6}{\familydefault}{\mddefault}{\updefault}$b$}}}}
\put(3480,1321){\makebox(0,0)[b]{\smash{{\SetFigFont{8}{9.6}{\familydefault}{\mddefault}{\updefault}$b$}}}}
\put(5460,1951){\makebox(0,0)[b]{\smash{{\SetFigFont{8}{9.6}{\familydefault}{\mddefault}{\updefault}$a$}}}}
\put(6855,2401){\makebox(0,0)[b]{\smash{{\SetFigFont{8}{9.6}{\familydefault}{\mddefault}{\updefault}$a,b$}}}}
\put(6405,1951){\makebox(0,0)[b]{\smash{{\SetFigFont{8}{9.6}{\familydefault}{\mddefault}{\updefault}$a,b$}}}}
\put(4290,1861){\makebox(0,0)[b]{\smash{{\SetFigFont{8}{9.6}{\familydefault}{\mddefault}{\updefault}$\cdots$}}}}
\put(3413,196){\makebox(0,0)[b]{\smash{{\SetFigFont{8}{9.6}{\familydefault}{\mddefault}{\updefault}$\cdots$}}}}
\put(4560,286){\makebox(0,0)[b]{\smash{{\SetFigFont{8}{9.6}{\familydefault}{\mddefault}{\updefault}$a$}}}}
\put(5460,286){\makebox(0,0)[b]{\smash{{\SetFigFont{8}{9.6}{\familydefault}{\mddefault}{\updefault}$a,b$}}}}
\put(4155,196){\makebox(0,0)[b]{\smash{{\SetFigFont{7}{8.4}{\familydefault}{\mddefault}{\updefault}$n-2$}}}}
\put(6855,1861){\makebox(0,0)[b]{\smash{{\SetFigFont{7}{8.4}{\familydefault}{\mddefault}{\updefault}$m$}}}}
\put(5955,196){\makebox(0,0)[b]{\smash{{\SetFigFont{7}{8.4}{\familydefault}{\mddefault}{\updefault}$n$}}}}
\put(5055,1861){\makebox(0,0)[b]{\smash{{\SetFigFont{6}{7.2}{\familydefault}{\mddefault}{\updefault}$m-2$}}}}
\put(5955,1861){\makebox(0,0)[b]{\smash{{\SetFigFont{6}{7.2}{\familydefault}{\mddefault}{\updefault}$m-1$}}}}
\put(5055,196){\makebox(0,0)[b]{\smash{{\SetFigFont{6}{7.2}{\familydefault}{\mddefault}{\updefault}$n-1$}}}}
\put(2580,2446){\makebox(0,0)[b]{\smash{{\SetFigFont{8}{9.6}{\familydefault}{\mddefault}{\updefault}$a$}}}}
\put(1680,826){\makebox(0,0)[b]{\smash{{\SetFigFont{8}{9.6}{\familydefault}{\mddefault}{\updefault}$b$}}}}
\put(2580,826){\makebox(0,0)[b]{\smash{{\SetFigFont{8}{9.6}{\familydefault}{\mddefault}{\updefault}$b$}}}}
\put(4155,826){\makebox(0,0)[b]{\smash{{\SetFigFont{8}{9.6}{\familydefault}{\mddefault}{\updefault}$b$}}}}
\put(5955,781){\makebox(0,0)[b]{\smash{{\SetFigFont{8}{9.6}{\familydefault}{\mddefault}{\updefault}$a,b$}}}}
\end{picture}
}

%% file: difference_factor_bin.eepic
\setlength{\unitlength}{0.00065617in}
\begingroup\makeatletter\ifx\SetFigFont\undefined%
\gdef\SetFigFont#1#2#3#4#5{%
  \reset@font\fontsize{#1}{#2pt}%
  \fontfamily{#3}\fontseries{#4}\fontshape{#5}%
  \selectfont}%
\fi\endgroup%
{\renewcommand{\dashlinestretch}{30}
\begin{picture}(4527,4701)(0,-10)
\put(1137,3802){\ellipse{402}{402}}
\put(1812,3802){\ellipse{402}{402}}
\put(2487,3802){\ellipse{402}{402}}
\put(3837,3802){\ellipse{402}{402}}
\put(1137,3127){\ellipse{402}{402}}
\put(1812,3127){\ellipse{402}{402}}
\put(2487,3127){\ellipse{402}{402}}
\put(1137,2452){\ellipse{402}{402}}
\put(2487,2452){\ellipse{402}{402}}
\put(3837,2452){\ellipse{402}{402}}
\put(1137,1777){\ellipse{402}{402}}
\put(3837,1777){\ellipse{402}{402}}
\put(1812,1102){\ellipse{402}{402}}
\put(3162,1102){\ellipse{402}{402}}
\put(3837,1102){\ellipse{402}{402}}
\put(1812,427){\ellipse{402}{402}}
\put(2487,427){\ellipse{402}{402}}
\put(3837,427){\ellipse{402}{402}}
\put(462,4477){\ellipse{402}{402}}
\put(3162,3127){\ellipse{402}{402}}
\put(3837,3127){\ellipse{402}{402}}
\put(1812,2452){\ellipse{402}{402}}
\put(2487,1777){\ellipse{402}{402}}
\put(2487,1102){\ellipse{402}{402}}
\put(3162,427){\ellipse{402}{402}}
\put(1812,1777){\ellipse{402}{402}}
\path(1272,1642)(1677,1237)
\blacken\path(1570.934,1300.640)(1677.000,1237.000)(1613.360,1343.066)(1570.934,1300.640)
\path(1947,1642)(2352,1237)
\blacken\path(2245.934,1300.640)(2352.000,1237.000)(2288.360,1343.066)(2245.934,1300.640)
\path(3297,967)(3702,562)
\blacken\path(3595.934,625.640)(3702.000,562.000)(3638.360,668.066)(3595.934,625.640)
\path(2622,967)(3027,562)
\blacken\path(2920.934,625.640)(3027.000,562.000)(2963.360,668.066)(2920.934,625.640)
\path(2622,1642)(3027,1237)
\blacken\path(2920.934,1300.640)(3027.000,1237.000)(2963.360,1343.066)(2920.934,1300.640)
\path(1947,967)(2352,562)
\blacken\path(2245.934,625.640)(2352.000,562.000)(2288.360,668.066)(2245.934,625.640)
\path(1340,3127)(1609,3127)
\blacken\path(1489.000,3097.000)(1609.000,3127.000)(1489.000,3157.000)(1489.000,3097.000)
\path(2015,3127)(2284,3127)
\blacken\path(2164.000,3097.000)(2284.000,3127.000)(2164.000,3157.000)(2164.000,3097.000)
\path(2690,3127)(2959,3127)
\blacken\path(2839.000,3097.000)(2959.000,3127.000)(2839.000,3157.000)(2839.000,3097.000)
\path(3365,3127)(3634,3127)
\blacken\path(3514.000,3097.000)(3634.000,3127.000)(3514.000,3157.000)(3514.000,3097.000)
\path(2015,427)(2284,427)
\blacken\path(2164.000,397.000)(2284.000,427.000)(2164.000,457.000)(2164.000,397.000)
\path(2690,427)(2959,427)
\blacken\path(2839.000,397.000)(2959.000,427.000)(2839.000,457.000)(2839.000,397.000)
\path(3365,427)(3634,427)
\blacken\path(3514.000,397.000)(3634.000,427.000)(3514.000,457.000)(3514.000,397.000)
\path(3837,899)(3837,630)
\blacken\path(3807.000,750.000)(3837.000,630.000)(3867.000,750.000)(3807.000,750.000)
\path(597,4342)(1002,3937)
\blacken\path(895.934,4000.640)(1002.000,3937.000)(938.360,4043.066)(895.934,4000.640)
\dashline{30.000}(3837,3599)(3837,3330)
\blacken\path(3807.000,3450.000)(3837.000,3330.000)(3867.000,3450.000)(3807.000,3450.000)
\dashline{30.000}(3837,2924)(3837,2655)
\blacken\path(3807.000,2775.000)(3837.000,2655.000)(3867.000,2775.000)(3807.000,2775.000)
\dashline{30.000}(3837,2249)(3837,1980)
\blacken\path(3807.000,2100.000)(3837.000,1980.000)(3867.000,2100.000)(3807.000,2100.000)
\path(3837,1574)(3837,1305)
\blacken\path(3807.000,1425.000)(3837.000,1305.000)(3867.000,1425.000)(3807.000,1425.000)
\dashline{45.000}(3297,2992)(3702,2587)
\blacken\path(3595.934,2650.640)(3702.000,2587.000)(3638.360,2693.066)(3595.934,2650.640)
\dashline{30.000}(1137,3599)(1137,3330)
\blacken\path(1107.000,3450.000)(1137.000,3330.000)(1167.000,3450.000)(1107.000,3450.000)
\dashline{30.000}(1137,2924)(1137,2655)
\blacken\path(1107.000,2775.000)(1137.000,2655.000)(1167.000,2775.000)(1107.000,2775.000)
\dashline{30.000}(1137,2249)(1137,1980)
\blacken\path(1107.000,2100.000)(1137.000,1980.000)(1167.000,2100.000)(1107.000,2100.000)
\dashline{30.000}(1812,899)(1812,630)
\blacken\path(1782.000,750.000)(1812.000,630.000)(1842.000,750.000)(1782.000,750.000)
\dashline{30.000}(2487,899)(2487,630)
\blacken\path(2457.000,750.000)(2487.000,630.000)(2517.000,750.000)(2457.000,750.000)
\dashline{30.000}(2487,1574)(2487,1305)
\blacken\path(2457.000,1425.000)(2487.000,1305.000)(2517.000,1425.000)(2457.000,1425.000)
\dashline{30.000}(2487,2249)(2487,1980)
\blacken\path(2457.000,2100.000)(2487.000,1980.000)(2517.000,2100.000)(2457.000,2100.000)
\dashline{30.000}(2487,2924)(2487,2655)
\blacken\path(2457.000,2775.000)(2487.000,2655.000)(2517.000,2775.000)(2457.000,2775.000)
\dashline{30.000}(2487,3599)(2487,3330)
\blacken\path(2457.000,3450.000)(2487.000,3330.000)(2517.000,3450.000)(2457.000,3450.000)
\dashline{30.000}(1812,3599)(1812,3330)
\blacken\path(1782.000,3450.000)(1812.000,3330.000)(1842.000,3450.000)(1782.000,3450.000)
\dashline{30.000}(1812,2924)(1812,2655)
\blacken\path(1782.000,2775.000)(1812.000,2655.000)(1842.000,2775.000)(1782.000,2775.000)
\dashline{30.000}(1812,2249)(1812,1980)
\blacken\path(1782.000,2100.000)(1812.000,1980.000)(1842.000,2100.000)(1782.000,2100.000)
\path(12,4477)(260,4477)
\blacken\path(140.000,4447.000)(260.000,4477.000)(140.000,4507.000)(140.000,4447.000)
\path(4017,3037)(4019,3035)(4024,3031)
	(4032,3025)(4043,3017)(4055,3009)
	(4069,3001)(4083,2994)(4098,2989)
	(4115,2986)(4133,2987)(4152,2992)
	(4167,3000)(4181,3009)(4191,3018)
	(4200,3026)(4206,3033)(4211,3039)
	(4216,3045)(4220,3050)(4224,3057)
	(4228,3066)(4232,3077)(4237,3091)
	(4240,3108)(4242,3127)(4240,3146)
	(4237,3163)(4232,3177)(4228,3188)
	(4224,3197)(4220,3204)(4216,3210)
	(4211,3215)(4206,3221)(4200,3228)
	(4191,3236)(4181,3245)(4167,3254)
	(4152,3262)(4133,3267)(4115,3268)
	(4098,3265)(4083,3260)(4069,3253)
	(4055,3245)(4043,3237)(4017,3217)
\blacken\path(4093.824,3313.944)(4017.000,3217.000)(4130.406,3266.387)(4093.824,3313.944)
\dashline{45.000}(3972,1912)(3973,1913)(3975,1915)
	(3978,1920)(3984,1926)(3992,1935)
	(4002,1948)(4015,1963)(4029,1981)
	(4046,2002)(4065,2026)(4084,2052)
	(4105,2080)(4127,2110)(4149,2141)
	(4171,2174)(4193,2209)(4215,2244)
	(4236,2281)(4256,2319)(4275,2359)
	(4293,2401)(4310,2444)(4326,2490)
	(4340,2538)(4352,2588)(4363,2641)
	(4370,2696)(4375,2754)(4377,2812)
	(4375,2870)(4370,2927)(4363,2981)
	(4352,3033)(4340,3082)(4326,3129)
	(4310,3173)(4293,3214)(4275,3254)
	(4256,3291)(4236,3327)(4215,3362)
	(4193,3395)(4171,3427)(4149,3457)
	(4127,3487)(4105,3514)(4084,3540)
	(4065,3564)(4046,3585)(4029,3604)
	(4015,3621)(4002,3635)(3992,3646)
	(3984,3654)(3972,3667)
\blacken\path(4075.438,3599.172)(3972.000,3667.000)(4031.350,3558.475)(4075.438,3599.172)
\dashline{30.000}(1227,1957)(1228,1958)(1229,1960)
	(1232,1965)(1237,1972)(1243,1981)
	(1251,1994)(1261,2009)(1272,2028)
	(1285,2049)(1299,2072)(1314,2098)
	(1329,2126)(1345,2155)(1361,2186)
	(1377,2219)(1392,2253)(1407,2288)
	(1421,2325)(1434,2364)(1447,2405)
	(1459,2449)(1469,2495)(1478,2544)
	(1486,2596)(1492,2650)(1496,2708)
	(1497,2767)(1496,2826)(1492,2884)
	(1486,2940)(1478,2994)(1469,3044)
	(1459,3092)(1447,3137)(1434,3180)
	(1421,3221)(1407,3261)(1392,3299)
	(1377,3335)(1361,3370)(1345,3404)
	(1329,3436)(1314,3466)(1299,3495)
	(1285,3521)(1272,3544)(1261,3564)
	(1251,3581)(1243,3595)(1237,3606)(1227,3622)
\blacken\path(1316.040,3536.140)(1227.000,3622.000)(1265.160,3504.340)(1316.040,3536.140)
\dashline{30.000}(1928,1965)(1929,1966)(1930,1968)
	(1933,1973)(1938,1980)(1944,1989)
	(1952,2002)(1962,2017)(1973,2036)
	(1986,2057)(2000,2080)(2015,2106)
	(2030,2134)(2046,2163)(2062,2194)
	(2078,2227)(2093,2261)(2108,2296)
	(2122,2333)(2135,2372)(2148,2413)
	(2160,2457)(2170,2503)(2179,2552)
	(2187,2604)(2193,2658)(2197,2716)
	(2198,2775)(2197,2834)(2193,2892)
	(2187,2948)(2179,3002)(2170,3052)
	(2160,3100)(2148,3145)(2135,3188)
	(2122,3229)(2108,3269)(2093,3307)
	(2078,3343)(2062,3378)(2046,3412)
	(2030,3444)(2015,3474)(2000,3503)
	(1986,3529)(1973,3552)(1962,3572)
	(1952,3589)(1944,3603)(1938,3614)(1928,3630)
\blacken\path(2017.040,3544.140)(1928.000,3630.000)(1966.160,3512.340)(2017.040,3544.140)
\dashline{30.000}(2571,1969)(2572,1970)(2573,1972)
	(2576,1977)(2581,1984)(2587,1993)
	(2595,2006)(2605,2021)(2616,2040)
	(2629,2061)(2643,2084)(2658,2110)
	(2673,2138)(2689,2167)(2705,2198)
	(2721,2231)(2736,2265)(2751,2300)
	(2765,2337)(2778,2376)(2791,2417)
	(2803,2461)(2813,2507)(2822,2556)
	(2830,2608)(2836,2662)(2840,2720)
	(2841,2779)(2840,2838)(2836,2896)
	(2830,2952)(2822,3006)(2813,3056)
	(2803,3104)(2791,3149)(2778,3192)
	(2765,3233)(2751,3273)(2736,3311)
	(2721,3347)(2705,3382)(2689,3416)
	(2673,3448)(2658,3478)(2643,3507)
	(2629,3533)(2616,3556)(2605,3576)
	(2595,3593)(2587,3607)(2581,3618)(2571,3634)
\blacken\path(2660.040,3548.140)(2571.000,3634.000)(2609.160,3516.340)(2660.040,3548.140)
\path(1722,247)(1720,245)(1716,240)
	(1710,232)(1702,221)(1694,209)
	(1686,195)(1679,181)(1674,166)
	(1671,149)(1672,131)(1677,112)
	(1685,97)(1694,83)(1703,73)
	(1711,64)(1718,58)(1724,53)
	(1730,48)(1735,44)(1742,40)
	(1751,36)(1762,32)(1776,27)
	(1793,24)(1812,22)(1831,24)
	(1848,27)(1862,32)(1873,36)
	(1882,40)(1889,44)(1895,48)
	(1900,53)(1906,58)(1913,64)
	(1921,73)(1930,83)(1939,97)
	(1947,112)(1952,131)(1953,149)
	(1950,166)(1945,181)(1938,195)
	(1930,209)(1922,221)(1902,247)
\blacken\path(1998.944,170.176)(1902.000,247.000)(1951.387,133.594)(1998.944,170.176)
\path(2390,237)(2388,235)(2384,230)
	(2378,222)(2370,211)(2362,199)
	(2354,185)(2347,171)(2342,156)
	(2339,139)(2340,121)(2345,102)
	(2353,87)(2362,73)(2371,63)
	(2379,54)(2386,48)(2392,43)
	(2398,38)(2403,34)(2410,30)
	(2419,26)(2430,22)(2444,17)
	(2461,14)(2480,12)(2499,14)
	(2516,17)(2530,22)(2541,26)
	(2550,30)(2557,34)(2563,38)
	(2568,43)(2574,48)(2581,54)
	(2589,63)(2598,73)(2607,87)
	(2615,102)(2620,121)(2621,139)
	(2618,156)(2613,171)(2606,185)
	(2598,199)(2590,211)(2570,237)
\blacken\path(2666.944,160.176)(2570.000,237.000)(2619.387,123.594)(2666.944,160.176)
\path(4017,337)(4018,337)(4022,339)
	(4030,341)(4043,345)(4059,350)
	(4077,353)(4095,354)(4114,353)
	(4133,348)(4152,337)(4166,325)
	(4176,312)(4184,301)(4190,291)
	(4194,284)(4197,277)(4200,270)
	(4202,262)(4203,251)(4204,238)
	(4202,221)(4197,202)(4190,188)
	(4182,176)(4174,166)(4167,157)
	(4160,150)(4155,144)(4149,139)
	(4145,134)(4139,130)(4133,126)
	(4126,122)(4118,118)(4107,115)
	(4094,112)(4079,111)(4062,112)
	(4040,119)(4020,129)(4003,143)
	(3987,158)(3973,175)(3961,192)
	(3950,209)(3940,224)(3927,247)
\blacken\path(4012.164,157.294)(3927.000,247.000)(3959.930,127.771)(4012.164,157.294)
\put(462,4432){\makebox(0,0)[b]{\smash{{\SetFigFont{8}{9.6}{\familydefault}{\mddefault}{\updefault}1,1}}}}
\put(1137,3757){\makebox(0,0)[b]{\smash{{\SetFigFont{8}{9.6}{\familydefault}{\mddefault}{\updefault}2,2}}}}
\put(1812,3757){\makebox(0,0)[b]{\smash{{\SetFigFont{8}{9.6}{\familydefault}{\mddefault}{\updefault}2,3}}}}
\put(2487,3757){\makebox(0,0)[b]{\smash{{\SetFigFont{8}{9.6}{\familydefault}{\mddefault}{\updefault}2,4}}}}
\put(3837,3757){\makebox(0,0)[b]{\smash{{\SetFigFont{8}{9.6}{\familydefault}{\mddefault}{\updefault}2,6}}}}
\put(1137,3082){\makebox(0,0)[b]{\smash{{\SetFigFont{8}{9.6}{\familydefault}{\mddefault}{\updefault}3,2}}}}
\put(1812,3082){\makebox(0,0)[b]{\smash{{\SetFigFont{8}{9.6}{\familydefault}{\mddefault}{\updefault}3,3}}}}
\put(2487,3082){\makebox(0,0)[b]{\smash{{\SetFigFont{8}{9.6}{\familydefault}{\mddefault}{\updefault}3,4}}}}
\put(3162,3082){\makebox(0,0)[b]{\smash{{\SetFigFont{8}{9.6}{\familydefault}{\mddefault}{\updefault}3,5}}}}
\put(1137,2407){\makebox(0,0)[b]{\smash{{\SetFigFont{8}{9.6}{\familydefault}{\mddefault}{\updefault}4,2}}}}
\put(1137,1732){\makebox(0,0)[b]{\smash{{\SetFigFont{8}{9.6}{\familydefault}{\mddefault}{\updefault}5,2}}}}
\put(822,4162){\makebox(0,0)[b]{\smash{{\SetFigFont{8}{9.6}{\familydefault}{\mddefault}{\updefault}$a$}}}}
\put(2262,787){\makebox(0,0)[b]{\smash{{\SetFigFont{8}{9.6}{\familydefault}{\mddefault}{\updefault}$a,b$}}}}
\put(2937,787){\makebox(0,0)[b]{\smash{{\SetFigFont{8}{9.6}{\familydefault}{\mddefault}{\updefault}$a,b$}}}}
\put(3522,472){\makebox(0,0)[b]{\smash{{\SetFigFont{8}{9.6}{\familydefault}{\mddefault}{\updefault}$a,b$}}}}
\put(507,2407){\makebox(0,0)[b]{\smash{{\SetFigFont{8}{9.6}{\familydefault}{\mddefault}{\updefault}$m-3$}}}}
\put(507,1732){\makebox(0,0)[b]{\smash{{\SetFigFont{8}{9.6}{\familydefault}{\mddefault}{\updefault}$m-2$}}}}
\put(507,1057){\makebox(0,0)[b]{\smash{{\SetFigFont{8}{9.6}{\familydefault}{\mddefault}{\updefault}$m-1$}}}}
\put(507,382){\makebox(0,0)[b]{\smash{{\SetFigFont{8}{9.6}{\familydefault}{\mddefault}{\updefault}$m$}}}}
\put(3612,742){\makebox(0,0)[b]{\smash{{\SetFigFont{8}{9.6}{\familydefault}{\mddefault}{\updefault}$a,b$}}}}
\put(4017,742){\makebox(0,0)[b]{\smash{{\SetFigFont{8}{9.6}{\familydefault}{\mddefault}{\updefault}$a,b$}}}}
\put(1542,1462){\makebox(0,0)[b]{\smash{{\SetFigFont{8}{9.6}{\familydefault}{\mddefault}{\updefault}$a$}}}}
\put(3972,3442){\makebox(0,0)[b]{\smash{{\SetFigFont{8}{9.6}{\familydefault}{\mddefault}{\updefault}$b$}}}}
\put(3972,2812){\makebox(0,0)[b]{\smash{{\SetFigFont{8}{9.6}{\familydefault}{\mddefault}{\updefault}$b$}}}}
\put(3972,2137){\makebox(0,0)[b]{\smash{{\SetFigFont{8}{9.6}{\familydefault}{\mddefault}{\updefault}$b$}}}}
\put(4512,2767){\makebox(0,0)[b]{\smash{{\SetFigFont{8}{9.6}{\familydefault}{\mddefault}{\updefault}$b$}}}}
\put(3972,1462){\makebox(0,0)[b]{\smash{{\SetFigFont{8}{9.6}{\familydefault}{\mddefault}{\updefault}$a$}}}}
\put(2217,1462){\makebox(0,0)[b]{\smash{{\SetFigFont{8}{9.6}{\familydefault}{\mddefault}{\updefault}$a$}}}}
\put(2892,1462){\makebox(0,0)[b]{\smash{{\SetFigFont{8}{9.6}{\familydefault}{\mddefault}{\updefault}$a$}}}}
\put(2172,472){\makebox(0,0)[b]{\smash{{\SetFigFont{8}{9.6}{\familydefault}{\mddefault}{\updefault}$a$}}}}
\put(2847,472){\makebox(0,0)[b]{\smash{{\SetFigFont{8}{9.6}{\familydefault}{\mddefault}{\updefault}$a$}}}}
\put(4287,382){\makebox(0,0)[b]{\smash{{\SetFigFont{8}{9.6}{\familydefault}{\mddefault}{\updefault}$a,b$}}}}
\put(1812,1057){\makebox(0,0)[b]{\smash{{\SetFigFont{8}{9.6}{\familydefault}{\mddefault}{\updefault}6,3}}}}
\put(1812,382){\makebox(0,0)[b]{\smash{{\SetFigFont{8}{9.6}{\familydefault}{\mddefault}{\updefault}7,3}}}}
\put(3477,3172){\makebox(0,0)[b]{\smash{{\SetFigFont{8}{9.6}{\familydefault}{\mddefault}{\updefault}$a$}}}}
\put(3837,3082){\makebox(0,0)[b]{\smash{{\SetFigFont{8}{9.6}{\familydefault}{\mddefault}{\updefault}3,6}}}}
\put(3837,382){\makebox(0,0)[b]{\smash{{\SetFigFont{8}{9.6}{\familydefault}{\mddefault}{\updefault}7,6}}}}
\end{picture}
}

%% file: binary_bifix_KL.eepic
\setlength{\unitlength}{0.00065617in}
\begingroup\makeatletter\ifx\SetFigFont\undefined%
\gdef\SetFigFont#1#2#3#4#5{%
  \reset@font\fontsize{#1}{#2pt}%
  \fontfamily{#3}\fontseries{#4}\fontshape{#5}%
  \selectfont}%
\fi\endgroup%
{\renewcommand{\dashlinestretch}{30}
\begin{picture}(6749,3615)(0,-10)
\put(5640,777){\ellipse{402}{402}}
\put(4740,777){\ellipse{402}{402}}
\put(3840,777){\ellipse{402}{402}}
\put(2940,777){\ellipse{402}{402}}
\put(1140,777){\ellipse{402}{402}}
\put(2040,777){\ellipse{402}{402}}
\put(4740,777){\ellipse{360}{360}}
\path(690,777)(915,777)
\blacken\path(795.000,747.000)(915.000,777.000)(795.000,807.000)(795.000,747.000)
\path(1365,777)(1815,777)
\blacken\path(1695.000,747.000)(1815.000,777.000)(1695.000,807.000)(1695.000,747.000)
\path(2265,777)(2715,777)
\blacken\path(2595.000,747.000)(2715.000,777.000)(2595.000,807.000)(2595.000,747.000)
\path(3165,777)(3615,777)
\blacken\path(3495.000,747.000)(3615.000,777.000)(3495.000,807.000)(3495.000,747.000)
\path(4065,777)(4515,777)
\blacken\path(4395.000,747.000)(4515.000,777.000)(4395.000,807.000)(4395.000,747.000)
\path(4965,777)(5415,777)
\blacken\path(5295.000,747.000)(5415.000,777.000)(5295.000,807.000)(5295.000,747.000)
\path(3705,642)(3704,641)(3702,639)
	(3697,635)(3690,629)(3681,621)
	(3668,610)(3653,597)(3635,582)
	(3614,565)(3591,546)(3565,526)
	(3538,505)(3510,484)(3479,463)
	(3448,443)(3415,422)(3381,402)
	(3346,383)(3309,365)(3271,349)
	(3230,333)(3187,319)(3142,307)
	(3095,297)(3045,289)(2993,284)
	(2940,282)(2887,284)(2836,289)
	(2787,297)(2741,307)(2697,319)
	(2656,333)(2618,349)(2581,365)
	(2547,383)(2514,402)(2482,422)
	(2452,443)(2423,463)(2396,485)
	(2370,505)(2345,526)(2322,546)
	(2301,565)(2283,582)(2266,597)
	(2253,610)(2241,621)(2233,629)(2220,642)
\blacken\path(2326.066,578.360)(2220.000,642.000)(2283.640,535.934)(2326.066,578.360)
\path(1320,642)(1321,642)(1322,641)
	(1325,639)(1329,636)(1336,633)
	(1345,627)(1356,621)(1370,613)
	(1386,603)(1406,592)(1428,579)
	(1453,565)(1481,550)(1512,533)
	(1545,515)(1581,495)(1619,475)
	(1659,454)(1701,433)(1745,411)
	(1790,389)(1838,366)(1886,344)
	(1936,321)(1987,299)(2040,277)
	(2093,256)(2149,234)(2205,214)
	(2263,194)(2323,174)(2384,155)
	(2448,137)(2513,120)(2580,104)
	(2650,88)(2723,74)(2798,61)
	(2875,49)(2956,38)(3039,29)
	(3124,22)(3211,17)(3300,13)
	(3390,12)(3480,13)(3569,17)
	(3656,22)(3741,29)(3824,38)
	(3905,49)(3982,61)(4057,74)
	(4130,88)(4200,104)(4267,120)
	(4332,137)(4396,155)(4457,174)
	(4517,194)(4575,214)(4631,234)
	(4687,256)(4740,277)(4793,299)
	(4844,321)(4894,344)(4943,366)
	(4990,389)(5035,411)(5079,433)
	(5121,454)(5161,475)(5199,495)
	(5235,515)(5268,533)(5299,550)
	(5327,565)(5352,579)(5374,592)
	(5394,603)(5410,613)(5424,621)
	(5435,627)(5444,633)(5451,636)(5460,642)
\blacken\path(5376.795,550.474)(5460.000,642.000)(5343.513,600.397)(5376.795,550.474)
\path(5550,957)(5548,960)(5544,967)
	(5538,978)(5530,993)(5522,1010)
	(5514,1029)(5507,1049)(5502,1069)
	(5499,1091)(5500,1114)(5505,1137)
	(5513,1155)(5522,1170)(5531,1182)
	(5539,1190)(5546,1196)(5552,1201)
	(5558,1205)(5563,1208)(5570,1211)
	(5579,1214)(5590,1218)(5604,1222)
	(5621,1225)(5640,1227)(5659,1225)
	(5676,1222)(5690,1218)(5701,1214)
	(5710,1211)(5717,1208)(5723,1204)
	(5728,1201)(5734,1196)(5741,1190)
	(5749,1182)(5758,1170)(5767,1155)
	(5775,1137)(5780,1114)(5781,1091)
	(5778,1069)(5773,1049)(5766,1029)
	(5758,1010)(5750,993)(5730,957)
\blacken\path(5762.052,1076.468)(5730.000,957.000)(5814.502,1047.330)(5762.052,1076.468)
\put(5640,2577){\ellipse{402}{402}}
\put(4740,2577){\ellipse{402}{402}}
\put(3840,2577){\ellipse{402}{402}}
\put(2940,2577){\ellipse{402}{402}}
\put(1140,2577){\ellipse{402}{402}}
\put(2040,2577){\ellipse{402}{402}}
\put(6540,2577){\ellipse{402}{402}}
\put(5640,2577){\ellipse{360}{360}}
\path(690,2577)(915,2577)
\blacken\path(795.000,2547.000)(915.000,2577.000)(795.000,2607.000)(795.000,2547.000)
\path(1365,2577)(1815,2577)
\blacken\path(1695.000,2547.000)(1815.000,2577.000)(1695.000,2607.000)(1695.000,2547.000)
\path(2265,2577)(2715,2577)
\blacken\path(2595.000,2547.000)(2715.000,2577.000)(2595.000,2607.000)(2595.000,2547.000)
\path(3165,2577)(3615,2577)
\blacken\path(3495.000,2547.000)(3615.000,2577.000)(3495.000,2607.000)(3495.000,2547.000)
\path(4065,2577)(4515,2577)
\blacken\path(4395.000,2547.000)(4515.000,2577.000)(4395.000,2607.000)(4395.000,2547.000)
\path(4965,2577)(5415,2577)
\blacken\path(5295.000,2547.000)(5415.000,2577.000)(5295.000,2607.000)(5295.000,2547.000)
\path(5865,2577)(6315,2577)
\blacken\path(6195.000,2547.000)(6315.000,2577.000)(6195.000,2607.000)(6195.000,2547.000)
\path(2850,2757)(2848,2760)(2844,2767)
	(2838,2778)(2830,2793)(2822,2810)
	(2814,2829)(2807,2849)(2802,2869)
	(2799,2891)(2800,2914)(2805,2937)
	(2813,2955)(2822,2970)(2831,2982)
	(2839,2990)(2846,2996)(2852,3001)
	(2858,3005)(2863,3008)(2870,3011)
	(2879,3014)(2890,3018)(2904,3022)
	(2921,3025)(2940,3027)(2959,3025)
	(2976,3022)(2990,3018)(3001,3014)
	(3010,3011)(3017,3008)(3023,3004)
	(3028,3001)(3034,2996)(3041,2990)
	(3049,2982)(3058,2970)(3067,2955)
	(3075,2937)(3080,2914)(3081,2891)
	(3078,2869)(3073,2849)(3066,2829)
	(3058,2810)(3050,2793)(3030,2757)
\blacken\path(3062.052,2876.468)(3030.000,2757.000)(3114.502,2847.330)(3062.052,2876.468)
\path(3743,2767)(3741,2770)(3737,2777)
	(3731,2788)(3723,2803)(3715,2820)
	(3707,2839)(3700,2859)(3695,2879)
	(3692,2901)(3693,2924)(3698,2947)
	(3706,2965)(3715,2980)(3724,2992)
	(3732,3000)(3739,3006)(3745,3011)
	(3751,3015)(3756,3018)(3763,3021)
	(3772,3024)(3783,3028)(3797,3032)
	(3814,3035)(3833,3037)(3852,3035)
	(3869,3032)(3883,3028)(3894,3024)
	(3903,3021)(3910,3018)(3916,3014)
	(3921,3011)(3927,3006)(3934,3000)
	(3942,2992)(3951,2980)(3960,2965)
	(3968,2947)(3973,2924)(3974,2901)
	(3971,2879)(3966,2859)(3959,2839)
	(3951,2820)(3943,2803)(3923,2767)
\blacken\path(3955.052,2886.468)(3923.000,2767.000)(4007.502,2857.330)(3955.052,2886.468)
\path(4605,2442)(4604,2441)(4602,2440)
	(4599,2438)(4593,2434)(4585,2429)
	(4574,2421)(4560,2413)(4543,2402)
	(4524,2389)(4501,2375)(4476,2360)
	(4448,2343)(4418,2325)(4385,2305)
	(4351,2286)(4315,2266)(4277,2245)
	(4238,2225)(4198,2204)(4157,2184)
	(4114,2164)(4071,2145)(4026,2126)
	(3979,2109)(3931,2091)(3882,2075)
	(3830,2060)(3777,2046)(3721,2033)
	(3663,2021)(3603,2012)(3541,2003)
	(3477,1997)(3411,1993)(3345,1992)
	(3275,1994)(3207,1998)(3142,2005)
	(3079,2014)(3019,2025)(2962,2037)
	(2909,2052)(2858,2067)(2809,2084)
	(2763,2101)(2719,2120)(2677,2140)
	(2636,2160)(2597,2181)(2560,2202)
	(2523,2223)(2489,2245)(2455,2267)
	(2424,2288)(2394,2309)(2366,2329)
	(2340,2348)(2317,2365)(2296,2381)
	(2278,2396)(2262,2408)(2249,2418)
	(2239,2426)(2231,2433)(2220,2442)
\blacken\path(2331.872,2389.230)(2220.000,2442.000)(2293.878,2342.793)(2331.872,2389.230)
\path(2130,2757)(2131,2758)(2133,2760)
	(2136,2763)(2142,2769)(2149,2777)
	(2160,2788)(2173,2801)(2189,2817)
	(2207,2836)(2229,2857)(2253,2881)
	(2280,2906)(2309,2933)(2340,2962)
	(2374,2991)(2409,3022)(2445,3052)
	(2483,3083)(2522,3114)(2563,3144)
	(2604,3174)(2648,3202)(2692,3230)
	(2738,3257)(2786,3283)(2836,3307)
	(2888,3330)(2942,3351)(2999,3371)
	(3058,3388)(3120,3403)(3185,3415)
	(3251,3424)(3320,3430)(3390,3432)
	(3460,3430)(3529,3424)(3595,3415)
	(3660,3403)(3722,3388)(3781,3371)
	(3838,3351)(3892,3330)(3944,3307)
	(3994,3283)(4042,3257)(4088,3230)
	(4132,3202)(4176,3174)(4217,3144)
	(4258,3114)(4297,3083)(4335,3052)
	(4371,3022)(4406,2991)(4440,2962)
	(4471,2933)(4500,2906)(4527,2881)
	(4551,2857)(4573,2836)(4591,2817)
	(4607,2801)(4620,2788)(4631,2777)
	(4638,2769)(4650,2757)
\blacken\path(4543.934,2820.640)(4650.000,2757.000)(4586.360,2863.066)(4543.934,2820.640)
\path(1320,2487)(1321,2487)(1322,2486)
	(1325,2484)(1329,2482)(1335,2478)
	(1344,2473)(1354,2467)(1368,2459)
	(1384,2449)(1403,2438)(1425,2426)
	(1450,2412)(1478,2396)(1509,2379)
	(1543,2360)(1579,2340)(1618,2319)
	(1659,2297)(1703,2275)(1749,2251)
	(1796,2227)(1846,2202)(1897,2177)
	(1950,2152)(2004,2127)(2060,2102)
	(2117,2077)(2175,2053)(2235,2028)
	(2296,2004)(2359,1981)(2423,1958)
	(2488,1936)(2556,1914)(2625,1893)
	(2696,1873)(2769,1854)(2844,1836)
	(2922,1818)(3003,1802)(3086,1787)
	(3171,1773)(3260,1760)(3351,1749)
	(3445,1740)(3541,1732)(3640,1727)
	(3739,1723)(3840,1722)(3941,1723)
	(4040,1727)(4139,1732)(4235,1740)
	(4329,1749)(4420,1760)(4509,1773)
	(4594,1787)(4677,1802)(4758,1818)
	(4836,1836)(4911,1854)(4984,1873)
	(5055,1893)(5124,1914)(5192,1936)
	(5257,1958)(5321,1981)(5384,2004)
	(5445,2028)(5505,2053)(5563,2077)
	(5620,2102)(5676,2127)(5730,2152)
	(5783,2177)(5834,2202)(5884,2227)
	(5931,2251)(5977,2275)(6021,2297)
	(6062,2319)(6101,2340)(6137,2360)
	(6171,2379)(6202,2396)(6230,2412)
	(6255,2426)(6277,2438)(6296,2449)
	(6312,2459)(6326,2467)(6336,2473)
	(6345,2478)(6351,2482)(6360,2487)
\blacken\path(6269.670,2402.498)(6360.000,2487.000)(6240.532,2454.948)(6269.670,2402.498)
\path(6450,2757)(6448,2760)(6444,2767)
	(6438,2778)(6430,2793)(6422,2810)
	(6414,2829)(6407,2849)(6402,2869)
	(6399,2891)(6400,2914)(6405,2937)
	(6413,2955)(6422,2970)(6431,2982)
	(6439,2990)(6446,2996)(6452,3001)
	(6458,3005)(6463,3008)(6470,3011)
	(6479,3014)(6490,3018)(6504,3022)
	(6521,3025)(6540,3027)(6559,3025)
	(6576,3022)(6590,3018)(6601,3014)
	(6610,3011)(6617,3008)(6623,3004)
	(6628,3001)(6634,2996)(6641,2990)
	(6649,2982)(6658,2970)(6667,2955)
	(6675,2937)(6680,2914)(6681,2891)
	(6678,2869)(6673,2849)(6666,2829)
	(6658,2810)(6650,2793)(6630,2757)
\blacken\path(6662.052,2876.468)(6630.000,2757.000)(6714.502,2847.330)(6662.052,2876.468)
\put(1140,732){\makebox(0,0)[b]{\smash{{\SetFigFont{8}{9.6}{\familydefault}{\mddefault}{\updefault}1}}}}
\put(2040,732){\makebox(0,0)[b]{\smash{{\SetFigFont{8}{9.6}{\familydefault}{\mddefault}{\updefault}2}}}}
\put(2940,732){\makebox(0,0)[b]{\smash{{\SetFigFont{8}{9.6}{\familydefault}{\mddefault}{\updefault}3}}}}
\put(3840,732){\makebox(0,0)[b]{\smash{{\SetFigFont{8}{9.6}{\familydefault}{\mddefault}{\updefault}4}}}}
\put(4740,732){\makebox(0,0)[b]{\smash{{\SetFigFont{8}{9.6}{\familydefault}{\mddefault}{\updefault}5}}}}
\put(5640,732){\makebox(0,0)[b]{\smash{{\SetFigFont{8}{9.6}{\familydefault}{\mddefault}{\updefault}6}}}}
\put(1545,822){\makebox(0,0)[b]{\smash{{\SetFigFont{8}{9.6}{\familydefault}{\mddefault}{\updefault}$a$}}}}
\put(2985,327){\makebox(0,0)[b]{\smash{{\SetFigFont{8}{9.6}{\familydefault}{\mddefault}{\updefault}$b$}}}}
\put(2445,822){\makebox(0,0)[b]{\smash{{\SetFigFont{8}{9.6}{\familydefault}{\mddefault}{\updefault}$a,b$}}}}
\put(3345,822){\makebox(0,0)[b]{\smash{{\SetFigFont{8}{9.6}{\familydefault}{\mddefault}{\updefault}$a,b$}}}}
\put(4245,822){\makebox(0,0)[b]{\smash{{\SetFigFont{8}{9.6}{\familydefault}{\mddefault}{\updefault}$a$}}}}
\put(5145,822){\makebox(0,0)[b]{\smash{{\SetFigFont{8}{9.6}{\familydefault}{\mddefault}{\updefault}$a,b$}}}}
\put(3435,57){\makebox(0,0)[b]{\smash{{\SetFigFont{8}{9.6}{\familydefault}{\mddefault}{\updefault}$b$}}}}
\put(5640,1272){\makebox(0,0)[b]{\smash{{\SetFigFont{8}{9.6}{\familydefault}{\mddefault}{\updefault}$a,b$}}}}
\put(1140,2532){\makebox(0,0)[b]{\smash{{\SetFigFont{8}{9.6}{\familydefault}{\mddefault}{\updefault}1}}}}
\put(2040,2532){\makebox(0,0)[b]{\smash{{\SetFigFont{8}{9.6}{\familydefault}{\mddefault}{\updefault}2}}}}
\put(2940,2532){\makebox(0,0)[b]{\smash{{\SetFigFont{8}{9.6}{\familydefault}{\mddefault}{\updefault}3}}}}
\put(3840,2532){\makebox(0,0)[b]{\smash{{\SetFigFont{8}{9.6}{\familydefault}{\mddefault}{\updefault}4}}}}
\put(4740,2532){\makebox(0,0)[b]{\smash{{\SetFigFont{8}{9.6}{\familydefault}{\mddefault}{\updefault}5}}}}
\put(5640,2532){\makebox(0,0)[b]{\smash{{\SetFigFont{8}{9.6}{\familydefault}{\mddefault}{\updefault}6}}}}
\put(6540,2532){\makebox(0,0)[b]{\smash{{\SetFigFont{8}{9.6}{\familydefault}{\mddefault}{\updefault}7}}}}
\put(1545,2622){\makebox(0,0)[b]{\smash{{\SetFigFont{8}{9.6}{\familydefault}{\mddefault}{\updefault}$a$}}}}
\put(2445,2622){\makebox(0,0)[b]{\smash{{\SetFigFont{8}{9.6}{\familydefault}{\mddefault}{\updefault}$b$}}}}
\put(3345,2622){\makebox(0,0)[b]{\smash{{\SetFigFont{8}{9.6}{\familydefault}{\mddefault}{\updefault}$b$}}}}
\put(4245,2622){\makebox(0,0)[b]{\smash{{\SetFigFont{8}{9.6}{\familydefault}{\mddefault}{\updefault}$b$}}}}
\put(5145,2622){\makebox(0,0)[b]{\smash{{\SetFigFont{8}{9.6}{\familydefault}{\mddefault}{\updefault}$a$}}}}
\put(2940,3072){\makebox(0,0)[b]{\smash{{\SetFigFont{8}{9.6}{\familydefault}{\mddefault}{\updefault}$a$}}}}
\put(3840,3072){\makebox(0,0)[b]{\smash{{\SetFigFont{8}{9.6}{\familydefault}{\mddefault}{\updefault}$a$}}}}
\put(3435,3477){\makebox(0,0)[b]{\smash{{\SetFigFont{8}{9.6}{\familydefault}{\mddefault}{\updefault}$a$}}}}
\put(3390,2082){\makebox(0,0)[b]{\smash{{\SetFigFont{8}{9.6}{\familydefault}{\mddefault}{\updefault}$b$}}}}
\put(6090,2622){\makebox(0,0)[b]{\smash{{\SetFigFont{8}{9.6}{\familydefault}{\mddefault}{\updefault}$a,b$}}}}
\put(3840,1767){\makebox(0,0)[b]{\smash{{\SetFigFont{8}{9.6}{\familydefault}{\mddefault}{\updefault}$b$}}}}
\put(6540,3072){\makebox(0,0)[b]{\smash{{\SetFigFont{8}{9.6}{\familydefault}{\mddefault}{\updefault}$a,b$}}}}
\put(15,2532){\makebox(0,0)[lb]{\smash{{\SetFigFont{9}{10.8}{\familydefault}{\mddefault}{\updefault}$K_7$}}}}
\put(15,732){\makebox(0,0)[lb]{\smash{{\SetFigFont{9}{10.8}{\familydefault}{\mddefault}{\updefault}$L_6$}}}}
\end{picture}
}

%% file: bifix_bin_product.eepic
\setlength{\unitlength}{0.00065617in}
\begingroup\makeatletter\ifx\SetFigFont\undefined%
\gdef\SetFigFont#1#2#3#4#5{%
  \reset@font\fontsize{#1}{#2pt}%
  \fontfamily{#3}\fontseries{#4}\fontshape{#5}%
  \selectfont}%
\fi\endgroup%
{\renewcommand{\dashlinestretch}{30}
\begin{picture}(4527,4824)(0,-10)
\put(1137,3925){\ellipse{402}{402}}
\put(1812,3925){\ellipse{402}{402}}
\put(2487,3925){\ellipse{402}{402}}
\put(3162,3925){\ellipse{402}{402}}
\put(3837,3925){\ellipse{402}{402}}
\put(1137,3250){\ellipse{402}{402}}
\put(1812,3250){\ellipse{402}{402}}
\put(2487,3250){\ellipse{402}{402}}
\put(1137,2575){\ellipse{402}{402}}
\put(2487,2575){\ellipse{402}{402}}
\put(3162,2575){\ellipse{402}{402}}
\put(3837,2575){\ellipse{402}{402}}
\put(1137,1900){\ellipse{402}{402}}
\put(3162,1900){\ellipse{402}{402}}
\put(3837,1900){\ellipse{402}{402}}
\put(1812,1225){\ellipse{402}{402}}
\put(3162,1225){\ellipse{402}{402}}
\put(3837,1225){\ellipse{402}{402}}
\put(1137,550){\ellipse{402}{402}}
\put(1812,550){\ellipse{402}{402}}
\put(2487,550){\ellipse{402}{402}}
\put(3837,550){\ellipse{402}{402}}
\put(462,4600){\ellipse{402}{402}}
\put(3162,3250){\ellipse{402}{402}}
\put(3837,3250){\ellipse{402}{402}}
\put(1812,2575){\ellipse{402}{402}}
\put(2487,1900){\ellipse{402}{402}}
\put(3162,1225){\ellipse{324}{324}}
\put(3162,550){\ellipse{402}{402}}
\put(3162,550){\ellipse{324}{324}}
\put(3162,1900){\ellipse{324}{324}}
\put(3162,2575){\ellipse{324}{324}}
\put(3162,3250){\ellipse{324}{324}}
\put(3162,3925){\ellipse{324}{324}}
\put(2487,1225){\ellipse{324}{324}}
\put(1812,1225){\ellipse{324}{324}}
\put(3837,1225){\ellipse{324}{324}}
\put(1812,1900){\ellipse{402}{402}}
\put(2487,1225){\ellipse{402}{402}}
\path(1272,1765)(1677,1360)
\blacken\path(1570.934,1423.640)(1677.000,1360.000)(1613.360,1466.066)(1570.934,1423.640)
\path(1947,1765)(2352,1360)
\blacken\path(2245.934,1423.640)(2352.000,1360.000)(2288.360,1466.066)(2245.934,1423.640)
\path(3297,1765)(3702,1360)
\blacken\path(3595.934,1423.640)(3702.000,1360.000)(3638.360,1466.066)(3595.934,1423.640)
\path(2622,1090)(3027,685)
\blacken\path(2920.934,748.640)(3027.000,685.000)(2963.360,791.066)(2920.934,748.640)
\path(2622,1765)(3027,1360)
\blacken\path(2920.934,1423.640)(3027.000,1360.000)(2963.360,1466.066)(2920.934,1423.640)
\path(1947,1090)(2352,685)
\blacken\path(2245.934,748.640)(2352.000,685.000)(2288.360,791.066)(2245.934,748.640)
\path(1340,3250)(1609,3250)
\blacken\path(1489.000,3220.000)(1609.000,3250.000)(1489.000,3280.000)(1489.000,3220.000)
\path(1340,2575)(1609,2575)
\blacken\path(1489.000,2545.000)(1609.000,2575.000)(1489.000,2605.000)(1489.000,2545.000)
\path(2015,3250)(2284,3250)
\blacken\path(2164.000,3220.000)(2284.000,3250.000)(2164.000,3280.000)(2164.000,3220.000)
\path(2015,2575)(2284,2575)
\blacken\path(2164.000,2545.000)(2284.000,2575.000)(2164.000,2605.000)(2164.000,2545.000)
\path(2690,3250)(2959,3250)
\blacken\path(2839.000,3220.000)(2959.000,3250.000)(2839.000,3280.000)(2839.000,3220.000)
\path(2690,2575)(2959,2575)
\blacken\path(2839.000,2545.000)(2959.000,2575.000)(2839.000,2605.000)(2839.000,2545.000)
\path(3365,3250)(3634,3250)
\blacken\path(3514.000,3220.000)(3634.000,3250.000)(3514.000,3280.000)(3514.000,3220.000)
\path(3365,2575)(3634,2575)
\blacken\path(3514.000,2545.000)(3634.000,2575.000)(3514.000,2605.000)(3514.000,2545.000)
\dashline{60.000}(1272,2035)(1677,3790)
\blacken\path(1679.249,3666.327)(1677.000,3790.000)(1620.785,3679.819)(1679.249,3666.327)
\dashline{60.000}(1938,2080)(2339,3793)
\blacken\path(2340.859,3669.321)(2339.000,3793.000)(2282.438,3682.997)(2340.859,3669.321)
\path(1272,3790)(1677,2035)
\blacken\path(1620.785,2145.181)(1677.000,2035.000)(1679.249,2158.673)(1620.785,2145.181)
\path(1976,3786)(2352,2035)
\blacken\path(2297.475,2146.027)(2352.000,2035.000)(2356.138,2158.624)(2297.475,2146.027)
\path(2638,3794)(3043,2039)
\blacken\path(2986.785,2149.181)(3043.000,2039.000)(3045.249,2162.673)(2986.785,2149.181)
\dashline{60.000}(3275,2080)(3679,3795)
\blacken\path(3680.686,3671.318)(3679.000,3795.000)(3622.284,3685.076)(3680.686,3671.318)
\path(1340,550)(1609,550)
\blacken\path(1489.000,520.000)(1609.000,550.000)(1489.000,580.000)(1489.000,520.000)
\path(2015,550)(2284,550)
\blacken\path(2164.000,520.000)(2284.000,550.000)(2164.000,580.000)(2164.000,520.000)
\path(2690,550)(2959,550)
\blacken\path(2839.000,520.000)(2959.000,550.000)(2839.000,580.000)(2839.000,520.000)
\path(3365,550)(3634,550)
\blacken\path(3514.000,520.000)(3634.000,550.000)(3514.000,580.000)(3514.000,520.000)
\path(3837,1022)(3837,753)
\blacken\path(3807.000,873.000)(3837.000,753.000)(3867.000,873.000)(3807.000,873.000)
\path(597,4465)(1002,4060)
\blacken\path(895.934,4123.640)(1002.000,4060.000)(938.360,4166.066)(895.934,4123.640)
\dashline{45.000}(1272,3790)(1677,3385)
\blacken\path(1570.934,3448.640)(1677.000,3385.000)(1613.360,3491.066)(1570.934,3448.640)
\dashline{45.000}(1272,3115)(1677,2710)
\blacken\path(1570.934,2773.640)(1677.000,2710.000)(1613.360,2816.066)(1570.934,2773.640)
\dashline{45.000}(1272,2440)(1677,2035)
\blacken\path(1570.934,2098.640)(1677.000,2035.000)(1613.360,2141.066)(1570.934,2098.640)
\dashline{45.000}(1947,3115)(2352,2710)
\blacken\path(2245.934,2773.640)(2352.000,2710.000)(2288.360,2816.066)(2245.934,2773.640)
\dashline{45.000}(1947,2440)(2352,2035)
\blacken\path(2245.934,2098.640)(2352.000,2035.000)(2288.360,2141.066)(2245.934,2098.640)
\dashline{45.000}(1969,3813)(2374,3408)
\blacken\path(2267.934,3471.640)(2374.000,3408.000)(2310.360,3514.066)(2267.934,3471.640)
\dashline{45.000}(2303,3843)(1268,3393)
\blacken\path(1366.087,3468.359)(1268.000,3393.000)(1390.010,3413.335)(1366.087,3468.359)
\dashline{45.000}(2304,3168)(1269,2718)
\blacken\path(1367.087,2793.359)(1269.000,2718.000)(1391.010,2738.335)(1367.087,2793.359)
\dashline{45.000}(2311,2477)(1276,2027)
\blacken\path(1374.087,2102.359)(1276.000,2027.000)(1398.010,2047.335)(1374.087,2102.359)
\dashline{45.000}(2311,1127)(1276,677)
\blacken\path(1374.087,752.359)(1276.000,677.000)(1398.010,697.335)(1374.087,752.359)
\dashline{30.000}(3837,3722)(3837,3453)
\blacken\path(3807.000,3573.000)(3837.000,3453.000)(3867.000,3573.000)(3807.000,3573.000)
\dashline{30.000}(3837,3047)(3837,2778)
\blacken\path(3807.000,2898.000)(3837.000,2778.000)(3867.000,2898.000)(3807.000,2898.000)
\dashline{30.000}(3837,2372)(3837,2103)
\blacken\path(3807.000,2223.000)(3837.000,2103.000)(3867.000,2223.000)(3807.000,2223.000)
\path(3837,1697)(3837,1428)
\blacken\path(3807.000,1548.000)(3837.000,1428.000)(3867.000,1548.000)(3807.000,1548.000)
\dashline{45.000}(3297,3790)(3702,3385)
\blacken\path(3595.934,3448.640)(3702.000,3385.000)(3638.360,3491.066)(3595.934,3448.640)
\dashline{45.000}(3297,3115)(3702,2710)
\blacken\path(3595.934,2773.640)(3702.000,2710.000)(3638.360,2816.066)(3595.934,2773.640)
\dashline{45.000}(3297,2440)(3702,2035)
\blacken\path(3595.934,2098.640)(3702.000,2035.000)(3638.360,2141.066)(3595.934,2098.640)
\path(12,4600)(237,4600)
\blacken\path(117.000,4570.000)(237.000,4600.000)(117.000,4630.000)(117.000,4570.000)
\dashline{60.000}(2667,1990)(2802,2170)(2757,4330)
	(1272,4330)(1182,4105)
\path(1198.713,4227.559)(1182.000,4105.000)(1254.421,4205.275)
\path(3319,1113)(3724,708)
\blacken\path(3617.934,771.640)(3724.000,708.000)(3660.360,814.066)(3617.934,771.640)
\dottedline{30}(2847,1495)(2847,145)(3522,145)
	(3522,910)(4242,910)(4242,1450)
	(4242,1495)(2847,1495)
\path(4017,3160)(4019,3158)(4024,3154)
	(4032,3148)(4043,3140)(4055,3132)
	(4069,3124)(4083,3117)(4098,3112)
	(4115,3109)(4133,3110)(4152,3115)
	(4167,3123)(4181,3132)(4191,3141)
	(4200,3149)(4206,3156)(4211,3162)
	(4216,3168)(4220,3173)(4224,3180)
	(4228,3189)(4232,3200)(4237,3214)
	(4240,3231)(4242,3250)(4240,3269)
	(4237,3286)(4232,3300)(4228,3311)
	(4224,3320)(4220,3327)(4216,3333)
	(4211,3338)(4206,3344)(4200,3351)
	(4191,3359)(4181,3368)(4167,3377)
	(4152,3385)(4133,3390)(4115,3391)
	(4098,3388)(4083,3383)(4069,3376)
	(4055,3368)(4043,3360)(4017,3340)
\blacken\path(4093.824,3436.944)(4017.000,3340.000)(4130.406,3389.387)(4093.824,3436.944)
\path(4027,2492)(4029,2490)(4034,2486)
	(4042,2480)(4053,2472)(4065,2464)
	(4079,2456)(4093,2449)(4108,2444)
	(4125,2441)(4143,2442)(4162,2447)
	(4177,2455)(4191,2464)(4201,2473)
	(4210,2481)(4216,2488)(4221,2494)
	(4226,2500)(4230,2505)(4234,2512)
	(4238,2521)(4242,2532)(4247,2546)
	(4250,2563)(4252,2582)(4250,2601)
	(4247,2618)(4242,2632)(4238,2643)
	(4234,2652)(4230,2659)(4226,2665)
	(4221,2670)(4216,2676)(4210,2683)
	(4201,2691)(4191,2700)(4177,2709)
	(4162,2717)(4143,2722)(4125,2723)
	(4108,2720)(4093,2715)(4079,2708)
	(4065,2700)(4053,2692)(4027,2672)
\blacken\path(4103.824,2768.944)(4027.000,2672.000)(4140.406,2721.387)(4103.824,2768.944)
\dashline{45.000}(2397,370)(2396,369)(2393,367)
	(2388,364)(2380,360)(2369,353)
	(2355,345)(2338,335)(2318,324)
	(2296,312)(2272,300)(2247,287)
	(2219,274)(2191,262)(2161,250)
	(2130,238)(2098,227)(2063,218)
	(2027,209)(1989,201)(1948,196)
	(1904,191)(1859,189)(1812,190)
	(1766,193)(1721,199)(1679,206)
	(1639,215)(1603,226)(1569,237)
	(1537,250)(1507,263)(1479,277)
	(1452,291)(1427,306)(1403,321)
	(1380,336)(1359,350)(1340,363)
	(1323,376)(1308,387)(1296,396)
	(1287,403)(1272,415)
\blacken\path(1384.445,363.463)(1272.000,415.000)(1346.963,316.611)(1384.445,363.463)
\dashline{45.000}(3972,2035)(3973,2036)(3975,2038)
	(3978,2043)(3984,2049)(3992,2058)
	(4002,2071)(4015,2086)(4029,2104)
	(4046,2125)(4065,2149)(4084,2175)
	(4105,2203)(4127,2233)(4149,2264)
	(4171,2297)(4193,2332)(4215,2367)
	(4236,2404)(4256,2442)(4275,2482)
	(4293,2524)(4310,2567)(4326,2613)
	(4340,2661)(4352,2711)(4363,2764)
	(4370,2819)(4375,2877)(4377,2935)
	(4375,2993)(4370,3050)(4363,3104)
	(4352,3156)(4340,3205)(4326,3252)
	(4310,3296)(4293,3337)(4275,3377)
	(4256,3414)(4236,3450)(4215,3485)
	(4193,3518)(4171,3550)(4149,3580)
	(4127,3610)(4105,3637)(4084,3663)
	(4065,3687)(4046,3708)(4029,3727)
	(4015,3744)(4002,3758)(3992,3769)
	(3984,3777)(3972,3790)
\blacken\path(4075.438,3722.172)(3972.000,3790.000)(4031.350,3681.475)(4075.438,3722.172)
\put(462,4555){\makebox(0,0)[b]{\smash{{\SetFigFont{8}{9.6}{\familydefault}{\mddefault}{\updefault}1,1}}}}
\put(1137,3880){\makebox(0,0)[b]{\smash{{\SetFigFont{8}{9.6}{\familydefault}{\mddefault}{\updefault}2,2}}}}
\put(1812,3880){\makebox(0,0)[b]{\smash{{\SetFigFont{8}{9.6}{\familydefault}{\mddefault}{\updefault}2,3}}}}
\put(2487,3880){\makebox(0,0)[b]{\smash{{\SetFigFont{8}{9.6}{\familydefault}{\mddefault}{\updefault}2,4}}}}
\put(3162,3880){\makebox(0,0)[b]{\smash{{\SetFigFont{8}{9.6}{\familydefault}{\mddefault}{\updefault}2,5}}}}
\put(3837,3880){\makebox(0,0)[b]{\smash{{\SetFigFont{8}{9.6}{\familydefault}{\mddefault}{\updefault}2,6}}}}
\put(1137,3205){\makebox(0,0)[b]{\smash{{\SetFigFont{8}{9.6}{\familydefault}{\mddefault}{\updefault}3,2}}}}
\put(1812,3205){\makebox(0,0)[b]{\smash{{\SetFigFont{8}{9.6}{\familydefault}{\mddefault}{\updefault}3,3}}}}
\put(2487,3205){\makebox(0,0)[b]{\smash{{\SetFigFont{8}{9.6}{\familydefault}{\mddefault}{\updefault}3,4}}}}
\put(3162,3205){\makebox(0,0)[b]{\smash{{\SetFigFont{8}{9.6}{\familydefault}{\mddefault}{\updefault}3,5}}}}
\put(1137,2530){\makebox(0,0)[b]{\smash{{\SetFigFont{8}{9.6}{\familydefault}{\mddefault}{\updefault}4,2}}}}
\put(1137,1855){\makebox(0,0)[b]{\smash{{\SetFigFont{8}{9.6}{\familydefault}{\mddefault}{\updefault}5,2}}}}
\put(1137,505){\makebox(0,0)[b]{\smash{{\SetFigFont{8}{9.6}{\familydefault}{\mddefault}{\updefault}7,2}}}}
\put(3162,1855){\makebox(0,0)[b]{\smash{{\SetFigFont{8}{9.6}{\familydefault}{\mddefault}{\updefault}5,5}}}}
\put(3837,3205){\makebox(0,0)[b]{\smash{{\SetFigFont{8}{9.6}{\familydefault}{\mddefault}{\updefault}3,6}}}}
\put(822,4285){\makebox(0,0)[b]{\smash{{\SetFigFont{8}{9.6}{\familydefault}{\mddefault}{\updefault}$a$}}}}
\put(1497,595){\makebox(0,0)[b]{\smash{{\SetFigFont{8}{9.6}{\familydefault}{\mddefault}{\updefault}$a,b$}}}}
\put(2172,595){\makebox(0,0)[b]{\smash{{\SetFigFont{8}{9.6}{\familydefault}{\mddefault}{\updefault}$a,b$}}}}
\put(2847,595){\makebox(0,0)[b]{\smash{{\SetFigFont{8}{9.6}{\familydefault}{\mddefault}{\updefault}$a,b$}}}}
\put(3522,595){\makebox(0,0)[b]{\smash{{\SetFigFont{8}{9.6}{\familydefault}{\mddefault}{\updefault}$a,b$}}}}
\put(507,2530){\makebox(0,0)[b]{\smash{{\SetFigFont{8}{9.6}{\familydefault}{\mddefault}{\updefault}$m-3$}}}}
\put(507,1855){\makebox(0,0)[b]{\smash{{\SetFigFont{8}{9.6}{\familydefault}{\mddefault}{\updefault}$m-2$}}}}
\put(507,1180){\makebox(0,0)[b]{\smash{{\SetFigFont{8}{9.6}{\familydefault}{\mddefault}{\updefault}$m-1$}}}}
\put(507,505){\makebox(0,0)[b]{\smash{{\SetFigFont{8}{9.6}{\familydefault}{\mddefault}{\updefault}$m$}}}}
\put(4017,865){\makebox(0,0)[b]{\smash{{\SetFigFont{8}{9.6}{\familydefault}{\mddefault}{\updefault}$a,b$}}}}
\put(1992,4420){\makebox(0,0)[b]{\smash{{\SetFigFont{8}{9.6}{\familydefault}{\mddefault}{\updefault}$b$}}}}
\put(1452,3655){\makebox(0,0)[b]{\smash{{\SetFigFont{8}{9.6}{\familydefault}{\mddefault}{\updefault}$b$}}}}
\put(1542,1585){\makebox(0,0)[b]{\smash{{\SetFigFont{8}{9.6}{\familydefault}{\mddefault}{\updefault}$a$}}}}
\put(3972,3565){\makebox(0,0)[b]{\smash{{\SetFigFont{8}{9.6}{\familydefault}{\mddefault}{\updefault}$b$}}}}
\put(3972,2935){\makebox(0,0)[b]{\smash{{\SetFigFont{8}{9.6}{\familydefault}{\mddefault}{\updefault}$b$}}}}
\put(3972,2260){\makebox(0,0)[b]{\smash{{\SetFigFont{8}{9.6}{\familydefault}{\mddefault}{\updefault}$b$}}}}
\put(4512,2890){\makebox(0,0)[b]{\smash{{\SetFigFont{8}{9.6}{\familydefault}{\mddefault}{\updefault}$b$}}}}
\put(3972,1585){\makebox(0,0)[b]{\smash{{\SetFigFont{8}{9.6}{\familydefault}{\mddefault}{\updefault}$a$}}}}
\put(2217,1585){\makebox(0,0)[b]{\smash{{\SetFigFont{8}{9.6}{\familydefault}{\mddefault}{\updefault}$a$}}}}
\put(2892,1585){\makebox(0,0)[b]{\smash{{\SetFigFont{8}{9.6}{\familydefault}{\mddefault}{\updefault}$a$}}}}
\put(3567,1585){\makebox(0,0)[b]{\smash{{\SetFigFont{8}{9.6}{\familydefault}{\mddefault}{\updefault}$a$}}}}
\put(1632,865){\makebox(0,0)[b]{\smash{{\SetFigFont{8}{9.6}{\familydefault}{\mddefault}{\updefault}$b$}}}}
\put(1857,55){\makebox(0,0)[b]{\smash{{\SetFigFont{8}{9.6}{\familydefault}{\mddefault}{\updefault}$b$}}}}
\put(3342,865){\makebox(0,0)[b]{\smash{{\SetFigFont{8}{9.6}{\familydefault}{\mddefault}{\updefault}$a,b$}}}}
\put(2667,865){\makebox(0,0)[b]{\smash{{\SetFigFont{8}{9.6}{\familydefault}{\mddefault}{\updefault}$a$}}}}
\put(2262,865){\makebox(0,0)[b]{\smash{{\SetFigFont{8}{9.6}{\familydefault}{\mddefault}{\updefault}$a,b$}}}}
\end{picture}
}

%% file: factor_bin_KL.eepic
\setlength{\unitlength}{0.00065617in}
\begingroup\makeatletter\ifx\SetFigFont\undefined%
\gdef\SetFigFont#1#2#3#4#5{%
  \reset@font\fontsize{#1}{#2pt}%
  \fontfamily{#3}\fontseries{#4}\fontshape{#5}%
  \selectfont}%
\fi\endgroup%
{\renewcommand{\dashlinestretch}{30}
\begin{picture}(6287,2269)(0,-10)
\put(780,1591){\ellipse{468}{468}}
\put(1680,1591){\ellipse{468}{468}}
\put(2580,1591){\ellipse{468}{468}}
\put(4245,1591){\ellipse{468}{468}}
\put(780,241){\ellipse{468}{468}}
\put(1680,241){\ellipse{468}{468}}
\put(2580,241){\ellipse{468}{468}}
\put(4245,241){\ellipse{468}{468}}
\put(5145,1591){\ellipse{406}{406}}
\put(6045,1591){\ellipse{468}{468}}
\put(5145,241){\ellipse{468}{468}}
\put(5145,241){\ellipse{406}{406}}
\put(6045,241){\ellipse{468}{468}}
\put(5145,1591){\ellipse{468}{468}}
\path(2805,241)(3255,241)
\blacken\path(3135.000,211.000)(3255.000,241.000)(3135.000,271.000)(3135.000,211.000)
\path(1005,241)(1433,241)
\blacken\path(1313.000,211.000)(1433.000,241.000)(1313.000,271.000)(1313.000,211.000)
\path(285,1591)(533,1591)
\blacken\path(413.000,1561.000)(533.000,1591.000)(413.000,1621.000)(413.000,1561.000)
\path(2805,1591)(3255,1591)
\blacken\path(3135.000,1561.000)(3255.000,1591.000)(3135.000,1621.000)(3135.000,1561.000)
\path(1005,1591)(1433,1591)
\blacken\path(1313.000,1561.000)(1433.000,1591.000)(1313.000,1621.000)(1313.000,1561.000)
\path(4470,241)(4897,241)
\blacken\path(4777.000,211.000)(4897.000,241.000)(4777.000,271.000)(4777.000,211.000)
\path(5370,1591)(5797,1591)
\blacken\path(5677.000,1561.000)(5797.000,1591.000)(5677.000,1621.000)(5677.000,1561.000)
\path(5370,241)(5797,241)
\blacken\path(5677.000,211.000)(5797.000,241.000)(5677.000,271.000)(5677.000,211.000)
\path(4470,1591)(4897,1591)
\blacken\path(4777.000,1561.000)(4897.000,1591.000)(4777.000,1621.000)(4777.000,1561.000)
\path(3615,1591)(3997,1591)
\blacken\path(3877.000,1561.000)(3997.000,1591.000)(3877.000,1621.000)(3877.000,1561.000)
\path(3615,241)(3997,241)
\blacken\path(3877.000,211.000)(3997.000,241.000)(3877.000,271.000)(3877.000,211.000)
\path(1905,1591)(2333,1591)
\blacken\path(2213.000,1561.000)(2333.000,1591.000)(2213.000,1621.000)(2213.000,1561.000)
\path(1905,241)(2333,241)
\blacken\path(2213.000,211.000)(2333.000,241.000)(2213.000,271.000)(2213.000,211.000)
\path(285,241)(533,241)
\blacken\path(413.000,211.000)(533.000,241.000)(413.000,271.000)(413.000,211.000)
\path(5955,466)(5953,469)(5949,476)
	(5943,487)(5935,502)(5927,519)
	(5919,538)(5912,558)(5907,578)
	(5904,600)(5905,623)(5910,646)
	(5918,664)(5927,679)(5936,691)
	(5944,699)(5951,705)(5957,710)
	(5963,714)(5968,717)(5975,720)
	(5984,723)(5995,727)(6009,731)
	(6026,734)(6045,736)(6064,734)
	(6081,731)(6095,727)(6106,723)
	(6115,720)(6122,717)(6128,713)
	(6133,710)(6139,705)(6146,699)
	(6154,691)(6163,679)(6172,664)
	(6180,646)(6185,623)(6186,600)
	(6183,578)(6178,558)(6171,538)
	(6163,519)(6155,502)(6135,466)
\blacken\path(6167.052,585.468)(6135.000,466.000)(6219.502,556.330)(6167.052,585.468)
\path(4155,466)(4153,469)(4149,476)
	(4143,487)(4135,502)(4127,519)
	(4119,538)(4112,558)(4107,578)
	(4104,600)(4105,623)(4110,646)
	(4118,664)(4127,679)(4136,691)
	(4144,699)(4151,705)(4157,710)
	(4163,714)(4168,717)(4175,720)
	(4184,723)(4195,727)(4209,731)
	(4226,734)(4245,736)(4264,734)
	(4281,731)(4295,727)(4306,723)
	(4315,720)(4322,717)(4328,713)
	(4333,710)(4339,705)(4346,699)
	(4354,691)(4363,679)(4372,664)
	(4380,646)(4385,623)(4386,600)
	(4383,578)(4378,558)(4371,538)
	(4363,519)(4355,502)(4335,466)
\blacken\path(4367.052,585.468)(4335.000,466.000)(4419.502,556.330)(4367.052,585.468)
\path(5955,1816)(5953,1819)(5949,1826)
	(5943,1837)(5935,1852)(5927,1869)
	(5919,1888)(5912,1908)(5907,1928)
	(5904,1950)(5905,1973)(5910,1996)
	(5917,2012)(5925,2026)(5933,2037)
	(5940,2045)(5947,2052)(5952,2057)
	(5958,2060)(5963,2064)(5968,2066)
	(5974,2069)(5981,2072)(5989,2075)
	(6000,2079)(6013,2082)(6028,2085)
	(6045,2086)(6062,2085)(6077,2082)
	(6090,2079)(6101,2075)(6109,2072)
	(6116,2069)(6122,2066)(6128,2064)
	(6132,2060)(6138,2057)(6143,2052)
	(6150,2045)(6157,2037)(6165,2026)
	(6173,2012)(6180,1996)(6185,1973)
	(6186,1950)(6183,1928)(6178,1908)
	(6171,1888)(6163,1869)(6155,1852)(6135,1816)
\blacken\path(6167.052,1935.468)(6135.000,1816.000)(6219.502,1906.330)(6167.052,1935.468)
\path(4155,1816)(4153,1819)(4149,1826)
	(4143,1837)(4135,1852)(4127,1869)
	(4119,1888)(4112,1908)(4107,1928)
	(4104,1950)(4105,1973)(4110,1996)
	(4117,2012)(4125,2026)(4133,2037)
	(4140,2045)(4147,2052)(4152,2057)
	(4158,2060)(4163,2064)(4168,2066)
	(4174,2069)(4181,2072)(4189,2075)
	(4200,2079)(4213,2082)(4228,2085)
	(4245,2086)(4262,2085)(4277,2082)
	(4290,2079)(4301,2075)(4309,2072)
	(4316,2069)(4322,2066)(4328,2064)
	(4332,2060)(4338,2057)(4343,2052)
	(4350,2045)(4357,2037)(4365,2026)
	(4373,2012)(4380,1996)(4385,1973)
	(4386,1950)(4383,1928)(4378,1908)
	(4371,1888)(4363,1869)(4355,1852)(4335,1816)
\blacken\path(4367.052,1935.468)(4335.000,1816.000)(4419.502,1906.330)(4367.052,1935.468)
\path(2490,1816)(2488,1819)(2484,1826)
	(2478,1837)(2470,1852)(2462,1869)
	(2454,1888)(2447,1908)(2442,1928)
	(2439,1950)(2440,1973)(2445,1996)
	(2452,2012)(2460,2026)(2468,2037)
	(2475,2045)(2482,2052)(2487,2057)
	(2493,2060)(2498,2064)(2503,2066)
	(2509,2069)(2516,2072)(2524,2075)
	(2535,2079)(2548,2082)(2563,2085)
	(2580,2086)(2597,2085)(2612,2082)
	(2625,2079)(2636,2075)(2644,2072)
	(2651,2069)(2657,2066)(2663,2064)
	(2667,2060)(2673,2057)(2678,2052)
	(2685,2045)(2692,2037)(2700,2026)
	(2708,2012)(2715,1996)(2720,1973)
	(2721,1950)(2718,1928)(2713,1908)
	(2706,1888)(2698,1869)(2690,1852)(2670,1816)
\blacken\path(2702.052,1935.468)(2670.000,1816.000)(2754.502,1906.330)(2702.052,1935.468)
\path(2490,466)(2488,469)(2484,476)
	(2478,487)(2470,502)(2462,519)
	(2454,538)(2447,558)(2442,578)
	(2439,600)(2440,623)(2445,646)
	(2453,664)(2462,679)(2471,691)
	(2479,699)(2486,705)(2492,710)
	(2498,714)(2503,717)(2510,720)
	(2519,723)(2530,727)(2544,731)
	(2561,734)(2580,736)(2599,734)
	(2616,731)(2630,727)(2641,723)
	(2650,720)(2657,717)(2663,713)
	(2668,710)(2674,705)(2681,699)
	(2689,691)(2698,679)(2707,664)
	(2715,646)(2720,623)(2721,600)
	(2718,578)(2713,558)(2706,538)
	(2698,519)(2690,502)(2670,466)
\blacken\path(2702.052,585.468)(2670.000,466.000)(2754.502,556.330)(2702.052,585.468)
\path(1590,466)(1588,469)(1584,476)
	(1578,487)(1570,502)(1562,519)
	(1554,538)(1547,558)(1542,578)
	(1539,600)(1540,623)(1545,646)
	(1553,664)(1562,679)(1571,691)
	(1579,699)(1586,705)(1592,710)
	(1598,714)(1603,717)(1610,720)
	(1619,723)(1630,727)(1644,731)
	(1661,734)(1680,736)(1699,734)
	(1716,731)(1730,727)(1741,723)
	(1750,720)(1757,717)(1763,713)
	(1768,710)(1774,705)(1781,699)
	(1789,691)(1798,679)(1807,664)
	(1815,646)(1820,623)(1821,600)
	(1818,578)(1813,558)(1806,538)
	(1798,519)(1790,502)(1770,466)
\blacken\path(1802.052,585.468)(1770.000,466.000)(1854.502,556.330)(1802.052,585.468)
\path(1590,1816)(1588,1819)(1584,1826)
	(1578,1837)(1570,1852)(1562,1869)
	(1554,1888)(1547,1908)(1542,1928)
	(1539,1950)(1540,1973)(1545,1996)
	(1552,2012)(1560,2026)(1568,2037)
	(1575,2045)(1582,2052)(1587,2057)
	(1593,2060)(1598,2064)(1603,2066)
	(1609,2069)(1616,2072)(1624,2075)
	(1635,2079)(1648,2082)(1663,2085)
	(1680,2086)(1697,2085)(1712,2082)
	(1725,2079)(1736,2075)(1744,2072)
	(1751,2069)(1757,2066)(1763,2064)
	(1767,2060)(1773,2057)(1778,2052)
	(1785,2045)(1792,2037)(1800,2026)
	(1808,2012)(1815,1996)(1820,1973)
	(1821,1950)(1818,1928)(1813,1908)
	(1806,1888)(1798,1869)(1790,1852)(1770,1816)
\blacken\path(1802.052,1935.468)(1770.000,1816.000)(1854.502,1906.330)(1802.052,1935.468)
\put(5595,286){\makebox(0,0)[b]{\smash{{\SetFigFont{8}{9.6}{\familydefault}{\mddefault}{\updefault}$a,b$}}}}
\put(3458,196){\makebox(0,0)[b]{\smash{{\SetFigFont{8}{9.6}{\familydefault}{\mddefault}{\updefault}$\cdots$}}}}
\put(2580,196){\makebox(0,0)[b]{\smash{{\SetFigFont{8}{9.6}{\familydefault}{\mddefault}{\updefault}3}}}}
\put(1680,196){\makebox(0,0)[b]{\smash{{\SetFigFont{8}{9.6}{\familydefault}{\mddefault}{\updefault}2}}}}
\put(780,196){\makebox(0,0)[b]{\smash{{\SetFigFont{8}{9.6}{\familydefault}{\mddefault}{\updefault}1}}}}
\put(1185,286){\makebox(0,0)[b]{\smash{{\SetFigFont{8}{9.6}{\familydefault}{\mddefault}{\updefault}$a,b$}}}}
\put(2085,286){\makebox(0,0)[b]{\smash{{\SetFigFont{8}{9.6}{\familydefault}{\mddefault}{\updefault}$b$}}}}
\put(2985,286){\makebox(0,0)[b]{\smash{{\SetFigFont{8}{9.6}{\familydefault}{\mddefault}{\updefault}$b$}}}}
\put(3795,286){\makebox(0,0)[b]{\smash{{\SetFigFont{8}{9.6}{\familydefault}{\mddefault}{\updefault}$b$}}}}
\put(4650,286){\makebox(0,0)[b]{\smash{{\SetFigFont{8}{9.6}{\familydefault}{\mddefault}{\updefault}$b$}}}}
\put(4245,781){\makebox(0,0)[b]{\smash{{\SetFigFont{8}{9.6}{\familydefault}{\mddefault}{\updefault}$a$}}}}
\put(6045,196){\makebox(0,0)[b]{\smash{{\SetFigFont{8}{9.6}{\familydefault}{\mddefault}{\updefault}$n$}}}}
\put(15,196){\makebox(0,0)[b]{\smash{{\SetFigFont{9}{10.8}{\familydefault}{\mddefault}{\updefault}$L$}}}}
\put(4650,1636){\makebox(0,0)[b]{\smash{{\SetFigFont{8}{9.6}{\familydefault}{\mddefault}{\updefault}$a$}}}}
\put(5595,1636){\makebox(0,0)[b]{\smash{{\SetFigFont{8}{9.6}{\familydefault}{\mddefault}{\updefault}$a,b$}}}}
\put(6045,1546){\makebox(0,0)[b]{\smash{{\SetFigFont{8}{9.6}{\familydefault}{\mddefault}{\updefault}$m$}}}}
\put(15,1546){\makebox(0,0)[b]{\smash{{\SetFigFont{9}{10.8}{\familydefault}{\mddefault}{\updefault}$K$}}}}
\put(3458,1546){\makebox(0,0)[b]{\smash{{\SetFigFont{8}{9.6}{\familydefault}{\mddefault}{\updefault}$\cdots$}}}}
\put(2580,1546){\makebox(0,0)[b]{\smash{{\SetFigFont{8}{9.6}{\familydefault}{\mddefault}{\updefault}3}}}}
\put(1680,1546){\makebox(0,0)[b]{\smash{{\SetFigFont{8}{9.6}{\familydefault}{\mddefault}{\updefault}2}}}}
\put(780,1546){\makebox(0,0)[b]{\smash{{\SetFigFont{8}{9.6}{\familydefault}{\mddefault}{\updefault}1}}}}
\put(1185,1636){\makebox(0,0)[b]{\smash{{\SetFigFont{8}{9.6}{\familydefault}{\mddefault}{\updefault}$a$}}}}
\put(2085,1636){\makebox(0,0)[b]{\smash{{\SetFigFont{8}{9.6}{\familydefault}{\mddefault}{\updefault}$a$}}}}
\put(2985,1636){\makebox(0,0)[b]{\smash{{\SetFigFont{8}{9.6}{\familydefault}{\mddefault}{\updefault}$a$}}}}
\put(3795,1636){\makebox(0,0)[b]{\smash{{\SetFigFont{8}{9.6}{\familydefault}{\mddefault}{\updefault}$a$}}}}
\put(4245,2131){\makebox(0,0)[b]{\smash{{\SetFigFont{8}{9.6}{\familydefault}{\mddefault}{\updefault}$b$}}}}
\put(6045,781){\makebox(0,0)[b]{\smash{{\SetFigFont{8}{9.6}{\familydefault}{\mddefault}{\updefault}$a,b$}}}}
\put(6045,2131){\makebox(0,0)[b]{\smash{{\SetFigFont{8}{9.6}{\familydefault}{\mddefault}{\updefault}$a,b$}}}}
\put(2580,2131){\makebox(0,0)[b]{\smash{{\SetFigFont{8}{9.6}{\familydefault}{\mddefault}{\updefault}$b$}}}}
\put(2580,781){\makebox(0,0)[b]{\smash{{\SetFigFont{8}{9.6}{\familydefault}{\mddefault}{\updefault}$a$}}}}
\put(1680,781){\makebox(0,0)[b]{\smash{{\SetFigFont{8}{9.6}{\familydefault}{\mddefault}{\updefault}$a$}}}}
\put(1680,2131){\makebox(0,0)[b]{\smash{{\SetFigFont{8}{9.6}{\familydefault}{\mddefault}{\updefault}$b$}}}}
\put(5145,1546){\makebox(0,0)[b]{\smash{{\SetFigFont{6}{7.2}{\familydefault}{\mddefault}{\updefault}$m-1$}}}}
\put(5145,196){\makebox(0,0)[b]{\smash{{\SetFigFont{6}{7.2}{\familydefault}{\mddefault}{\updefault}$n-1$}}}}
\put(4245,1546){\makebox(0,0)[b]{\smash{{\SetFigFont{6}{7.2}{\familydefault}{\mddefault}{\updefault}$m-2$}}}}
\put(4245,196){\makebox(0,0)[b]{\smash{{\SetFigFont{6}{7.2}{\familydefault}{\mddefault}{\updefault}$n-2$}}}}
\end{picture}
}

%% file: union_factorfree.eepic
\setlength{\unitlength}{0.00065617in}
\begingroup\makeatletter\ifx\SetFigFont\undefined%
\gdef\SetFigFont#1#2#3#4#5{%
  \reset@font\fontsize{#1}{#2pt}%
  \fontfamily{#3}\fontseries{#4}\fontshape{#5}%
  \selectfont}%
\fi\endgroup%
{\renewcommand{\dashlinestretch}{30}
\begin{picture}(5614,4994)(0,-10)
\put(1774,3992){\ellipse{382}{382}}
\put(5308.692,3256.500){\arc{316.547}{3.6174}{8.9490}}
\blacken\path(5279.040,3129.500)(5168.000,3184.000)(5240.329,3083.658)(5279.040,3129.500)
\put(5302.692,3995.500){\arc{316.547}{3.6174}{8.9490}}
\blacken\path(5273.040,3868.500)(5162.000,3923.000)(5234.329,3822.658)(5273.040,3868.500)
\put(5308.692,2470.500){\arc{316.547}{3.6174}{8.9490}}
\blacken\path(5279.040,2343.500)(5168.000,2398.000)(5240.329,2297.658)(5279.040,2343.500)
\put(2566.500,657.956){\arc{216.099}{5.1883}{10.5197}}
\blacken\path(2699.434,662.683)(2616.000,754.000)(2646.645,634.163)(2699.434,662.683)
\put(5307.266,985.500){\arc{319.348}{3.6129}{8.9535}}
\blacken\path(5276.248,857.573)(5165.000,913.000)(5237.077,811.742)(5276.248,857.573)
\put(1781.000,658.255){\arc{216.962}{5.2018}{10.5061}}
\blacken\path(1914.551,661.884)(1832.000,754.000)(1861.490,633.874)(1914.551,661.884)
\put(2989.000,3810.625){\arc{7151.816}{1.0139}{2.1277}}
\blacken\path(4791.355,687.717)(4879.000,775.000)(4760.592,739.230)(4791.355,687.717)
\put(3398.500,657.956){\arc{216.099}{5.1883}{10.5197}}
\blacken\path(3531.434,662.683)(3448.000,754.000)(3478.645,634.163)(3531.434,662.683)
\put(987,3992){\ellipse{382}{382}}
\put(199,4780){\ellipse{382}{382}}
\put(2562,3992){\ellipse{382}{382}}
\put(1774,3227){\ellipse{382}{382}}
\put(2562,3227){\ellipse{382}{382}}
\put(4204,3992){\ellipse{382}{382}}
\put(2562,2462){\ellipse{382}{382}}
\put(3394,3992){\ellipse{382}{382}}
\put(4193,3227){\ellipse{384}{382}}
\put(4969,3227){\ellipse{382}{382}}
\put(3394,2462){\ellipse{382}{382}}
\put(4193,2451){\ellipse{384}{382}}
\put(4979,1712){\ellipse{386}{382}}
\put(3394,1697){\ellipse{382}{382}}
\put(2562,1697){\ellipse{382}{382}}
\put(1774,1697){\ellipse{382}{382}}
\put(2562,932){\ellipse{382}{382}}
\put(1774,932){\ellipse{382}{382}}
\put(3394,932){\ellipse{382}{382}}
\put(4181,932){\ellipse{382}{382}}
\put(4979,947){\ellipse{386}{382}}
\put(3394,3227){\ellipse{382}{382}}
\put(1774,2462){\ellipse{382}{382}}
\put(4969,3992){\ellipse{382}{382}}
\put(4969,2462){\ellipse{382}{382}}
\put(942,932){\ellipse{382}{382}}
\path(1189,3992)(1594,3992)
\blacken\path(1474.000,3962.000)(1594.000,3992.000)(1474.000,4022.000)(1474.000,3962.000)
\path(2764,3992)(3169,3992)
\blacken\path(3049.000,3962.000)(3169.000,3992.000)(3049.000,4022.000)(3049.000,3962.000)
\path(334,4645)(829,4150)
\blacken\path(722.934,4213.640)(829.000,4150.000)(765.360,4256.066)(722.934,4213.640)
\path(1954,3992)(2359,3992)
\blacken\path(2239.000,3962.000)(2359.000,3992.000)(2239.000,4022.000)(2239.000,3962.000)
\path(1954,3205)(2359,3205)
\blacken\path(2239.000,3175.000)(2359.000,3205.000)(2239.000,3235.000)(2239.000,3175.000)
\path(2764,3205)(3169,3205)
\blacken\path(3049.000,3175.000)(3169.000,3205.000)(3049.000,3235.000)(3049.000,3175.000)
\path(3574,4015)(3979,4015)
\blacken\path(3859.000,3985.000)(3979.000,4015.000)(3859.000,4045.000)(3859.000,3985.000)
\path(4316,3093)(4834,2620)
\blacken\path(4725.156,2678.763)(4834.000,2620.000)(4765.615,2723.070)(4725.156,2678.763)
\path(4384,3992)(4744,3992)
\blacken\path(4624.000,3962.000)(4744.000,3992.000)(4624.000,4022.000)(4624.000,3962.000)
\path(1954,2440)(2359,2440)
\blacken\path(2239.000,2410.000)(2359.000,2440.000)(2239.000,2470.000)(2239.000,2410.000)
\path(2764,2440)(3169,2440)
\blacken\path(3049.000,2410.000)(3169.000,2440.000)(3049.000,2470.000)(3049.000,2410.000)
\path(3596,2440)(3979,2440)
\blacken\path(3859.000,2410.000)(3979.000,2440.000)(3859.000,2470.000)(3859.000,2410.000)
\path(4384,2440)(4789,2440)
\blacken\path(4669.000,2410.000)(4789.000,2440.000)(4669.000,2470.000)(4669.000,2410.000)
\path(3394,2260)(3394,1900)
\blacken\path(3364.000,2020.000)(3394.000,1900.000)(3424.000,2020.000)(3364.000,2020.000)
\path(1909,1540)(2404,1090)
\blacken\path(2295.027,1148.523)(2404.000,1090.000)(2335.387,1192.919)(2295.027,1148.523)
\path(2674,1540)(3259,1090)
\blacken\path(3145.594,1139.387)(3259.000,1090.000)(3182.176,1186.944)(3145.594,1139.387)
\path(1954,910)(2359,910)
\blacken\path(2239.000,880.000)(2359.000,910.000)(2239.000,940.000)(2239.000,880.000)
\path(4387,910)(4773,910)
\blacken\path(4652.420,879.855)(4773.000,910.000)(4652.420,940.145)(4652.420,879.855)
\path(2764,910)(3169,910)
\blacken\path(3049.000,880.000)(3169.000,910.000)(3049.000,940.000)(3049.000,880.000)
\path(3574,910)(3979,910)
\blacken\path(3859.000,880.000)(3979.000,910.000)(3859.000,940.000)(3859.000,880.000)
\path(3568,1533)(4063,1083)
\blacken\path(3954.027,1141.523)(4063.000,1083.000)(3994.387,1185.919)(3954.027,1141.523)
\path(1144,3880)(4789,3340)
\blacken\path(4665.899,3327.910)(4789.000,3340.000)(4674.692,3387.262)(4665.899,3327.910)
\drawline(1909,1540)(1909,1540)
\drawline(1999,1540)(1999,1540)
\drawline(1909,775)(1909,775)
\drawline(1999,775)(1999,775)
\path(1144,910)(1549,910)
\blacken\path(1429.000,880.000)(1549.000,910.000)(1429.000,940.000)(1429.000,880.000)
\path(3574,3227)(3979,3227)
\blacken\path(3859.000,3197.000)(3979.000,3227.000)(3859.000,3257.000)(3859.000,3197.000)
\path(4339,3835)(4856,3408)
\blacken\path(4744.373,3461.286)(4856.000,3408.000)(4782.581,3507.547)(4744.373,3461.286)
\path(4384,3227)(4767,3227)
\blacken\path(4647.000,3197.000)(4767.000,3227.000)(4647.000,3257.000)(4647.000,3197.000)
\path(244,4555)(874,1135)
\blacken\path(822.757,1247.580)(874.000,1135.000)(881.764,1258.449)(822.757,1247.580)
\dashline{60.000}(3956,1330)(3956,565)(4410,565)
	(4410,1465)(5228,1465)(5228,1915)
	(5228,1960)(3956,1960)(3956,1330)
\path(4979,2260)(4979,1877)
\blacken\path(4949.000,1997.000)(4979.000,1877.000)(5009.000,1997.000)(4949.000,1997.000)
\path(4979,3790)(4979,3430)
\blacken\path(4949.000,3550.000)(4979.000,3430.000)(5009.000,3550.000)(4949.000,3550.000)
\path(4204,3025)(4204,2665)
\blacken\path(4174.000,2785.000)(4204.000,2665.000)(4234.000,2785.000)(4174.000,2785.000)
\path(4979,3025)(4979,2665)
\blacken\path(4949.000,2785.000)(4979.000,2665.000)(5009.000,2785.000)(4949.000,2785.000)
\path(3394,3025)(3394,2665)
\blacken\path(3364.000,2785.000)(3394.000,2665.000)(3424.000,2785.000)(3364.000,2785.000)
\path(2562,3025)(2562,2665)
\blacken\path(2532.000,2785.000)(2562.000,2665.000)(2592.000,2785.000)(2532.000,2785.000)
\path(1774,3025)(1774,2665)
\blacken\path(1744.000,2785.000)(1774.000,2665.000)(1804.000,2785.000)(1744.000,2785.000)
\path(1774,3790)(1774,3430)
\blacken\path(1744.000,3550.000)(1774.000,3430.000)(1804.000,3550.000)(1744.000,3550.000)
\path(2562,3790)(2562,3430)
\blacken\path(2532.000,3550.000)(2562.000,3430.000)(2592.000,3550.000)(2532.000,3550.000)
\path(3394,3790)(3394,3430)
\blacken\path(3364.000,3550.000)(3394.000,3430.000)(3424.000,3550.000)(3364.000,3550.000)
\path(4204,3790)(4204,3430)
\blacken\path(4174.000,3550.000)(4204.000,3430.000)(4234.000,3550.000)(4174.000,3550.000)
\path(1774,2260)(1774,1900)
\blacken\path(1744.000,2020.000)(1774.000,1900.000)(1804.000,2020.000)(1744.000,2020.000)
\path(2562,2260)(2562,1900)
\blacken\path(2532.000,2020.000)(2562.000,1900.000)(2592.000,2020.000)(2532.000,2020.000)
\path(1774,1495)(1774,1135)
\blacken\path(1744.000,1255.000)(1774.000,1135.000)(1804.000,1255.000)(1744.000,1255.000)
\path(2562,1495)(2562,1135)
\blacken\path(2532.000,1255.000)(2562.000,1135.000)(2592.000,1255.000)(2532.000,1255.000)
\path(3394,1495)(3394,1135)
\blacken\path(3364.000,1255.000)(3394.000,1135.000)(3424.000,1255.000)(3364.000,1255.000)
\path(4979,1510)(4979,1157)
\blacken\path(4949.000,1277.000)(4979.000,1157.000)(5009.000,1277.000)(4949.000,1277.000)
\put(1774,3925){\makebox(0,0)[b]{\smash{{\SetFigFont{8}{9.6}{\familydefault}{\mddefault}{\updefault}2,3}}}}
\put(5419,3475){\makebox(0,0)[b]{\smash{{\SetFigFont{8}{9.6}{\familydefault}{\mddefault}{\updefault}$b$}}}}
\put(199,4722){\makebox(0,0)[b]{\smash{{\SetFigFont{8}{9.6}{\familydefault}{\mddefault}{\updefault}1,1}}}}
\put(987,3925){\makebox(0,0)[b]{\smash{{\SetFigFont{8}{9.6}{\familydefault}{\mddefault}{\updefault}2,2}}}}
\put(2562,3925){\makebox(0,0)[b]{\smash{{\SetFigFont{8}{9.6}{\familydefault}{\mddefault}{\updefault}2,4}}}}
\put(1316,4020){\makebox(0,0)[b]{\smash{{\SetFigFont{8}{9.6}{\familydefault}{\mddefault}{\updefault}$b$}}}}
\put(583,4460){\makebox(0,0)[b]{\smash{{\SetFigFont{8}{9.6}{\familydefault}{\mddefault}{\updefault}$a$}}}}
\put(1639,3565){\makebox(0,0)[b]{\smash{{\SetFigFont{8}{9.6}{\familydefault}{\mddefault}{\updefault}$a$}}}}
\put(2089,3250){\makebox(0,0)[b]{\smash{{\SetFigFont{8}{9.6}{\familydefault}{\mddefault}{\updefault}$b$}}}}
\put(2899,3250){\makebox(0,0)[b]{\smash{{\SetFigFont{8}{9.6}{\familydefault}{\mddefault}{\updefault}$b$}}}}
\put(1639,2845){\makebox(0,0)[b]{\smash{{\SetFigFont{8}{9.6}{\familydefault}{\mddefault}{\updefault}$a$}}}}
\put(2404,2845){\makebox(0,0)[b]{\smash{{\SetFigFont{8}{9.6}{\familydefault}{\mddefault}{\updefault}$a$}}}}
\put(1774,3160){\makebox(0,0)[b]{\smash{{\SetFigFont{8}{9.6}{\familydefault}{\mddefault}{\updefault}3,3}}}}
\put(2562,3160){\makebox(0,0)[b]{\smash{{\SetFigFont{8}{9.6}{\familydefault}{\mddefault}{\updefault}3,4}}}}
\put(3214,2845){\makebox(0,0)[b]{\smash{{\SetFigFont{8}{9.6}{\familydefault}{\mddefault}{\updefault}$a$}}}}
\put(4519,4015){\makebox(0,0)[b]{\smash{{\SetFigFont{8}{9.6}{\familydefault}{\mddefault}{\updefault}$b$}}}}
\put(5419,4195){\makebox(0,0)[b]{\smash{{\SetFigFont{8}{9.6}{\familydefault}{\mddefault}{\updefault}$b$}}}}
\put(5149,3565){\makebox(0,0)[b]{\smash{{\SetFigFont{8}{9.6}{\familydefault}{\mddefault}{\updefault}$a$}}}}
\put(5149,2755){\makebox(0,0)[b]{\smash{{\SetFigFont{8}{9.6}{\familydefault}{\mddefault}{\updefault}$a$}}}}
\put(4654,3655){\makebox(0,0)[b]{\smash{{\SetFigFont{8}{9.6}{\familydefault}{\mddefault}{\updefault}$a$}}}}
\put(4609,2890){\makebox(0,0)[b]{\smash{{\SetFigFont{8}{9.6}{\familydefault}{\mddefault}{\updefault}$a$}}}}
\put(2089,2485){\makebox(0,0)[b]{\smash{{\SetFigFont{8}{9.6}{\familydefault}{\mddefault}{\updefault}$b$}}}}
\put(2899,2485){\makebox(0,0)[b]{\smash{{\SetFigFont{8}{9.6}{\familydefault}{\mddefault}{\updefault}$b$}}}}
\put(1639,2035){\makebox(0,0)[b]{\smash{{\SetFigFont{8}{9.6}{\familydefault}{\mddefault}{\updefault}$a$}}}}
\put(2404,2035){\makebox(0,0)[b]{\smash{{\SetFigFont{8}{9.6}{\familydefault}{\mddefault}{\updefault}$a$}}}}
\put(4564,2485){\makebox(0,0)[b]{\smash{{\SetFigFont{8}{9.6}{\familydefault}{\mddefault}{\updefault}$b$}}}}
\put(5464,2620){\makebox(0,0)[b]{\smash{{\SetFigFont{8}{9.6}{\familydefault}{\mddefault}{\updefault}$b$}}}}
\put(2044,1225){\makebox(0,0)[b]{\smash{{\SetFigFont{8}{9.6}{\familydefault}{\mddefault}{\updefault}$b$}}}}
\put(2854,1225){\makebox(0,0)[b]{\smash{{\SetFigFont{8}{9.6}{\familydefault}{\mddefault}{\updefault}$b$}}}}
\put(3259,1315){\makebox(0,0)[b]{\smash{{\SetFigFont{8}{9.6}{\familydefault}{\mddefault}{\updefault}$a$}}}}
\put(5169,1285){\makebox(0,0)[b]{\smash{{\SetFigFont{8}{9.6}{\familydefault}{\mddefault}{\updefault}$a,b$}}}}
\put(5599,1180){\makebox(0,0)[b]{\smash{{\SetFigFont{8}{9.6}{\familydefault}{\mddefault}{\updefault}$a,b$}}}}
\put(2944,55){\makebox(0,0)[b]{\smash{{\SetFigFont{8}{9.6}{\familydefault}{\mddefault}{\updefault}$a$}}}}
\put(4519,3250){\makebox(0,0)[b]{\smash{{\SetFigFont{8}{9.6}{\familydefault}{\mddefault}{\updefault}$b$}}}}
\put(4069,3655){\makebox(0,0)[b]{\smash{{\SetFigFont{8}{9.6}{\familydefault}{\mddefault}{\updefault}$a$}}}}
\put(3304,3655){\makebox(0,0)[b]{\smash{{\SetFigFont{8}{9.6}{\familydefault}{\mddefault}{\updefault}$a$}}}}
\put(2449,3520){\makebox(0,0)[b]{\smash{{\SetFigFont{8}{9.6}{\familydefault}{\mddefault}{\updefault}$a$}}}}
\put(2854,3655){\makebox(0,0)[b]{\smash{{\SetFigFont{8}{9.6}{\familydefault}{\mddefault}{\updefault}$a$}}}}
\put(3754,4060){\makebox(0,0)[b]{\smash{{\SetFigFont{8}{9.6}{\familydefault}{\mddefault}{\updefault}$b$}}}}
\put(3754,3250){\makebox(0,0)[b]{\smash{{\SetFigFont{8}{9.6}{\familydefault}{\mddefault}{\updefault}$b$}}}}
\put(2899,4038){\makebox(0,0)[b]{\smash{{\SetFigFont{8}{9.6}{\familydefault}{\mddefault}{\updefault}$b$}}}}
\put(2089,4038){\makebox(0,0)[b]{\smash{{\SetFigFont{8}{9.6}{\familydefault}{\mddefault}{\updefault}$b$}}}}
\put(4193,3160){\makebox(0,0)[b]{\smash{{\SetFigFont{8}{9.6}{\familydefault}{\mddefault}{\updefault}3,6}}}}
\put(4969,3160){\makebox(0,0)[b]{\smash{{\SetFigFont{8}{9.6}{\familydefault}{\mddefault}{\updefault}3,7}}}}
\put(2562,1630){\makebox(0,0)[b]{\smash{{\SetFigFont{8}{9.6}{\familydefault}{\mddefault}{\updefault}5,4}}}}
\put(1774,1630){\makebox(0,0)[b]{\smash{{\SetFigFont{8}{9.6}{\familydefault}{\mddefault}{\updefault}5,3}}}}
\put(2562,865){\makebox(0,0)[b]{\smash{{\SetFigFont{8}{9.6}{\familydefault}{\mddefault}{\updefault}6,4}}}}
\put(1774,865){\makebox(0,0)[b]{\smash{{\SetFigFont{8}{9.6}{\familydefault}{\mddefault}{\updefault}6,3}}}}
\put(3394,3160){\makebox(0,0)[b]{\smash{{\SetFigFont{8}{9.6}{\familydefault}{\mddefault}{\updefault}3,5}}}}
\put(1774,2395){\makebox(0,0)[b]{\smash{{\SetFigFont{8}{9.6}{\familydefault}{\mddefault}{\updefault}4,3}}}}
\put(3394,2395){\makebox(0,0)[b]{\smash{{\SetFigFont{8}{9.6}{\familydefault}{\mddefault}{\updefault}4,5}}}}
\put(2562,2395){\makebox(0,0)[b]{\smash{{\SetFigFont{8}{9.6}{\familydefault}{\mddefault}{\updefault}4,4}}}}
\put(4193,2384){\makebox(0,0)[b]{\smash{{\SetFigFont{8}{9.6}{\familydefault}{\mddefault}{\updefault}4,6}}}}
\put(3394,3925){\makebox(0,0)[b]{\smash{{\SetFigFont{8}{9.6}{\familydefault}{\mddefault}{\updefault}2,5}}}}
\put(4204,3925){\makebox(0,0)[b]{\smash{{\SetFigFont{8}{9.6}{\familydefault}{\mddefault}{\updefault}2,6}}}}
\put(4969,3925){\makebox(0,0)[b]{\smash{{\SetFigFont{8}{9.6}{\familydefault}{\mddefault}{\updefault}2,7}}}}
\put(4969,2395){\makebox(0,0)[b]{\smash{{\SetFigFont{8}{9.6}{\familydefault}{\mddefault}{\updefault}4,7}}}}
\put(4979,1645){\makebox(0,0)[b]{\smash{{\SetFigFont{8}{9.6}{\familydefault}{\mddefault}{\updefault}5,7}}}}
\put(3394,1630){\makebox(0,0)[b]{\smash{{\SetFigFont{8}{9.6}{\familydefault}{\mddefault}{\updefault}5,5}}}}
\put(3394,865){\makebox(0,0)[b]{\smash{{\SetFigFont{8}{9.6}{\familydefault}{\mddefault}{\updefault}6,5}}}}
\put(4181,865){\makebox(0,0)[b]{\smash{{\SetFigFont{8}{9.6}{\familydefault}{\mddefault}{\updefault}6,6}}}}
\put(4979,880){\makebox(0,0)[b]{\smash{{\SetFigFont{8}{9.6}{\familydefault}{\mddefault}{\updefault}6,7}}}}
\put(942,865){\makebox(0,0)[b]{\smash{{\SetFigFont{8}{9.6}{\familydefault}{\mddefault}{\updefault}6,2}}}}
\put(1999,595){\makebox(0,0)[b]{\smash{{\SetFigFont{8}{9.6}{\familydefault}{\mddefault}{\updefault}$a$}}}}
\put(2809,595){\makebox(0,0)[b]{\smash{{\SetFigFont{8}{9.6}{\familydefault}{\mddefault}{\updefault}$a$}}}}
\put(2089,955){\makebox(0,0)[b]{\smash{{\SetFigFont{8}{9.6}{\familydefault}{\mddefault}{\updefault}$b$}}}}
\put(2899,955){\makebox(0,0)[b]{\smash{{\SetFigFont{8}{9.6}{\familydefault}{\mddefault}{\updefault}$b$}}}}
\put(1324,955){\makebox(0,0)[b]{\smash{{\SetFigFont{8}{9.6}{\familydefault}{\mddefault}{\updefault}$b$}}}}
\put(3754,955){\makebox(0,0)[b]{\smash{{\SetFigFont{8}{9.6}{\familydefault}{\mddefault}{\updefault}$b$}}}}
\put(4564,955){\makebox(0,0)[b]{\smash{{\SetFigFont{8}{9.6}{\familydefault}{\mddefault}{\updefault}$a,b$}}}}
\put(3664,595){\makebox(0,0)[b]{\smash{{\SetFigFont{8}{9.6}{\familydefault}{\mddefault}{\updefault}$a$}}}}
\put(3709,1225){\makebox(0,0)[b]{\smash{{\SetFigFont{8}{9.6}{\familydefault}{\mddefault}{\updefault}$b$}}}}
\put(4069,2845){\makebox(0,0)[b]{\smash{{\SetFigFont{8}{9.6}{\familydefault}{\mddefault}{\updefault}$a$}}}}
\put(3214,2035){\makebox(0,0)[b]{\smash{{\SetFigFont{8}{9.6}{\familydefault}{\mddefault}{\updefault}$a$}}}}
\put(2404,1315){\makebox(0,0)[b]{\smash{{\SetFigFont{8}{9.6}{\familydefault}{\mddefault}{\updefault}$a$}}}}
\put(1639,1315){\makebox(0,0)[b]{\smash{{\SetFigFont{8}{9.6}{\familydefault}{\mddefault}{\updefault}$a$}}}}
\put(3754,2485){\makebox(0,0)[b]{\smash{{\SetFigFont{8}{9.6}{\familydefault}{\mddefault}{\updefault}$b$}}}}
\put(5104,2080){\makebox(0,0)[b]{\smash{{\SetFigFont{8}{9.6}{\familydefault}{\mddefault}{\updefault}$a$}}}}
\put(334,2530){\makebox(0,0)[b]{\smash{{\SetFigFont{8}{9.6}{\familydefault}{\mddefault}{\updefault}$b$}}}}
\end{picture}
}

%% file: booleansub.eepic
\setlength{\unitlength}{0.00034996in}
\begingroup\makeatletter\ifx\SetFigFont\undefined%
\gdef\SetFigFont#1#2#3#4#5{%
  \reset@font\fontsize{#1}{#2pt}%
  \fontfamily{#3}\fontseries{#4}\fontshape{#5}%
  \selectfont}%
\fi\endgroup%
{\renewcommand{\dashlinestretch}{30}
\begin{picture}(11783,5491)(0,-10)
\put(1558,1166){\makebox(0,0)[lb]{\smash{{\SetFigFont{7}{8.4}{\familydefault}{\mddefault}{\updefault}$L$}}}}
\put(5329,4220){\ellipse{810}{810}}
\put(7590,4218){\ellipse{720}{720}}
\put(3078,4223){\ellipse{810}{810}}
\put(829,4226){\ellipse{810}{810}}
\put(2369,1579){\ellipse{810}{810}}
\put(4626,1582){\ellipse{810}{810}}
\put(6857,1579){\ellipse{810}{810}}
\put(9115,1588){\ellipse{810}{810}}
\put(11370,1582){\ellipse{810}{810}}
\put(11369,1586){\ellipse{720}{720}}
\path(5742,4219)(7145,4219)
\blacken\thicklines
\path(7010.000,4181.500)(7145.000,4219.000)(7010.000,4256.500)(7010.000,4181.500)
\thinlines
\path(3485,4219)(4888,4219)
\blacken\thicklines
\path(4753.000,4181.500)(4888.000,4219.000)(4753.000,4256.500)(4753.000,4181.500)
\thinlines
\path(1235,4219)(2638,4219)
\blacken\thicklines
\path(2503.000,4181.500)(2638.000,4219.000)(2503.000,4256.500)(2503.000,4181.500)
\thinlines
\path(12,4212)(387,4212)
\blacken\thicklines
\path(252.000,4174.500)(387.000,4212.000)(252.000,4249.500)(252.000,4174.500)
\thinlines
\path(1557,1586)(1932,1586)
\blacken\thicklines
\path(1797.000,1548.500)(1932.000,1586.000)(1797.000,1623.500)(1797.000,1548.500)
\thinlines
\path(2780,1586)(4183,1586)
\blacken\thicklines
\path(4048.000,1548.500)(4183.000,1586.000)(4048.000,1623.500)(4048.000,1548.500)
\thinlines
\path(5023,1571)(6426,1571)
\blacken\thicklines
\path(6291.000,1533.500)(6426.000,1571.000)(6291.000,1608.500)(6291.000,1533.500)
\thinlines
\path(7265,1578)(8668,1578)
\blacken\thicklines
\path(8533.000,1540.500)(8668.000,1578.000)(8533.000,1615.500)(8533.000,1540.500)
\thinlines
\path(9530,1578)(10933,1578)
\blacken\thicklines
\path(10798.000,1540.500)(10933.000,1578.000)(10798.000,1615.500)(10798.000,1540.500)
\thinlines
\path(5555,4542)(5556,4543)(5558,4544)
	(5562,4547)(5569,4552)(5578,4559)
	(5590,4567)(5605,4578)(5624,4590)
	(5644,4604)(5668,4620)(5694,4637)
	(5722,4655)(5751,4673)(5783,4692)
	(5816,4710)(5850,4729)(5885,4747)
	(5922,4765)(5960,4782)(6000,4798)
	(6042,4814)(6085,4828)(6131,4842)
	(6179,4854)(6230,4865)(6283,4874)
	(6338,4880)(6396,4885)(6455,4887)
	(6514,4886)(6571,4882)(6626,4876)
	(6679,4868)(6728,4859)(6775,4847)
	(6820,4835)(6862,4821)(6902,4806)
	(6941,4790)(6977,4774)(7013,4757)
	(7047,4739)(7079,4721)(7110,4703)
	(7140,4685)(7168,4667)(7195,4650)
	(7219,4633)(7241,4618)(7261,4604)
	(7278,4592)(7292,4582)(7303,4573)
	(7312,4567)(7325,4557)
\blacken\thicklines
\path(7195.132,4609.588)(7325.000,4557.000)(7240.860,4669.034)(7195.132,4609.588)
\thinlines
\path(3200,4624)(3201,4625)(3202,4626)
	(3205,4628)(3209,4632)(3216,4637)
	(3224,4643)(3235,4652)(3249,4663)
	(3265,4676)(3284,4691)(3306,4708)
	(3331,4727)(3358,4748)(3388,4770)
	(3421,4794)(3456,4820)(3494,4847)
	(3533,4875)(3575,4903)(3618,4933)
	(3663,4962)(3710,4992)(3758,5022)
	(3808,5052)(3859,5082)(3911,5111)
	(3965,5140)(4020,5168)(4076,5196)
	(4135,5223)(4195,5249)(4256,5274)
	(4320,5298)(4386,5321)(4455,5343)
	(4526,5363)(4599,5382)(4676,5400)
	(4755,5416)(4837,5430)(4922,5441)
	(5010,5451)(5100,5458)(5192,5463)
	(5285,5464)(5374,5462)(5463,5458)
	(5550,5451)(5636,5442)(5720,5431)
	(5801,5418)(5880,5403)(5956,5386)
	(6030,5368)(6102,5349)(6172,5328)
	(6239,5306)(6305,5284)(6368,5260)
	(6430,5235)(6491,5210)(6550,5184)
	(6608,5157)(6664,5130)(6719,5102)
	(6773,5073)(6826,5045)(6877,5016)
	(6928,4987)(6976,4958)(7024,4929)
	(7069,4901)(7113,4873)(7156,4846)
	(7196,4820)(7234,4795)(7269,4772)
	(7302,4749)(7332,4729)(7360,4710)
	(7385,4692)(7407,4677)(7426,4663)
	(7442,4652)(7455,4642)(7466,4634)
	(7475,4628)(7481,4624)(7490,4617)
\blacken\thicklines
\path(7360.415,4670.281)(7490.000,4617.000)(7406.460,4729.483)(7360.415,4670.281)
\thinlines
\path(1062,3897)(1063,3897)(1064,3896)
	(1067,3894)(1072,3891)(1080,3886)
	(1089,3881)(1101,3873)(1117,3864)
	(1135,3854)(1156,3841)(1180,3827)
	(1207,3812)(1238,3795)(1271,3776)
	(1306,3756)(1344,3736)(1385,3714)
	(1427,3692)(1472,3669)(1518,3646)
	(1566,3623)(1615,3600)(1666,3576)
	(1718,3553)(1772,3531)(1827,3508)
	(1883,3486)(1941,3465)(2001,3445)
	(2062,3425)(2125,3406)(2190,3387)
	(2258,3370)(2328,3354)(2400,3339)
	(2475,3325)(2553,3312)(2634,3301)
	(2717,3291)(2803,3284)(2891,3278)
	(2981,3275)(3072,3274)(3163,3276)
	(3253,3280)(3341,3286)(3427,3294)
	(3510,3304)(3590,3316)(3668,3329)
	(3742,3344)(3815,3359)(3884,3376)
	(3951,3394)(4016,3413)(4078,3432)
	(4139,3453)(4198,3474)(4256,3495)
	(4311,3518)(4366,3540)(4419,3564)
	(4471,3587)(4521,3611)(4570,3634)
	(4617,3658)(4663,3681)(4707,3704)
	(4749,3727)(4788,3749)(4826,3770)
	(4861,3790)(4894,3808)(4924,3825)
	(4951,3841)(4974,3855)(4995,3868)
	(5013,3879)(5028,3888)(5040,3895)
	(5050,3901)(5057,3906)(5067,3912)
\blacken\thicklines
\path(4970.532,3810.387)(5067.000,3912.000)(4931.945,3874.699)(4970.532,3810.387)
\thinlines
\path(980,3837)(981,3837)(982,3836)
	(984,3835)(988,3832)(994,3829)
	(1001,3825)(1011,3819)(1024,3812)
	(1038,3804)(1056,3794)(1077,3782)
	(1100,3769)(1127,3755)(1156,3739)
	(1188,3722)(1224,3703)(1262,3683)
	(1302,3661)(1346,3639)(1391,3615)
	(1439,3591)(1490,3566)(1542,3540)
	(1596,3514)(1652,3488)(1710,3461)
	(1769,3434)(1830,3407)(1892,3380)
	(1956,3353)(2021,3326)(2087,3300)
	(2155,3274)(2224,3249)(2295,3223)
	(2367,3199)(2441,3175)(2517,3152)
	(2594,3129)(2674,3107)(2755,3086)
	(2839,3065)(2926,3046)(3014,3027)
	(3106,3010)(3200,2994)(3297,2978)
	(3397,2965)(3500,2952)(3606,2941)
	(3714,2932)(3825,2924)(3938,2919)
	(4052,2915)(4167,2914)(4278,2915)
	(4389,2918)(4498,2923)(4606,2930)
	(4712,2938)(4815,2948)(4916,2960)
	(5015,2972)(5111,2986)(5204,3002)
	(5295,3018)(5383,3035)(5469,3053)
	(5553,3072)(5635,3092)(5714,3113)
	(5792,3134)(5868,3156)(5943,3179)
	(6016,3202)(6087,3226)(6157,3250)
	(6226,3274)(6294,3299)(6360,3325)
	(6425,3350)(6488,3376)(6551,3402)
	(6612,3428)(6671,3454)(6729,3480)
	(6786,3505)(6841,3530)(6894,3555)
	(6945,3579)(6994,3603)(7041,3625)
	(7085,3647)(7127,3668)(7167,3688)
	(7203,3706)(7237,3724)(7268,3740)
	(7297,3754)(7322,3768)(7345,3780)
	(7364,3790)(7381,3799)(7396,3807)
	(7407,3813)(7417,3818)(7424,3822)
	(7429,3825)(7437,3829)
\blacken\thicklines
\path(7333.023,3735.085)(7437.000,3829.000)(7299.482,3802.167)(7333.023,3735.085)
\thinlines
\path(9433,1849)(9434,1850)(9436,1852)
	(9440,1855)(9447,1861)(9456,1868)
	(9468,1878)(9482,1890)(9500,1903)
	(9520,1919)(9542,1935)(9567,1953)
	(9593,1972)(9621,1990)(9651,2009)
	(9681,2028)(9714,2046)(9747,2063)
	(9783,2080)(9820,2096)(9859,2111)
	(9900,2124)(9944,2136)(9991,2146)
	(10041,2154)(10093,2160)(10148,2164)
	(10205,2164)(10262,2161)(10318,2155)
	(10371,2147)(10423,2136)(10471,2123)
	(10517,2109)(10561,2094)(10603,2077)
	(10642,2060)(10680,2041)(10717,2022)
	(10752,2002)(10786,1982)(10819,1962)
	(10850,1942)(10879,1922)(10907,1903)
	(10932,1885)(10954,1869)(10974,1854)
	(10990,1842)(11004,1831)(11014,1823)(11030,1811)
\blacken\thicklines
\path(10899.500,1862.000)(11030.000,1811.000)(10944.500,1922.000)(10899.500,1862.000)
\thinlines
\path(7032,1961)(7033,1961)(7034,1962)
	(7037,1964)(7041,1967)(7048,1971)
	(7056,1976)(7067,1983)(7081,1991)
	(7098,2001)(7117,2012)(7139,2025)
	(7164,2039)(7192,2055)(7222,2072)
	(7255,2091)(7290,2110)(7328,2131)
	(7368,2152)(7410,2174)(7453,2196)
	(7498,2219)(7545,2241)(7593,2264)
	(7642,2287)(7693,2309)(7745,2332)
	(7798,2353)(7853,2375)(7909,2396)
	(7966,2416)(8025,2435)(8086,2454)
	(8149,2472)(8213,2490)(8280,2506)
	(8349,2521)(8421,2535)(8495,2548)
	(8572,2559)(8651,2570)(8733,2578)
	(8817,2585)(8903,2589)(8991,2592)
	(9080,2592)(9169,2590)(9257,2585)
	(9343,2579)(9427,2570)(9509,2560)
	(9588,2548)(9665,2535)(9739,2521)
	(9810,2505)(9879,2488)(9945,2470)
	(10010,2452)(10072,2432)(10133,2412)
	(10192,2391)(10249,2370)(10304,2348)
	(10359,2325)(10412,2302)(10463,2279)
	(10514,2255)(10563,2231)(10611,2208)
	(10657,2184)(10702,2160)(10745,2137)
	(10787,2114)(10826,2092)(10863,2071)
	(10899,2051)(10931,2032)(10962,2014)
	(10989,1997)(11014,1982)(11036,1969)
	(11055,1957)(11071,1947)(11085,1938)
	(11096,1931)(11104,1926)(11111,1922)(11120,1916)
\blacken\thicklines
\path(10986.872,1959.683)(11120.000,1916.000)(11028.474,2022.086)(10986.872,1959.683)
\thinlines
\path(4692,1976)(4693,1976)(4694,1977)
	(4696,1978)(4700,1980)(4705,1983)
	(4712,1987)(4722,1992)(4734,1999)
	(4748,2007)(4765,2016)(4785,2026)
	(4808,2038)(4833,2052)(4862,2067)
	(4893,2083)(4927,2100)(4964,2119)
	(5004,2139)(5046,2161)(5091,2183)
	(5138,2206)(5188,2230)(5239,2254)
	(5292,2279)(5348,2305)(5405,2331)
	(5463,2357)(5523,2383)(5585,2410)
	(5648,2436)(5712,2462)(5777,2488)
	(5844,2514)(5912,2539)(5981,2564)
	(6052,2589)(6124,2613)(6198,2637)
	(6273,2660)(6350,2682)(6428,2704)
	(6509,2725)(6591,2745)(6676,2764)
	(6763,2783)(6852,2800)(6944,2817)
	(7038,2832)(7135,2847)(7234,2860)
	(7336,2872)(7441,2882)(7548,2891)
	(7657,2898)(7768,2903)(7879,2906)
	(7992,2907)(8104,2906)(8216,2903)
	(8326,2898)(8435,2891)(8541,2882)
	(8645,2872)(8746,2860)(8845,2847)
	(8941,2832)(9034,2817)(9124,2800)
	(9212,2783)(9298,2764)(9381,2745)
	(9462,2725)(9541,2704)(9618,2682)
	(9693,2660)(9767,2637)(9839,2613)
	(9909,2589)(9978,2564)(10046,2539)
	(10112,2514)(10177,2488)(10240,2462)
	(10302,2436)(10363,2410)(10423,2383)
	(10481,2357)(10538,2331)(10593,2305)
	(10646,2279)(10698,2254)(10748,2230)
	(10795,2206)(10841,2183)(10884,2161)
	(10925,2139)(10963,2119)(10999,2100)
	(11032,2083)(11062,2067)(11089,2052)
	(11114,2038)(11136,2026)(11155,2016)
	(11171,2007)(11185,1999)(11196,1992)
	(11205,1987)(11212,1983)(11218,1980)(11225,1976)
\blacken\thicklines
\path(11089.182,2010.420)(11225.000,1976.000)(11126.392,2075.538)(11089.182,2010.420)
\thinlines
\path(2660,1294)(2661,1294)(2662,1293)
	(2666,1291)(2671,1289)(2678,1286)
	(2688,1281)(2700,1276)(2715,1269)
	(2734,1260)(2755,1251)(2780,1240)
	(2808,1228)(2838,1215)(2872,1201)
	(2908,1185)(2946,1170)(2987,1153)
	(3030,1136)(3075,1118)(3122,1101)
	(3170,1083)(3220,1065)(3271,1047)
	(3323,1029)(3377,1012)(3432,995)
	(3488,978)(3546,962)(3605,946)
	(3666,931)(3729,916)(3793,903)
	(3860,890)(3929,877)(4000,866)
	(4074,856)(4150,847)(4229,839)
	(4310,832)(4394,827)(4480,824)
	(4567,822)(4655,822)(4743,824)
	(4829,828)(4914,834)(4996,841)
	(5076,850)(5152,860)(5226,871)
	(5297,884)(5366,897)(5431,911)
	(5495,926)(5556,941)(5615,957)
	(5672,974)(5727,991)(5781,1009)
	(5833,1027)(5883,1045)(5933,1064)
	(5981,1083)(6027,1102)(6073,1122)
	(6116,1141)(6158,1160)(6199,1178)
	(6238,1197)(6274,1214)(6309,1231)
	(6341,1247)(6371,1262)(6399,1276)
	(6423,1289)(6445,1300)(6464,1311)
	(6481,1319)(6494,1327)(6506,1333)
	(6514,1337)(6521,1341)(6530,1346)
\blacken\thicklines
\path(6430.200,1247.657)(6530.000,1346.000)(6393.777,1313.219)(6430.200,1247.657)
\thinlines
\path(2502,1197)(2504,1196)(2506,1195)
	(2509,1193)(2514,1191)(2521,1187)
	(2530,1182)(2542,1176)(2555,1169)
	(2572,1161)(2591,1151)(2613,1140)
	(2638,1127)(2666,1113)(2697,1098)
	(2731,1081)(2768,1062)(2809,1042)
	(2852,1021)(2898,999)(2947,976)
	(2998,951)(3052,926)(3109,899)
	(3167,872)(3228,844)(3291,816)
	(3356,787)(3423,758)(3492,728)
	(3562,699)(3633,669)(3706,639)
	(3781,609)(3857,579)(3934,550)
	(4012,520)(4092,491)(4173,463)
	(4255,434)(4339,407)(4423,379)
	(4510,353)(4598,327)(4687,301)
	(4778,276)(4871,252)(4966,229)
	(5063,207)(5162,185)(5263,164)
	(5366,145)(5472,126)(5580,108)
	(5691,92)(5804,77)(5919,63)
	(6038,50)(6158,40)(6281,30)
	(6406,23)(6533,17)(6661,13)
	(6790,12)(6920,12)(7053,15)
	(7186,20)(7316,27)(7445,37)
	(7572,48)(7695,61)(7816,75)
	(7934,91)(8049,109)(8161,128)
	(8270,148)(8376,169)(8479,191)
	(8580,215)(8678,239)(8773,264)
	(8866,290)(8957,317)(9045,344)
	(9132,372)(9217,401)(9300,431)
	(9381,461)(9461,491)(9539,522)
	(9616,554)(9691,585)(9765,617)
	(9837,650)(9908,682)(9978,715)
	(10046,747)(10112,780)(10177,812)
	(10241,845)(10302,876)(10362,908)
	(10419,938)(10475,968)(10529,997)
	(10580,1026)(10629,1053)(10675,1079)
	(10719,1104)(10760,1127)(10798,1150)
	(10833,1170)(10866,1189)(10895,1207)
	(10922,1223)(10946,1237)(10967,1250)
	(10985,1261)(11001,1270)(11014,1279)
	(11025,1285)(11034,1291)(11040,1295)(11052,1302)
\blacken\thicklines
\path(10954.285,1201.586)(11052.000,1302.000)(10916.495,1266.369)(10954.285,1201.586)
\thinlines
\path(2577,1226)(2578,1226)(2579,1225)
	(2581,1224)(2585,1222)(2591,1219)
	(2598,1216)(2608,1211)(2621,1206)
	(2636,1199)(2654,1190)(2674,1181)
	(2698,1170)(2724,1158)(2754,1145)
	(2786,1131)(2821,1115)(2859,1099)
	(2900,1081)(2944,1063)(2989,1044)
	(3037,1024)(3088,1003)(3140,982)
	(3194,961)(3250,939)(3307,917)
	(3366,895)(3427,873)(3489,851)
	(3552,829)(3616,808)(3682,786)
	(3749,765)(3818,744)(3888,724)
	(3959,704)(4032,685)(4107,667)
	(4183,648)(4261,631)(4341,614)
	(4424,598)(4508,583)(4595,569)
	(4684,556)(4776,543)(4871,532)
	(4968,522)(5068,513)(5170,506)
	(5275,499)(5382,495)(5491,492)
	(5601,491)(5712,492)(5823,495)
	(5933,500)(6041,506)(6148,514)
	(6252,524)(6354,535)(6453,547)
	(6549,560)(6643,575)(6734,590)
	(6822,607)(6908,624)(6991,642)
	(7072,661)(7151,680)(7227,700)
	(7302,721)(7375,742)(7447,763)
	(7516,786)(7585,808)(7652,831)
	(7717,854)(7781,878)(7844,902)
	(7906,926)(7966,950)(8025,974)
	(8082,998)(8138,1022)(8192,1045)
	(8244,1068)(8295,1091)(8343,1113)
	(8390,1135)(8434,1156)(8476,1176)
	(8515,1194)(8552,1212)(8586,1229)
	(8618,1244)(8646,1259)(8672,1271)
	(8694,1283)(8714,1293)(8731,1302)
	(8746,1309)(8758,1315)(8767,1320)
	(8775,1324)(8780,1327)(8788,1331)
\blacken\thicklines
\path(8684.023,1237.085)(8788.000,1331.000)(8650.482,1304.167)(8684.023,1237.085)
\put(2991,4152){\makebox(0,0)[lb]{\smash{{\SetFigFont{6}{7.2}{\familydefault}{\mddefault}{\updefault}$2$}}}}
\put(5232,4152){\makebox(0,0)[lb]{\smash{{\SetFigFont{6}{7.2}{\familydefault}{\mddefault}{\updefault}$3$}}}}
\put(7504,4152){\makebox(0,0)[lb]{\smash{{\SetFigFont{6}{7.2}{\familydefault}{\mddefault}{\updefault}$4$}}}}
\put(771,4152){\makebox(0,0)[lb]{\smash{{\SetFigFont{6}{7.2}{\familydefault}{\mddefault}{\updefault}$1$}}}}
\put(1700,4353){\makebox(0,0)[lb]{\smash{{\SetFigFont{7}{8.4}{\familydefault}{\mddefault}{\updefault}$a,b$}}}}
\put(6169,4354){\makebox(0,0)[lb]{\smash{{\SetFigFont{7}{8.4}{\familydefault}{\mddefault}{\updefault}$a,b$}}}}
\put(1318,3918){\makebox(0,0)[lb]{\smash{{\SetFigFont{7}{8.4}{\familydefault}{\mddefault}{\updefault}$e_3,e_4,e_5$}}}}
\put(5847,3911){\makebox(0,0)[lb]{\smash{{\SetFigFont{7}{8.4}{\familydefault}{\mddefault}{\updefault}$e_3,e_4,e_5$}}}}
\put(4076,4354){\makebox(0,0)[lb]{\smash{{\SetFigFont{7}{8.4}{\familydefault}{\mddefault}{\updefault}$a$}}}}
\put(5441,4857){\makebox(0,0)[lb]{\smash{{\SetFigFont{7}{8.4}{\familydefault}{\mddefault}{\updefault}$d_4$}}}}
\put(3178,5074){\makebox(0,0)[lb]{\smash{{\SetFigFont{7}{8.4}{\familydefault}{\mddefault}{\updefault}$d_3$}}}}
\put(2916,3409){\makebox(0,0)[lb]{\smash{{\SetFigFont{7}{8.4}{\familydefault}{\mddefault}{\updefault}$d_3$}}}}
\put(4706,3147){\makebox(0,0)[lb]{\smash{{\SetFigFont{7}{8.4}{\familydefault}{\mddefault}{\updefault}$d_4$}}}}
\put(3193,1317){\makebox(0,0)[lb]{\smash{{\SetFigFont{7}{8.4}{\familydefault}{\mddefault}{\updefault}$d_3,d_4$}}}}
\put(9936,1295){\makebox(0,0)[lb]{\smash{{\SetFigFont{7}{8.4}{\familydefault}{\mddefault}{\updefault}$d_3,d_4$}}}}
\put(4902,979){\makebox(0,0)[lb]{\smash{{\SetFigFont{7}{8.4}{\familydefault}{\mddefault}{\updefault}$e_3$}}}}
\put(4543,1512){\makebox(0,0)[lb]{\smash{{\SetFigFont{6}{7.2}{\familydefault}{\mddefault}{\updefault}$2$}}}}
\put(5621,1722){\makebox(0,0)[lb]{\smash{{\SetFigFont{7}{8.4}{\familydefault}{\mddefault}{\updefault}$a$}}}}
\put(7826,1745){\makebox(0,0)[lb]{\smash{{\SetFigFont{7}{8.4}{\familydefault}{\mddefault}{\updefault}$a$}}}}
\put(9372,2141){\makebox(0,0)[lb]{\smash{{\SetFigFont{7}{8.4}{\familydefault}{\mddefault}{\updefault}$e_5$}}}}
\put(7115,2352){\makebox(0,0)[lb]{\smash{{\SetFigFont{7}{8.4}{\familydefault}{\mddefault}{\updefault}$e_4$}}}}
\put(5104,2509){\makebox(0,0)[lb]{\smash{{\SetFigFont{7}{8.4}{\familydefault}{\mddefault}{\updefault}$e_3$}}}}
\put(10033,1700){\makebox(0,0)[lb]{\smash{{\SetFigFont{7}{8.4}{\familydefault}{\mddefault}{\updefault}$a,c$}}}}
\put(3260,1707){\makebox(0,0)[lb]{\smash{{\SetFigFont{7}{8.4}{\familydefault}{\mddefault}{\updefault}$a,c$}}}}
\put(8082,326){\makebox(0,0)[lb]{\smash{{\SetFigFont{7}{8.4}{\familydefault}{\mddefault}{\updefault}$e_5$}}}}
\put(6485,747){\makebox(0,0)[lb]{\smash{{\SetFigFont{7}{8.4}{\familydefault}{\mddefault}{\updefault}$e_4$}}}}
\put(2278,1512){\makebox(0,0)[lb]{\smash{{\SetFigFont{6}{7.2}{\familydefault}{\mddefault}{\updefault}$1$}}}}
\put(6755,1520){\makebox(0,0)[lb]{\smash{{\SetFigFont{6}{7.2}{\familydefault}{\mddefault}{\updefault}$3$}}}}
\put(9012,1527){\makebox(0,0)[lb]{\smash{{\SetFigFont{6}{7.2}{\familydefault}{\mddefault}{\updefault}$4$}}}}
\put(11285,1512){\makebox(0,0)[lb]{\smash{{\SetFigFont{6}{7.2}{\familydefault}{\mddefault}{\updefault}$5$}}}}
\put(19,3784){\makebox(0,0)[lb]{\smash{{\SetFigFont{7}{8.4}{\familydefault}{\mddefault}{\updefault}$K$}}}}
\thinlines
\put(7588,4219){\ellipse{810}{810}}
\end{picture}
}

%% file: KunionL.eepic
\setlength{\unitlength}{0.00027997in}
\begingroup\makeatletter\ifx\SetFigFont\undefined%
\gdef\SetFigFont#1#2#3#4#5{%
  \reset@font\fontsize{#1}{#2pt}%
  \fontfamily{#3}\fontseries{#4}\fontshape{#5}%
  \selectfont}%
\fi\endgroup%
{\renewcommand{\dashlinestretch}{30}
\begin{picture}(9503,9034)(0,-10)
\put(6878,911){\makebox(0,0)[lb]{\smash{{\SetFigFont{6}{7.2}{\familydefault}{\mddefault}{\updefault}$5,5$}}}}
\put(3871,4882){\ellipse{902}{902}}
\put(5500,4872){\ellipse{902}{902}}
\put(7109,4872){\ellipse{902}{902}}
\put(8728,4872){\ellipse{902}{902}}
\put(3881,2686){\ellipse{902}{902}}
\put(3884,2684){\ellipse{736}{736}}
\put(2239,2704){\ellipse{902}{902}}
\put(2246,2696){\ellipse{736}{736}}
\put(5490,2707){\ellipse{902}{902}}
\put(7119,2716){\ellipse{902}{902}}
\put(8731,2723){\ellipse{902}{902}}
\put(8733,2725){\ellipse{736}{736}}
\put(7111,4876){\ellipse{736}{736}}
\put(3873,1016){\ellipse{902}{902}}
\put(5496,1018){\ellipse{902}{902}}
\put(7121,1016){\ellipse{902}{902}}
\put(2252,4889){\ellipse{902}{902}}
\put(2236,6592){\ellipse{902}{902}}
\put(3862,6591){\ellipse{902}{902}}
\put(5503,6595){\ellipse{902}{902}}
\put(7119,6582){\ellipse{902}{902}}
\put(8709,6586){\ellipse{902}{902}}
\put(592,8292){\ellipse{902}{902}}
\put(5493,2698){\ellipse{736}{736}}
\put(8747,1007){\ellipse{902}{902}}
\put(2174,1013){\ellipse{902}{902}}
\put(7122,2716){\ellipse{736}{736}}
\put(7116,1018){\ellipse{736}{736}}
\path(2520,4541)(8321,2902)
\blacken\thicklines
\path(8180.890,2902.618)(8321.000,2902.000)(8201.282,2974.793)(8230.060,2927.694)(8180.890,2902.618)
\thinlines
\path(4179,4531)(8341,2971)
\blacken\thicklines
\path(8201.426,2983.267)(8341.000,2971.000)(8227.750,3053.496)(8252.512,3004.167)(8201.426,2983.267)
\thinlines
\path(7477,4571)(8441,3091)
\blacken\thicklines
\path(8335.897,3183.653)(8441.000,3091.000)(8398.741,3224.587)(8389.423,3170.184)(8335.897,3183.653)
\thinlines
\path(5868,4591)(8371,3011)
\blacken\thicklines
\path(8236.825,3051.351)(8371.000,3011.000)(8276.859,3114.772)(8291.089,3061.443)(8236.825,3051.351)
\thinlines
\path(8731,4431)(8731,3171)
\blacken\thicklines
\path(8693.500,3306.000)(8731.000,3171.000)(8768.500,3306.000)(8731.000,3265.500)(8693.500,3306.000)
\thinlines
\path(2346,4435)(3516,2995)
\blacken\thicklines
\path(3401.766,3076.128)(3516.000,2995.000)(3459.974,3123.423)(3456.409,3068.343)(3401.766,3076.128)
\thinlines
\path(4001,4472)(5171,3032)
\blacken\thicklines
\path(5056.766,3113.128)(5171.000,3032.000)(5114.974,3160.423)(5111.409,3105.343)(5056.766,3113.128)
\thinlines
\path(5581,4431)(6751,2991)
\blacken\thicklines
\path(6636.766,3072.128)(6751.000,2991.000)(6694.974,3119.423)(6691.409,3064.343)(6636.766,3072.128)
\thinlines
\path(4321,1011)(5041,1011)
\blacken\thicklines
\path(4906.000,973.500)(5041.000,1011.000)(4906.000,1048.500)(4946.500,1011.000)(4906.000,973.500)
\thinlines
\path(5941,1011)(6661,1011)
\blacken\thicklines
\path(6526.000,973.500)(6661.000,1011.000)(6526.000,1048.500)(6566.500,1011.000)(6526.000,973.500)
\thinlines
\path(2521,2316)(3511,1281)
\blacken\thicklines
\path(3390.586,1352.636)(3511.000,1281.000)(3444.784,1404.477)(3445.679,1349.290)(3390.586,1352.636)
\thinlines
\path(2661,1004)(3381,1004)
\blacken\thicklines
\path(3246.000,966.500)(3381.000,1004.000)(3246.000,1041.500)(3286.500,1004.000)(3246.000,966.500)
\thinlines
\path(7566,1004)(8286,1004)
\blacken\thicklines
\path(8151.000,966.500)(8286.000,1004.000)(8151.000,1041.500)(8191.500,1004.000)(8151.000,966.500)
\thinlines
\path(8731,2271)(8731,1461)
\blacken\thicklines
\path(8693.500,1596.000)(8731.000,1461.000)(8768.500,1596.000)(8731.000,1555.500)(8693.500,1596.000)
\thinlines
\path(7378,2308)(8368,1273)
\blacken\thicklines
\path(8247.586,1344.636)(8368.000,1273.000)(8301.784,1396.477)(8302.679,1341.290)(8247.586,1344.636)
\thinlines
\path(5774,2314)(6764,1279)
\blacken\thicklines
\path(6643.586,1350.636)(6764.000,1279.000)(6697.784,1402.477)(6698.679,1347.290)(6643.586,1350.636)
\thinlines
\path(4206,2325)(5196,1290)
\blacken\thicklines
\path(5075.586,1361.636)(5196.000,1290.000)(5129.784,1413.477)(5130.679,1358.290)(5075.586,1361.636)
\thinlines
\path(4143,6193)(5133,5158)
\blacken\thicklines
\path(5012.586,5229.636)(5133.000,5158.000)(5066.784,5281.477)(5067.679,5226.290)(5012.586,5229.636)
\thinlines
\path(5798,6218)(6788,5183)
\blacken\thicklines
\path(6667.586,5254.636)(6788.000,5183.000)(6721.784,5306.477)(6722.679,5251.290)(6667.586,5254.636)
\thinlines
\path(7377,6179)(8367,5144)
\blacken\thicklines
\path(8246.586,5215.636)(8367.000,5144.000)(8300.784,5267.477)(8301.679,5212.290)(8246.586,5215.636)
\thinlines
\path(8739,6125)(8739,5315)
\blacken\thicklines
\path(8701.500,5450.000)(8739.000,5315.000)(8776.500,5450.000)(8739.000,5409.500)(8701.500,5450.000)
\thinlines
\path(2540,6209)(3530,5174)
\blacken\thicklines
\path(3409.586,5245.636)(3530.000,5174.000)(3463.784,5297.477)(3464.679,5242.290)(3409.586,5245.636)
\thinlines
\path(897,7932)(1887,6897)
\blacken\thicklines
\path(1766.586,6968.636)(1887.000,6897.000)(1820.784,7020.477)(1821.679,6965.290)(1766.586,6968.636)
\thinlines
\path(5689,6186)(6941,1447)
\blacken\thicklines
\path(6870.261,1567.943)(6941.000,1447.000)(6942.773,1587.100)(6916.862,1538.365)(6870.261,1567.943)
\thinlines
\path(5510,4408)(6831,1388)
\blacken\thicklines
\path(6742.541,1496.657)(6831.000,1388.000)(6811.255,1526.713)(6793.129,1474.579)(6742.541,1496.657)
\path(22,8997)(307,8652)
\blacken\path(192.110,8732.197)(307.000,8652.000)(249.932,8779.963)(246.815,8724.856)(192.110,8732.197)
\thinlines
\path(991,8481)(992,8481)(993,8481)
	(995,8482)(999,8482)(1005,8483)
	(1013,8485)(1023,8486)(1036,8488)
	(1051,8491)(1070,8494)(1091,8497)
	(1116,8501)(1144,8505)(1175,8510)
	(1210,8515)(1247,8521)(1288,8527)
	(1332,8533)(1379,8540)(1429,8547)
	(1482,8554)(1537,8561)(1595,8568)
	(1656,8576)(1718,8583)(1782,8591)
	(1849,8599)(1917,8606)(1987,8613)
	(2058,8620)(2131,8627)(2205,8634)
	(2281,8640)(2358,8645)(2436,8651)
	(2515,8656)(2595,8660)(2677,8664)
	(2760,8667)(2844,8670)(2930,8671)
	(3017,8673)(3106,8673)(3196,8673)
	(3288,8671)(3382,8669)(3478,8666)
	(3576,8661)(3676,8656)(3778,8649)
	(3883,8641)(3990,8632)(4099,8621)
	(4211,8609)(4325,8596)(4441,8581)
	(4558,8564)(4678,8546)(4798,8526)
	(4920,8504)(5041,8481)(5166,8455)
	(5289,8429)(5411,8400)(5530,8371)
	(5646,8341)(5759,8311)(5870,8279)
	(5977,8248)(6080,8215)(6181,8183)
	(6278,8150)(6373,8117)(6464,8084)
	(6553,8051)(6639,8017)(6722,7984)
	(6804,7950)(6883,7916)(6959,7882)
	(7034,7848)(7108,7814)(7179,7780)
	(7249,7745)(7317,7711)(7383,7677)
	(7449,7642)(7512,7608)(7575,7574)
	(7635,7540)(7695,7506)(7752,7473)
	(7809,7440)(7863,7408)(7916,7376)
	(7967,7345)(8016,7315)(8063,7286)
	(8108,7258)(8151,7230)(8192,7205)
	(8230,7180)(8266,7157)(8299,7135)
	(8329,7115)(8357,7097)(8382,7080)
	(8405,7065)(8425,7051)(8442,7040)
	(8457,7030)(8470,7021)(8480,7014)
	(8488,7008)(8495,7004)(8506,6996)
\blacken\thicklines
\path(8374.764,7045.076)(8506.000,6996.000)(8418.877,7105.731)(8429.574,7051.582)(8374.764,7045.076)
\thinlines
\path(1036,8346)(1037,8346)(1038,8346)
	(1041,8347)(1045,8347)(1051,8348)
	(1060,8349)(1070,8351)(1084,8352)
	(1100,8355)(1120,8357)(1142,8360)
	(1168,8363)(1197,8367)(1229,8371)
	(1265,8375)(1303,8379)(1344,8384)
	(1389,8389)(1436,8394)(1485,8399)
	(1537,8404)(1591,8409)(1647,8414)
	(1705,8419)(1765,8424)(1827,8429)
	(1890,8433)(1954,8437)(2020,8441)
	(2086,8445)(2154,8448)(2223,8450)
	(2294,8452)(2365,8454)(2438,8455)
	(2511,8455)(2586,8455)(2663,8453)
	(2741,8451)(2820,8448)(2901,8445)
	(2984,8440)(3068,8434)(3155,8427)
	(3243,8419)(3334,8409)(3427,8398)
	(3521,8386)(3618,8372)(3717,8357)
	(3818,8340)(3920,8321)(4023,8301)
	(4127,8279)(4231,8256)(4342,8229)
	(4451,8201)(4558,8172)(4663,8141)
	(4764,8110)(4861,8079)(4956,8047)
	(5047,8015)(5134,7982)(5218,7950)
	(5300,7917)(5378,7884)(5453,7851)
	(5526,7818)(5596,7785)(5664,7752)
	(5729,7718)(5793,7685)(5855,7652)
	(5915,7619)(5974,7585)(6031,7552)
	(6086,7519)(6140,7486)(6192,7453)
	(6243,7421)(6293,7389)(6340,7357)
	(6386,7326)(6431,7296)(6473,7267)
	(6513,7239)(6552,7211)(6588,7186)
	(6621,7161)(6653,7138)(6682,7117)
	(6708,7097)(6732,7080)(6753,7064)
	(6772,7049)(6788,7037)(6801,7027)
	(6813,7018)(6822,7011)(6828,7006)
	(6834,7002)(6841,6996)
\blacken\thicklines
\path(6714.096,7055.385)(6841.000,6996.000)(6762.905,7112.329)(6769.250,7057.500)(6714.096,7055.385)
\thinlines
\path(1036,8211)(1037,8211)(1038,8211)
	(1042,8211)(1047,8212)(1054,8212)
	(1064,8212)(1076,8213)(1092,8214)
	(1111,8215)(1133,8215)(1159,8217)
	(1188,8218)(1221,8219)(1257,8220)
	(1297,8222)(1339,8223)(1384,8224)
	(1433,8226)(1483,8227)(1536,8228)
	(1592,8229)(1649,8230)(1707,8230)
	(1768,8231)(1829,8231)(1892,8230)
	(1956,8230)(2020,8229)(2086,8227)
	(2152,8225)(2219,8223)(2287,8220)
	(2355,8217)(2425,8212)(2495,8207)
	(2565,8202)(2637,8195)(2709,8188)
	(2783,8180)(2858,8171)(2933,8160)
	(3010,8149)(3088,8136)(3167,8122)
	(3247,8107)(3328,8090)(3410,8072)
	(3492,8053)(3574,8032)(3656,8010)
	(3736,7986)(3833,7955)(3926,7923)
	(4015,7890)(4099,7856)(4178,7822)
	(4252,7788)(4321,7754)(4386,7719)
	(4446,7685)(4503,7651)(4555,7617)
	(4605,7583)(4651,7550)(4694,7516)
	(4735,7483)(4773,7449)(4809,7416)
	(4843,7383)(4875,7350)(4906,7318)
	(4934,7286)(4961,7254)(4986,7224)
	(5010,7194)(5032,7165)(5053,7137)
	(5071,7111)(5089,7087)(5104,7064)
	(5118,7043)(5130,7025)(5141,7008)
	(5150,6994)(5157,6982)(5163,6973)
	(5168,6965)(5171,6959)(5176,6951)
\blacken\thicklines
\path(5072.650,7045.605)(5176.000,6951.000)(5136.250,7085.355)(5125.915,7031.136)(5072.650,7045.605)
\thinlines
\path(226,7986)(226,7985)(226,7984)
	(225,7982)(224,7978)(223,7972)
	(221,7965)(219,7955)(217,7943)
	(214,7928)(210,7910)(206,7890)
	(201,7867)(196,7840)(191,7811)
	(184,7778)(178,7743)(171,7705)
	(163,7664)(156,7620)(148,7574)
	(140,7525)(131,7474)(123,7421)
	(115,7365)(106,7308)(98,7249)
	(90,7189)(83,7127)(75,7063)
	(68,6998)(61,6932)(55,6864)
	(49,6795)(44,6725)(40,6653)
	(36,6580)(33,6505)(30,6429)
	(29,6351)(28,6271)(29,6190)
	(30,6106)(33,6021)(36,5933)
	(41,5843)(48,5750)(55,5655)
	(64,5557)(75,5457)(88,5353)
	(102,5247)(118,5139)(136,5028)
	(155,4914)(177,4799)(201,4683)
	(226,4566)(252,4453)(280,4341)
	(309,4229)(339,4120)(370,4013)
	(402,3908)(434,3805)(467,3705)
	(500,3608)(533,3513)(567,3421)
	(601,3332)(635,3244)(669,3159)
	(704,3077)(738,2996)(773,2917)
	(808,2840)(843,2765)(878,2691)
	(913,2619)(949,2549)(984,2479)
	(1019,2411)(1055,2345)(1090,2279)
	(1125,2215)(1161,2152)(1195,2091)
	(1230,2031)(1264,1972)(1298,1914)
	(1331,1858)(1364,1804)(1396,1752)
	(1427,1701)(1457,1652)(1487,1606)
	(1514,1561)(1541,1519)(1566,1480)
	(1590,1443)(1613,1408)(1633,1376)
	(1652,1347)(1670,1321)(1685,1297)
	(1699,1276)(1711,1258)(1721,1242)
	(1730,1229)(1738,1218)(1743,1210)
	(1748,1203)(1756,1191)
\blacken\thicklines
\path(1649.914,1282.526)(1756.000,1191.000)(1712.317,1324.128)(1703.581,1269.629)(1649.914,1282.526)
\thinlines
\path(406,7851)(406,7850)(406,7849)
	(405,7846)(405,7841)(404,7833)
	(402,7824)(401,7811)(398,7795)
	(396,7776)(393,7754)(390,7729)
	(387,7700)(383,7668)(379,7632)
	(374,7593)(370,7551)(365,7507)
	(360,7459)(356,7410)(351,7357)
	(347,7303)(342,7247)(338,7189)
	(334,7130)(331,7069)(328,7006)
	(325,6943)(323,6878)(322,6812)
	(321,6745)(321,6676)(322,6607)
	(324,6536)(327,6463)(331,6389)
	(335,6313)(342,6236)(349,6157)
	(358,6075)(368,5992)(380,5906)
	(393,5818)(409,5728)(426,5636)
	(445,5541)(466,5445)(489,5348)
	(514,5249)(541,5151)(570,5053)
	(600,4957)(631,4864)(663,4772)
	(696,4684)(730,4599)(763,4517)
	(797,4438)(831,4362)(866,4289)
	(900,4219)(934,4151)(969,4086)
	(1003,4024)(1038,3963)(1072,3905)
	(1107,3848)(1141,3793)(1176,3740)
	(1210,3688)(1245,3638)(1279,3589)
	(1313,3542)(1347,3496)(1381,3451)
	(1414,3408)(1446,3366)(1478,3326)
	(1509,3287)(1539,3250)(1568,3215)
	(1595,3182)(1621,3151)(1646,3122)
	(1669,3095)(1690,3070)(1710,3048)
	(1727,3029)(1742,3011)(1756,2996)
	(1767,2983)(1776,2973)(1784,2964)
	(1790,2958)(1801,2946)
\blacken\thicklines
\path(1682.134,3020.176)(1801.000,2946.000)(1737.420,3070.855)(1737.144,3015.661)(1682.134,3020.176)
\thinlines
\path(586,7851)(586,7850)(586,7848)
	(586,7844)(586,7838)(586,7829)
	(587,7818)(587,7803)(587,7784)
	(588,7762)(589,7737)(590,7708)
	(591,7676)(592,7640)(594,7602)
	(595,7561)(598,7518)(600,7472)
	(603,7425)(606,7376)(610,7325)
	(614,7273)(619,7221)(624,7167)
	(630,7113)(637,7057)(644,7001)
	(652,6945)(661,6887)(671,6829)
	(682,6769)(694,6709)(708,6647)
	(723,6585)(739,6521)(757,6456)
	(777,6390)(798,6323)(821,6255)
	(846,6186)(873,6118)(901,6051)
	(936,5975)(971,5902)(1007,5834)
	(1044,5770)(1080,5710)(1116,5655)
	(1152,5604)(1187,5556)(1222,5513)
	(1257,5472)(1291,5434)(1325,5399)
	(1358,5367)(1392,5336)(1425,5308)
	(1457,5281)(1490,5256)(1521,5233)
	(1552,5211)(1582,5190)(1611,5171)
	(1638,5154)(1664,5138)(1688,5123)
	(1710,5111)(1730,5099)(1747,5090)
	(1761,5082)(1773,5075)(1783,5070)
	(1790,5067)(1801,5061)
\blacken\thicklines
\path(1664.527,5092.724)(1801.000,5061.000)(1700.441,5158.566)(1718.039,5106.252)(1664.527,5092.724)
\thinlines
\path(991,8076)(992,8076)(994,8076)
	(998,8075)(1004,8074)(1012,8073)
	(1024,8071)(1038,8069)(1056,8066)
	(1078,8063)(1103,8059)(1131,8054)
	(1162,8049)(1197,8043)(1234,8037)
	(1274,8030)(1316,8023)(1360,8015)
	(1406,8007)(1454,7998)(1503,7989)
	(1553,7979)(1604,7969)(1655,7958)
	(1708,7947)(1761,7935)(1815,7923)
	(1869,7910)(1924,7896)(1979,7881)
	(2035,7866)(2092,7850)(2150,7833)
	(2209,7815)(2269,7795)(2330,7775)
	(2391,7753)(2453,7730)(2516,7706)
	(2578,7680)(2640,7653)(2701,7626)
	(2778,7589)(2850,7552)(2917,7515)
	(2978,7479)(3034,7444)(3084,7409)
	(3130,7376)(3172,7344)(3209,7312)
	(3243,7282)(3274,7252)(3302,7222)
	(3328,7194)(3352,7166)(3373,7138)
	(3393,7111)(3411,7085)(3427,7061)
	(3442,7037)(3455,7015)(3466,6994)
	(3477,6976)(3485,6959)(3492,6945)
	(3498,6933)(3503,6924)(3506,6917)(3511,6906)
\blacken\thicklines
\path(3420.998,7013.382)(3511.000,6906.000)(3489.275,7044.417)(3471.896,6992.030)(3420.998,7013.382)
\put(8821,3711){\makebox(0,0)[lb]{\smash{{\SetFigFont{8}{9.6}{\familydefault}{\mddefault}{\updefault}$a,b$}}}}
\put(2001,6484){\makebox(0,0)[lb]{\smash{{\SetFigFont{6}{7.2}{\familydefault}{\mddefault}{\updefault}$2,2$}}}}
\put(3639,6484){\makebox(0,0)[lb]{\smash{{\SetFigFont{6}{7.2}{\familydefault}{\mddefault}{\updefault}$2,3$}}}}
\put(5269,6484){\makebox(0,0)[lb]{\smash{{\SetFigFont{6}{7.2}{\familydefault}{\mddefault}{\updefault}$2,4$}}}}
\put(6878,6474){\makebox(0,0)[lb]{\smash{{\SetFigFont{6}{7.2}{\familydefault}{\mddefault}{\updefault}$2,5$}}}}
\put(8477,6473){\makebox(0,0)[lb]{\smash{{\SetFigFont{6}{7.2}{\familydefault}{\mddefault}{\updefault}$2,6$}}}}
\put(2021,4775){\makebox(0,0)[lb]{\smash{{\SetFigFont{6}{7.2}{\familydefault}{\mddefault}{\updefault}$3,2$}}}}
\put(1991,2600){\makebox(0,0)[lb]{\smash{{\SetFigFont{6}{7.2}{\familydefault}{\mddefault}{\updefault}$4,2$}}}}
\put(3630,2570){\makebox(0,0)[lb]{\smash{{\SetFigFont{6}{7.2}{\familydefault}{\mddefault}{\updefault}$4,3$}}}}
\put(5268,2600){\makebox(0,0)[lb]{\smash{{\SetFigFont{6}{7.2}{\familydefault}{\mddefault}{\updefault}$4,4$}}}}
\put(6888,2610){\makebox(0,0)[lb]{\smash{{\SetFigFont{6}{7.2}{\familydefault}{\mddefault}{\updefault}$4,5$}}}}
\put(8477,2620){\makebox(0,0)[lb]{\smash{{\SetFigFont{6}{7.2}{\familydefault}{\mddefault}{\updefault}$4,6$}}}}
\put(1941,931){\makebox(0,0)[lb]{\smash{{\SetFigFont{6}{7.2}{\familydefault}{\mddefault}{\updefault}$5,2$}}}}
\put(3649,931){\makebox(0,0)[lb]{\smash{{\SetFigFont{6}{7.2}{\familydefault}{\mddefault}{\updefault}$5,3$}}}}
\put(5278,921){\makebox(0,0)[lb]{\smash{{\SetFigFont{6}{7.2}{\familydefault}{\mddefault}{\updefault}$5,4$}}}}
\put(8517,911){\makebox(0,0)[lb]{\smash{{\SetFigFont{6}{7.2}{\familydefault}{\mddefault}{\updefault}$5,6$}}}}
\put(8497,4765){\makebox(0,0)[lb]{\smash{{\SetFigFont{6}{7.2}{\familydefault}{\mddefault}{\updefault}$3,6$}}}}
\put(6868,4765){\makebox(0,0)[lb]{\smash{{\SetFigFont{6}{7.2}{\familydefault}{\mddefault}{\updefault}$3,5$}}}}
\put(5259,4785){\makebox(0,0)[lb]{\smash{{\SetFigFont{6}{7.2}{\familydefault}{\mddefault}{\updefault}$3,4$}}}}
\put(3650,4785){\makebox(0,0)[lb]{\smash{{\SetFigFont{6}{7.2}{\familydefault}{\mddefault}{\updefault}$3,3$}}}}
\put(362,8182){\makebox(0,0)[lb]{\smash{{\SetFigFont{6}{7.2}{\familydefault}{\mddefault}{\updefault}$1,1$}}}}
\put(989,7281){\makebox(0,0)[lb]{\smash{{\SetFigFont{8}{9.6}{\familydefault}{\mddefault}{\updefault}$a$}}}}
\put(4218,5593){\makebox(0,0)[lb]{\smash{{\SetFigFont{8}{9.6}{\familydefault}{\mddefault}{\updefault}$a$}}}}
\put(7417,5593){\makebox(0,0)[lb]{\smash{{\SetFigFont{8}{9.6}{\familydefault}{\mddefault}{\updefault}$a$}}}}
\put(2569,5593){\makebox(0,0)[lb]{\smash{{\SetFigFont{8}{9.6}{\familydefault}{\mddefault}{\updefault}$a$}}}}
\put(2638,1629){\makebox(0,0)[lb]{\smash{{\SetFigFont{8}{9.6}{\familydefault}{\mddefault}{\updefault}$a$}}}}
\put(4317,1639){\makebox(0,0)[lb]{\smash{{\SetFigFont{8}{9.6}{\familydefault}{\mddefault}{\updefault}$a$}}}}
\put(5940,641){\makebox(0,0)[lb]{\smash{{\SetFigFont{8}{9.6}{\familydefault}{\mddefault}{\updefault}$a,c$}}}}
\put(4555,685){\makebox(0,0)[lb]{\smash{{\SetFigFont{8}{9.6}{\familydefault}{\mddefault}{\updefault}$a$}}}}
\put(2916,685){\makebox(0,0)[lb]{\smash{{\SetFigFont{8}{9.6}{\familydefault}{\mddefault}{\updefault}$a$}}}}
\put(7793,675){\makebox(0,0)[lb]{\smash{{\SetFigFont{8}{9.6}{\familydefault}{\mddefault}{\updefault}$a$}}}}
\put(4537,3159){\makebox(0,0)[lb]{\smash{{\SetFigFont{8}{9.6}{\familydefault}{\mddefault}{\updefault}$a$}}}}
\put(2868,3160){\makebox(0,0)[lb]{\smash{{\SetFigFont{8}{9.6}{\familydefault}{\mddefault}{\updefault}$a$}}}}
\put(7466,1638){\makebox(0,0)[lb]{\smash{{\SetFigFont{8}{9.6}{\familydefault}{\mddefault}{\updefault}$a$}}}}
\put(5432,5650){\makebox(0,0)[lb]{\smash{{\SetFigFont{8}{9.6}{\familydefault}{\mddefault}{\updefault}$c$}}}}
\put(7841,4108){\makebox(0,0)[lb]{\smash{{\SetFigFont{8}{9.6}{\familydefault}{\mddefault}{\updefault}$a,b$}}}}
\put(6404,5671){\makebox(0,0)[lb]{\smash{{\SetFigFont{8}{9.6}{\familydefault}{\mddefault}{\updefault}$a$}}}}
\put(9454,136){\makebox(0,0)[lb]{\smash{{\SetFigFont{8}{9.6}{\familydefault}{\mddefault}{\updefault}$\Sigma$}}}}
\put(4816,7536){\makebox(0,0)[lb]{\smash{{\SetFigFont{8}{9.6}{\familydefault}{\mddefault}{\updefault}$e_4$}}}}
\put(7876,7536){\makebox(0,0)[lb]{\smash{{\SetFigFont{8}{9.6}{\familydefault}{\mddefault}{\updefault}$b$}}}}
\put(6211,7536){\makebox(0,0)[lb]{\smash{{\SetFigFont{8}{9.6}{\familydefault}{\mddefault}{\updefault}$e_5$}}}}
\put(1126,4071){\makebox(0,0)[lb]{\smash{{\SetFigFont{8}{9.6}{\familydefault}{\mddefault}{\updefault}$d_4$}}}}
\put(901,2901){\makebox(0,0)[lb]{\smash{{\SetFigFont{8}{9.6}{\familydefault}{\mddefault}{\updefault}$c$}}}}
\put(1126,5826){\makebox(0,0)[lb]{\smash{{\SetFigFont{8}{9.6}{\familydefault}{\mddefault}{\updefault}$d_3$}}}}
\put(3016,7536){\makebox(0,0)[lb]{\smash{{\SetFigFont{8}{9.6}{\familydefault}{\mddefault}{\updefault}$e_3$}}}}
\put(4591,4476){\makebox(0,0)[lb]{\smash{{\SetFigFont{8}{9.6}{\familydefault}{\mddefault}{\updefault}$b$}}}}
\put(7021,3936){\makebox(0,0)[lb]{\smash{{\SetFigFont{8}{9.6}{\familydefault}{\mddefault}{\updefault}$b$}}}}
\put(5986,4071){\makebox(0,0)[lb]{\smash{{\SetFigFont{8}{9.6}{\familydefault}{\mddefault}{\updefault}$a$}}}}
\put(3016,4476){\makebox(0,0)[lb]{\smash{{\SetFigFont{8}{9.6}{\familydefault}{\mddefault}{\updefault}$b$}}}}
\put(5561,3316){\makebox(0,0)[lb]{\smash{{\SetFigFont{8}{9.6}{\familydefault}{\mddefault}{\updefault}$c$}}}}
\put(5647,1599){\makebox(0,0)[lb]{\smash{{\SetFigFont{8}{9.6}{\familydefault}{\mddefault}{\updefault}$a,c$}}}}
\put(8880,5591){\makebox(0,0)[lb]{\smash{{\SetFigFont{8}{9.6}{\familydefault}{\mddefault}{\updefault}$a$}}}}
\put(8880,1745){\makebox(0,0)[lb]{\smash{{\SetFigFont{8}{9.6}{\familydefault}{\mddefault}{\updefault}$a$}}}}
\thinlines
\put(9233.634,667.976){\arc{521.428}{4.5985}{9.2084}}
\blacken\thicklines
\path(9075.372,510.296)(8979.000,612.000)(9009.110,475.162)(9023.269,528.510)(9075.372,510.296)
\end{picture}
}

%% file: sebej3.eepic
\setlength{\unitlength}{0.00052493in}
\begingroup\makeatletter\ifx\SetFigFont\undefined%
\gdef\SetFigFont#1#2#3#4#5{%
  \reset@font\fontsize{#1}{#2pt}%
  \fontfamily{#3}\fontseries{#4}\fontshape{#5}%
  \selectfont}%
\fi\endgroup%
{\renewcommand{\dashlinestretch}{30}
\begin{picture}(9352,1706)(0,-10)
\put(6402,245){\makebox(0,0)[lb]{\smash{{\SetFigFont{8}{9.6}{\familydefault}{\mddefault}{\updefault}$a$}}}}
\put(7879.000,1263.000){\arc{332.800}{2.2155}{7.2093}}
\path(8016.379,1226.062)(7979.000,1130.000)(8057.187,1197.170)
\put(5987.000,1271.000){\arc{332.800}{2.2155}{7.2093}}
\path(6124.379,1234.062)(6087.000,1138.000)(6165.187,1205.170)
\put(2755,830){\ellipse{630}{630}}
\put(8984,833){\ellipse{720}{720}}
\put(1685,828){\ellipse{630}{630}}
\put(8985,831){\ellipse{630}{630}}
\put(3848,820){\ellipse{630}{630}}
\put(4891,820){\ellipse{630}{630}}
\put(5986,828){\ellipse{630}{630}}
\put(7891,821){\ellipse{630}{630}}
\put(597,829){\ellipse{630}{630}}
\path(1992,830)(2442,830)
\path(2322.000,800.000)(2442.000,830.000)(2322.000,860.000)
\path(3072,830)(3522,830)
\path(3402.000,800.000)(3522.000,830.000)(3402.000,860.000)
\path(8217,823)(8615,823)
\path(8495.000,793.000)(8615.000,823.000)(8495.000,853.000)
\path(5217,830)(5667,830)
\path(5547.000,800.000)(5667.000,830.000)(5547.000,860.000)
\path(6289,830)(6739,830)
\path(6619.000,800.000)(6739.000,830.000)(6619.000,860.000)
\path(7137,830)(7587,830)
\path(7467.000,800.000)(7587.000,830.000)(7467.000,860.000)
\path(919,830)(1369,830)
\path(1249.000,800.000)(1369.000,830.000)(1249.000,860.000)
\path(12,838)(282,838)
\path(162.000,808.000)(282.000,838.000)(162.000,868.000)
\path(2562,1093)(2561,1094)(2559,1098)
	(2555,1103)(2549,1111)(2540,1122)
	(2530,1136)(2518,1151)(2504,1167)
	(2489,1184)(2473,1200)(2455,1217)
	(2437,1233)(2417,1247)(2395,1261)
	(2371,1272)(2346,1282)(2317,1290)
	(2287,1294)(2254,1295)(2221,1292)
	(2189,1285)(2159,1276)(2131,1264)
	(2105,1250)(2080,1235)(2057,1219)
	(2035,1202)(2014,1184)(1994,1166)
	(1975,1148)(1958,1131)(1943,1115)
	(1930,1101)(1919,1090)(1902,1070)
\path(1956.860,1180.862)(1902.000,1070.000)(2002.576,1142.003)
\path(4099,1017)(4101,1018)(4105,1021)
	(4111,1026)(4121,1032)(4133,1040)
	(4148,1049)(4164,1059)(4183,1069)
	(4202,1079)(4224,1088)(4247,1096)
	(4273,1103)(4302,1109)(4334,1113)
	(4369,1115)(4404,1114)(4437,1111)
	(4467,1106)(4495,1099)(4520,1092)
	(4544,1083)(4566,1074)(4587,1065)
	(4606,1056)(4623,1047)(4637,1039)(4662,1025)
\path(4542.641,1057.457)(4662.000,1025.000)(4571.957,1109.808)
\path(4647,613)(4645,612)(4641,609)
	(4635,604)(4625,598)(4613,590)
	(4598,581)(4582,571)(4563,561)
	(4544,551)(4522,542)(4499,534)
	(4473,527)(4444,521)(4412,517)
	(4377,515)(4342,516)(4309,519)
	(4279,524)(4251,531)(4226,538)
	(4202,547)(4180,556)(4159,565)
	(4140,574)(4123,583)(4109,591)(4084,605)
\path(4203.359,572.543)(4084.000,605.000)(4174.043,520.192)
\path(3580,629)(3579,628)(3577,626)
	(3573,622)(3566,616)(3557,608)
	(3545,597)(3531,585)(3513,569)
	(3493,553)(3471,534)(3446,515)
	(3420,494)(3392,474)(3362,453)
	(3331,432)(3298,412)(3264,392)
	(3228,373)(3191,355)(3151,338)
	(3108,323)(3063,309)(3015,296)
	(2963,285)(2909,277)(2852,271)
	(2793,269)(2738,270)(2683,273)
	(2630,279)(2579,287)(2530,297)
	(2484,309)(2440,322)(2398,336)
	(2357,351)(2318,367)(2281,384)
	(2245,402)(2210,420)(2176,438)
	(2143,457)(2112,476)(2082,494)
	(2054,512)(2028,529)(2005,544)
	(1984,558)(1966,571)(1951,581)
	(1938,590)(1929,597)(1915,607)
\path(2030.085,561.663)(1915.000,607.000)(1995.211,512.839)
\path(7712,546)(7711,545)(7709,544)
	(7706,541)(7701,537)(7693,532)
	(7683,524)(7671,515)(7655,504)
	(7637,491)(7616,477)(7593,461)
	(7568,444)(7540,426)(7510,407)
	(7478,388)(7445,369)(7410,349)
	(7373,330)(7335,311)(7295,292)
	(7253,274)(7210,256)(7164,239)
	(7116,224)(7065,209)(7011,195)
	(6955,182)(6894,170)(6830,161)
	(6763,152)(6692,146)(6618,142)
	(6542,141)(6473,142)(6405,145)
	(6337,150)(6271,157)(6207,165)
	(6145,174)(6085,185)(6027,197)
	(5971,209)(5916,223)(5864,237)
	(5812,252)(5762,267)(5714,283)
	(5666,300)(5620,317)(5574,334)
	(5530,351)(5487,369)(5446,386)
	(5406,403)(5367,420)(5331,437)
	(5296,452)(5264,467)(5235,481)
	(5209,494)(5185,505)(5165,515)
	(5148,523)(5134,530)(5123,536)
	(5114,540)(5102,546)
\path(5222.748,519.167)(5102.000,546.000)(5195.915,465.502)
\put(1624,763){\makebox(0,0)[lb]{\smash{{\SetFigFont{8}{9.6}{\rmdefault}{\mddefault}{\updefault}$1$}}}}
\put(2711,763){\makebox(0,0)[lb]{\smash{{\SetFigFont{8}{9.6}{\rmdefault}{\mddefault}{\updefault}$2$}}}}
\put(1617,1520){\makebox(0,0)[lb]{\smash{{\SetFigFont{8}{9.6}{\familydefault}{\mddefault}{\updefault}$b$}}}}
\put(7256,935){\makebox(0,0)[lb]{\smash{{\SetFigFont{8}{9.6}{\familydefault}{\mddefault}{\updefault}$a$}}}}
\put(3199,935){\makebox(0,0)[lb]{\smash{{\SetFigFont{8}{9.6}{\familydefault}{\mddefault}{\updefault}$a$}}}}
\put(2127,935){\makebox(0,0)[lb]{\smash{{\SetFigFont{8}{9.6}{\familydefault}{\mddefault}{\updefault}$a$}}}}
\put(2217,1370){\makebox(0,0)[lb]{\smash{{\SetFigFont{8}{9.6}{\familydefault}{\mddefault}{\updefault}$b$}}}}
\put(3791,770){\makebox(0,0)[lb]{\smash{{\SetFigFont{8}{9.6}{\rmdefault}{\mddefault}{\updefault}$3$}}}}
\put(4826,770){\makebox(0,0)[lb]{\smash{{\SetFigFont{8}{9.6}{\rmdefault}{\mddefault}{\updefault}$4$}}}}
\put(4294,1235){\makebox(0,0)[lb]{\smash{{\SetFigFont{8}{9.6}{\familydefault}{\mddefault}{\updefault}$b$}}}}
\put(4317,275){\makebox(0,0)[lb]{\smash{{\SetFigFont{8}{9.6}{\familydefault}{\mddefault}{\updefault}$b$}}}}
\put(2795,73){\makebox(0,0)[lb]{\smash{{\SetFigFont{8}{9.6}{\familydefault}{\mddefault}{\updefault}$a$}}}}
\put(5337,935){\makebox(0,0)[lb]{\smash{{\SetFigFont{8}{9.6}{\familydefault}{\mddefault}{\updefault}$a$}}}}
\put(5928,778){\makebox(0,0)[lb]{\smash{{\SetFigFont{8}{9.6}{\rmdefault}{\mddefault}{\updefault}$5$}}}}
\put(6807,792){\makebox(0,0)[lb]{\smash{{\SetFigFont{8}{9.6}{\familydefault}{\mddefault}{\updefault}$\cdots$}}}}
\put(6461,928){\makebox(0,0)[lb]{\smash{{\SetFigFont{8}{9.6}{\familydefault}{\mddefault}{\updefault}$a$}}}}
\put(8313,935){\makebox(0,0)[lb]{\smash{{\SetFigFont{8}{9.6}{\familydefault}{\mddefault}{\updefault}$c$}}}}
\put(7805,1520){\makebox(0,0)[lb]{\smash{{\SetFigFont{8}{9.6}{\familydefault}{\mddefault}{\updefault}$b$}}}}
\put(5937,1520){\makebox(0,0)[lb]{\smash{{\SetFigFont{8}{9.6}{\familydefault}{\mddefault}{\updefault}$b$}}}}
\put(8725,767){\makebox(0,0)[lb]{\smash{{\SetFigFont{7}{8.4}{\familydefault}{\mddefault}{\updefault}$n-2$}}}}
\put(7637,767){\makebox(0,0)[lb]{\smash{{\SetFigFont{7}{8.4}{\familydefault}{\mddefault}{\updefault}$n-3$}}}}
\put(1039,928){\makebox(0,0)[lb]{\smash{{\SetFigFont{8}{9.6}{\familydefault}{\mddefault}{\updefault}$c$}}}}
\put(537,763){\makebox(0,0)[lb]{\smash{{\SetFigFont{8}{9.6}{\rmdefault}{\mddefault}{\updefault}$0$}}}}
\put(1694.000,1248.000){\arc{332.800}{2.2155}{7.2093}}
\path(1831.379,1211.062)(1794.000,1115.000)(1872.187,1182.170)
\end{picture}
}

%% file: Quotients_Free.bbl
\begin{thebibliography}{10}

\bibitem{AnBr09}
Ang, T., Brzozowski, J.:
\newblock Languages convex with respect to binary relations, and their closure
  properties.
\newblock Acta Cybernet. \textbf{19}(2) (2009)  445--464

\bibitem{BGN10}
Bassino, F., Giambruno, L., Nicaud, C.:
\newblock Complexity of operations on cofinite languages.
\newblock In L\'opez-Ortiz, A., ed.: Proceedings of the 9th Latin American
  Theoretical Informatics Symposium, $($LATIN\/$)$. Volume 6034 of LNCS,
  Springer (2010)  222--233

\bibitem{BPR10}
Berstel, J., Perrin, D., Reutenauer, C.:
\newblock Codes and Automata (Encyclopedia of Mathematics and its
  Applications).
\newblock Cambridge University Press (2010)

\bibitem{Brz64}
Brzozowski, J.:
\newblock Derivatives of regular expressions.
\newblock J. ACM \textbf{11}(4) (1964)  481--494

\bibitem{Brz10}
Brzozowski, J.:
\newblock Quotient complexity of regular languages.
\newblock In Dassow, J., Pighizzini, G., Truthe, B., eds.: Proceedings of the
  11th International Workshop on Descriptional Complexity of Formal Systems,
  Magdeburg, Germany, Otto-von-Guericke-Universit{\"a}t (2009)  25--42. To
  appear in J. Autom. Lang. Comb. (Extended abstract at {\tt
  http://arxiv.org/abs/0907.4547}).

\bibitem{BJL10}
Brzozowski, J., Jir{\'a}skov{\'a}, G., Li, B.:
\newblock Quotient complexity of ideal languages.
\newblock In L\'opez-Ortiz, A., ed.: Proceedings of the 9th Latin American
  Theoretical Informatics Symposium, $($LATIN\/$)$. Volume 6034 of LNCS,
  Springer (2010)  208--211

\bibitem{BJZ10}
Brzozowski, J., Jir{\'a}skov{\'a}, G., Zou, C.:
\newblock Quotient complexity of closed languages.
\newblock In Ablayev, F., Mayr, E.W., eds.: Proceedings of the 5th
  International Computer Science Symposium in Russia, $($CSR\/$)$. Volume 6072
  of LNCS, Springer (2010)  84--95

\bibitem{CCSY01}
C\^ampeanu, C., Culik~II, K., Salomaa, K., Yu, S.:
\newblock State complexity of basic operations on finite languages.
\newblock In Boldt, O., J{\"u}rgensen, H., eds.: Revised Papers from the 4th
  International Workshop on Automata Implementation, $($WIA\/$)$. Volume 2214
  of LNCS, Springer (2001)  60--70

\bibitem{Cmo11}
Cmorik, R.:
\newblock State complexity of basic operations on binary and ternary
  suffix-free languages.
\newblock Manuscript (2011)

\bibitem{HaSa08}
Han, Y.S., Salomaa, K.:
\newblock State complexity of union and intersection of finite languages.
\newblock Internat.\ J.\ Found.\ Comput.\ Sci. \textbf{19}(3) (2008)  581--595

\bibitem{HaSa09}
Han, Y.S., Salomaa, K.:
\newblock State complexity of basic operations on suffix-free regular
  languages.
\newblock Theoret. Comput. Sci. \textbf{410}(27-29) (2009)  2537--2548

\bibitem{HSW09}
Han, Y.S., Salomaa, K., Wood, D.:
\newblock Operational state complexity of prefix-free regular languages.
\newblock In {\'E}sik, Z., F{\"u}l{\"o}p, Z., eds.: Automata, Formal Languages,
  and Related Topics, University of Szeged, Hungary (2009)  99--115

\bibitem{JiKr10}
Jir\'askov\'a, G., Krausov\'a, M.:
\newblock Complexity in prefix-free regular languages.
\newblock In McQuillan, I., Pighizzini, G., Trost, B., eds.: Proceedings of the
  12th International Workshop on Descriptional Complexity of Formal Systems
  $($DCFS\/$)$, University of Saskatchewan (2010)  236--244

\bibitem{JiOl09}
Jir\'askov\'a, G., Olej\'ar, P.:
\newblock State complexity of union and intersection of binary suffix-free
  languages.
\newblock In Bordihn, H., Freund, R., Holzer, M., Kutrib, M., Otto, F., eds.:
  Proc. of the Workshop on Non-Classical Models for Automata and Applications
  $($NCMA\/$)$, Austrian Computer Society (2009)  151--166

\bibitem{JuKo97}
J\"{u}rgensen, H., Konstantinidis, S.:
\newblock Codes.
\newblock In Rozenberg, G., Salomaa, A., eds.: Handbook of Formal Languages,
  Volume 1: Word, Language, Grammar.
\newblock Springer (1997)  511--607

\bibitem{Lei81}
Leiss, E.:
\newblock Succinct representation of regular languages by boolean automata.
\newblock Theoret. Comput. Sci. \textbf{13} (2009)  323--330

\bibitem{Mas70}
Maslov, A.N.:
\newblock Estimates of the number of states of finite automata.
\newblock Dokl. Akad. Nauk SSSR \textbf{194} (1970)  1266--1268 (Russian).
  English translation: Soviet Math. Dokl. {\bf 11} (1970) 1373--1375.

\bibitem{Mir66}
Mirkin, B.G.:
\newblock On dual automata.
\newblock Kibernetika (Kiev) \textbf{2} (1966)  7--10 (Russian). English
  translation: Cybernetics {\bf 2} (1966) 6--9.

\bibitem{Per90}
Perrin, D.:
\newblock Finite automata.
\newblock In van Leewen, J., ed.: Handbook of Theoretical Computer Science.
  Volume~B.
\newblock Elsevier (1990)  1--57

\bibitem{PiSh02}
Pighizzini, G., Shallit, J.:
\newblock Unary language operations, state complexity and $\text{J}$acobsthal's
  function.
\newblock Internat.\ J.\ Found.\ Comput.\ Sci. \textbf{13} (2002)  145--159

\bibitem{Shy01}
Shyr, H.J.:
\newblock Free Monoids and Languages.
\newblock Hon Min Book Co, Taiwan (2001)

\bibitem{ShTh74}
Shyr, H.J., Thierrin, G.:
\newblock Hypercodes.
\newblock Inform. and Control \textbf{24} (1974)  45--54

\bibitem{Thi73}
Thierrin, G.:
\newblock Convex languages.
\newblock In Nivat, M., ed.: Automata, Languages and Programming.
\newblock North-Holland (1973)  481--492

\bibitem{Seb10}
\v{S}ebej, J.:
\newblock Reversal of regular languages and state complexity.
\newblock In Pardubsk\'a, D., ed.: Proc. 10th ITAT, \v{S}af\'arik University,
  Ko\v{s}ice (2010)  47--54

\bibitem{Yu97}
Yu, S.:
\newblock Regular languages.
\newblock In Rozenberg, G., Salomaa, A., eds.: Handbook of Formal Languages.
  Volume~1.
\newblock Springer (1997)  41--110

\bibitem{Yu01}
Yu, S.:
\newblock State complexity of regular languages.
\newblock J. Autom. Lang. Comb. \textbf{6} (2001)  221--234

\bibitem{YZS94}
Yu, S., Zhuang, Q., Salomaa, K.:
\newblock The state complexities of some basic operations on regular languages.
\newblock Theoret. Comput. Sci. \textbf{125} (1994)  315--328

\end{thebibliography}
